\documentclass[11pt]{article}
\usepackage[margin=1in]{geometry}
\usepackage{appendix}
\usepackage{babel}
\usepackage{amsthm}
\usepackage{graphicx}
\usepackage{amsmath}
\usepackage{amssymb}
\usepackage{array}
\usepackage[colorlinks=true, citecolor=black, linkcolor=black]{hyperref}
\usepackage[linesnumbered,ruled,vlined]{algorithm2e}
\usepackage{algorithmicx}
\usepackage{nicefrac}
\usepackage{xspace}
\usepackage{thm-restate}
\usepackage{paralist}
\usepackage{microtype}
\usepackage{subcaption}
\usepackage[justification=centering]{caption}
\usepackage{bm}
\usepackage{bbm}
\usepackage{graphicx}
\usepackage{multirow}
\usepackage{tabularx}
\usepackage{tikz}
\usepackage{balance}
\usepackage{cleveref}
\usepackage{siunitx}
\usepackage{float}
\usepackage{framed}
\usepackage{alphalph}

\usepackage[export]{adjustbox}

\makeatletter
\newcommand{\multiline}[1]{%
  \begin{tabularx}{\dimexpr\linewidth-\ALG@thistlm}[t]{@{}X@{}}
    #1
  \end{tabularx}
}
\makeatother

\newtheorem{theorem}{Theorem}[section]

\newtheorem{lemma}[theorem]{Lemma}

\newtheorem{fact}{Fact}[section]

\theoremstyle{definition}
\newtheorem{definition}[theorem]{Definition}

\newcommand{\algoname}[1]{\ensuremath{\normalfont\textsc{#1}}\xspace}
\newcommand{\ApproxDD}{\algoname{ApproxDD}}
\newcommand{\ApproxDDb}{\algoname{ApproxDD-ES}}

\newcommand{\PeelN}{\algoname{Select-Large}}
\newcommand{\UniformAlgo}{\algoname{Ugs}}
\newcommand{\StreamAlgo}{\algoname{stream-Ugs}}

\newcommand{\InitDistrib}{\algoname{InitDistrib}}
\newcommand{\SampleGraphlet}{\algoname{SampleGraphlet}}
\newcommand{\gd}{\algoname{Counter}}
\newcommand{\gdr}{\algoname{Rejection}}

\newcommand{\poly}{\operatorname{poly}}
\newcommand{\E}{\mathbb{E}}

\newcommand{\eps}{\epsilon}
\newcommand{\et}{K}
\newcommand{\graphlet}{g}
\newcommand{\ddorder}{\prec}

\newcommand{\bp}{\boldsymbol{p}}
\newcommand{\bmu}{\boldsymbol{\mu}}

\newcommand{\Ind}[1]{\mathbbm{1}_{#1}}

\renewcommand{\Pr}[1]{\mathbf{Pr}\left[{#1}\right]}

\makeatletter
\newcommand{\Rmnum}[1]{\expandafter\@slowromancap\romannumeral #1@}
\makeatother

\newcommand{\ignore}[1]{}
\newcommand{\noindentparagraph}[1]{\noindent\textbf{{#1}.}} 
\newcommand*\samethanks[1][\value{footnote}]{\footnotemark[#1]}

\title{An Efficient Streaming Algorithm \\ for Approximating Graphlet Distributions\thanks{Authors appear in alphabetical order.}}
\author{Marco Bressan\thanks{University of Milan.} \and T-H. Hubert Chan\thanks{The University of Hong Kong.} \and Qipeng Kuang\samethanks \and Mauro Sozio\thanks{LUISS University.}}
\date{}

\begin{document}

\maketitle

\begin{abstract}
In recent years, the problem of computing the frequencies of the induced $k$-vertex subgraphs of a graph, or \emph{$k$-graphlets}, has become central.
One approach for this problem is to sample $k$-graphlets randomly. 
Classic algorithms for $k$-graphlet sampling require loading the entire graph into main memory, making them impractical for massive graphs.
To bypass this limitation, Bourreau et al.\ (NeurIPS 2024) introduced a \emph{streaming} algorithm that through nontrivial techniques makes only $O(\log n)$ passes using $O(n \log n)$ memory.
In this work we break their $O(\log n)$-pass bound by giving an algorithm that, for any fixed $c>0$, makes $O(1/c)$ passes using $\tilde O(n^{1+c})$ memory.
As a consequence of their lower bound, our algorithm is optimal up to a factor of $\tilde{O}(n^c)$ in the memory usage.
We use this sampling algorithm to obtain an efficient method of approximating $k$-graphlet distributions. 
Experiments on real-world and synthetic graphs show that our algorithm is always at least as good as the one of Bourreau et al., and outperforms it by orders of magnitude on mildly dense graphs.
\end{abstract}


\section{Introduction}
A \emph{$k$-graphlet} of a simple graph $G$ is a connected, induced subgraph formed by exactly $k$ vertices. Two $k$-vertex subsets are \emph{equivalent} if their induced subgraphs are isomorphic. The so-called \emph{$k$-graphlet distribution} of $G$
is the frequency vector that provides the relative size of each such equivalence class.
It is known that the $k$-graphlet distribution captures important structural information~\cite{Milo824} that is relevant in many real-world applications, such as graph classification~\cite{Tu2019-gl2vec}, graph kernels~\cite{Shervashidze2009}, graph neural networks~\cite{Peng2020-GNN}, and federated learning \cite{DBLP:journals/tit/GhoshCYR22}.
In this work, we study the problem of approximating the $k$-graphlet distribution when the input graph cannot be stored in memory.


A first major distinction must be made between \emph{exact} and \emph{approximate} algorithms.
A trivial exact algorithm is the one that enumerates and checks all $k$-vertex subsets of the input graph.
This algorithm is obviously inefficient, and there is strong evidence that one cannot perform substantially better (see classic lower bounds from parameterized complexity, such as~\cite{Chen&2006} or~\cite{Jerrum&2015}).
Our focus is instead on approximate algorithms, which (roughly speaking) compute the $k$-graphlet distribution within some small additive error $\alpha>0$.
It is clear that this problem can be reduced, in the obvious way, to sampling $k$-graphlets randomly from the underlying graph.
%
%
%
%
Many algorithms for sampling $k$-graphlets have been proposed, broadly divided in those based on random walks \cite{Agostini19mixing,Bhuiyan&2012,Chen&2016,Han&2016,Saha&2015MCMC,Wang&2014} and those based on color coding \cite{Bressan&2017,Bressan&2018b,Bressan2019-VLDB,Bressan21TKDD}.
Eventually, \cite{Bressan21STOC,B23} described \UniformAlgo, an algorithm that preprocesses the graph with $n$ vertices and $m$ edges in time $O(n k^2 \log k + m)$ and then produces independent uniform $k$-graphlets in $k^{O(k)} \log n$ time per sample, yielding an efficient algorithm to approximate the $k$-graphlet distribution.

Although \UniformAlgo\ allows one to sample $k$-graphlets efficiently, it does so only in the standard RAM model, where the input must fit entirely in memory.
%
%
Unfortunately, in practice the graph is often too large for the machine's memory; this is commonplace for real-world graphs, such as those representing social networks, the Web, the human brain and many others.
In these cases, \UniformAlgo\ simply does not work, and the same holds for the other algorithms listed above.

To bypass these limitations,~\cite{Bourreau2024} has recently initiated the study of the $k$-graphlet sampling problem in the \emph{streaming model}.
In this model, the input edges can only be accessed sequentially in an arbitrary order, while using an amount of memory sublinear in the size of the input graph. Streaming algorithms process the graph over multiple \emph{passes}, with each pass involving the processing of all edges in the input graph. Studying the $k$-graphlet distribution problem in the streaming model avoids the obstacle of storing the graph in memory, and turns the optimization target to the number of passes versus the memory usage.


The bottleneck of the algorithm provided by~\cite{Bourreau2024} is in the preprocessing phase, called \ApproxDD, which computes an approximate \emph{degree-dominating} order of vertices---a vertex ordering where each vertex has approximately the highest degree in the subgraph induced by its suffix.
This preprocessing phase uses $O(n \log n)$ memory and $O(\log n)$ passes or, alternatively, $O(n)$ memory and $O(\log n \cdot \log \log n)$ passes.
(This regime, where the memory allowed is linear or near-linear in $n$, is called the \emph{semi-streaming model}).
Unfortunately, \ApproxDD\ is still inefficient for at least two reasons.
The first reason is that it requires $O(\log n)$ passes. It is not clear that this is optimal; the very lower bounds from~\cite{Bourreau2024} only rule out $O(1)$ passes with $o(n)$ memory, but not, say, $O(1)$ or $O(\log \log n)$ passes with $O(n)$ or $O(n \log n)$ memory.
The second reason is that, as shown by our experiments, \ApproxDD\ basically breaks down on mildly dense graphs: it takes thousands of passes, or even fails to run, even with a memory of (say) $10\%$ the size of the input graph.
In this work we aim at bypassing these limitations, by seeking an algorithm that runs in $O(1)$ passes and guarantees good performance regardless of the density of the graph.

\subsection{Contributions}

We develop a streaming algorithm to approximate the $k$-graphlet distribution, with the following guarantees:

\begin{theorem}
\label{thm:1}
There is a two-phase semi-streaming algorithm with the following guarantees. Given $c>0$, a simple graph $G$, a positive integer $k \ge 3$, an error parameter $\alpha \in (0,1)$ and a failure probability $\delta \in (0,1)$, the algorithm returns a distribution $\hat \bp$ such that:
\begin{itemize}
  \item The first phase (``preprocessing'') uses $O\left(1/c\right)$ passes and $\tilde O\left(n^{1+c}\right)$ bits of memory\footnote{The standard $\tilde O(\cdot)$ notation hides $\poly\log(\cdot)$ factors.
  }, with probability at least $1-\frac{1}{n}$.
  \item The second phase (``sampling'') always uses $O\left(\frac{k^{O(k)}}{\alpha^4 M} \ln \frac{1}{\delta}\right)$ passes and $M$ bits of memory.
  \item The $L_{\infty}$ distance between $\hat \bp$ and the $k$-graphlet distribution of $G$ is at most $\alpha$, with probability at least $1-\delta$.
\end{itemize}
\end{theorem}

Our algorithm follows the \StreamAlgo framework introduced by Bourreau et al.~\cite{Bourreau2024}:
\begin{framed}
\textbf{Preprocessing Phase}
\begin{enumerate}[(a)]
  \item Compute an \emph{approximate degree-dominating order} (\emph{DD order}, see \Cref{def:ddorder}), a special total order of the vertices that enables fast uniform graphlet sampling~\cite{Bressan21STOC,B23}; \label{step:compute-dd-order} 
  \item Compute the initial probabilities;
\end{enumerate}

\bigskip

\textbf{Sampling Phase}
\begin{enumerate}[(a)]
  \setcounter{enumi}{2}
  \item Repeatedly sample $k$-graphlets. 
  \item Either (i) execute probabilistic rejection (in \cite{Bourreau2024}), or (ii) obtain the $k$-graphlet distribution (in this work). \label{step:rejection-or-gd}
\end{enumerate}
\end{framed}

On the theoretical side, the main contribution behind \Cref{thm:1} is a new streaming algorithm for computing a DD order in Step (\ref{step:compute-dd-order}).
Unlike the algorithm of~\cite{Bourreau2024}, which intrinsically requires $O(\log n)$ passes because of its ``halving degree'' strategy, our algorithm requires only $O(1/c)$ passes\footnote{Note that the guarantees of \cite{Bourreau2024} are probabilistic for number of passes and deterministic for the memory, while our guarantees are probabilistic for both.}, at the cost of higher memory usage.
This is thanks to a strategy that, speaking broadly, consists in sampling enough edges from the input graph so as to make, in a single pass, much more progress than the algorithm of~\cite{Bourreau2024}.
We remark that, by a lower bound of~\cite{Bourreau2024}, no algorithm using  $O(1)$ passes and $o(n)$ memory can even decide if the input graph contains some $k$-graphlet, for a fixed $k \ge 2$.
Since computing the graphlet distribution implies solving such a decision problem, the memory used by our graphlet distribution algorithm is optimal up to a factor of $\tilde O(n^c)$.

Our algorithm also differs from the corresponding algorithm in \cite{Bourreau2024} in Step (\ref{step:rejection-or-gd}): in particular, our work employs the Horvitz-Thompson-based algorithm \cite{horvitzthompsonestimator} to approximate the $k$-graphlet distribution within $L_{\infty}$ distance $\alpha$, whereas \cite{Bourreau2024} focuses on uniform $k$-graphlet sampling using probabilistic rejection.

On the practical side, although our memory usage goes with $n^{1+c}$ rather than just $n$, it turns out that picking $c$ small enough makes our algorithm very competitive compared to~\cite{Bourreau2024}.
In our experiments, by choosing $c=0.1$, our algorithm makes around $20$ preprocessing passes on most datasets, and yet its memory usage is always within one and a half times that of~\cite{Bourreau2024}.
In fact, on datasets that are only \emph{mildly} dense, our algorithm drastically outperforms the one of~\cite{Bourreau2024}.
For example, on a dataset with $n \approx 2 \times 10^5$ vertices and average degree $\approx 500 = 0.0025 n$, both algorithms use around 20MB of memory, but ours makes less than $20$ passes and the one of~\cite{Bourreau2024} more than $3000$.
Our algorithm thus performs well both in theory and in practice, and is scalable and predictable regardless of the structure of the input graph.

\subsection{Related Work}

Counting and sampling triangles can be considered as a special case of computing the $k$-graphlet distribution as $k=3$. Most streaming algorithms for this task in the literature, such as \cite{baryossef&2002-streaming,becchetti_efficient_2008,pavan_counting_2013,ahmed_graph_2014,jha&2015-stream,DeStefani&2017b,lim_memory-efficient_2018}, do not generalize or give guarantees for $k>3$.
Streaming algorithms for counting or sampling $k$-graphlets for a generic $k$ exist, such as \cite{kane2012-stream} and \cite{de_stefani_tiered_2017}.
Notice that their streaming setting is different from ours, as they consider only \emph{single}-pass algorithms, they ask to maintain a good approximation throughout the entire stream, and they have stricter memory requirements.
The downside is that both works give guarantees only for specific graphlets.

Many algorithms exist for counting and sampling $k$-graphlets, or estimating the $k$-graphlet distribution, in the standard computation model \cite{Agostini19mixing,Bhuiyan&2012,Bressan&2017,Bressan&2018b,Bressan2019-VLDB,Bressan21TKDD,Chen&2016,Han&2016,Gionis&2020mixing,Paramonov19-lifting,Saha&2015MCMC,Wang&2014}.
All those algorithms have poor memory consumption, as they need to store the whole graph in memory.
Among those algorithms, one that is still relatively efficient is \textsc{Motivo}~\cite{Bressan21TKDD}, which is based on the color-coding technique.
\textsc{Motivo} has formal approximation guarantees, and in practice it can approximate well the $k$-graphlet distribution on large graphs in a matter of minutes or hours (depending on $k$) when running on a workstation.
In this work we use it to compute the ground-truth $k$-graphlet distributions for our experiments.

Topological orders play an important role in subgraph mining problems. For example, the degeneracy (or core) order, one of the most commonly used orders, enables efficient subgraph-counting algorithms~\cite{Matula83-Ordering,Chiba&1985,Bressan21Algo,BeraGLSS22}. There are also streaming algorithms involving computing topological orders \cite{bahmani_densest_2012,sariyuce_streaming_2013,Bourreau2024}. We remark that computing the core order is technically different from computing the \emph{degree-dominating order} that is used in this work.
Indeed, while a core order is obtained by repeatedly removing a vertex of minimum degree, a degree-dominating order is obtained by removing one of maximum degree, and it is easy to find examples where the two orders are not obviously related in any way.
Therefore, existing streaming algorithms for core orders \cite{bahmani_densest_2012,sariyuce_streaming_2013} cannot be employed here.

We face a technical challenge similar to the one of \cite{k-core-mpc}, where the authors approximate the $k$-core decomposition in parallel: keeping accurate estimates of vertex degrees while iteratively removing vertices in a large graph.
Maintaining the degrees by scanning the entire edge list after each removal would require too many passes.
Their solution is to work on a sub-sampled graph, produced by sampling each edge with an appropriate probability $p$.
On the one hand, by standard concentration results, the degrees in the sub-sampled graph provide a good estimate of the true degree.
On the other hand, only a small fraction of edges is sampled, thus one may fit the memory requirements of the streaming algorithm.

We use a Horvitz–Thompson estimator \cite{horvitzthompsonestimator} to estimate the distribution of a set of objects when only a non-uniform sampling of the objects can be accessed. This estimator helps us approximate the $k$-graphlet distribution without sampling $k$-graphlets uniformly at random; see also~\cite{Bressan21TKDD}.

\subsection{Organization}

We define notation and key notions in \Cref{sec:prelim}. The new algorithm for computing the DD order is presented in \Cref{sec:fast-ddorder}. The algorithm for estimating the $k$-graphlet distribution is then completed in \Cref{sec:graphlet-distribution}. Finally, the experiments to evaluate our algorithm and compare with existing work are shown in \Cref{sec:experiments}. 

\section{Preliminaries}\label{sec:prelim}
This section introduces basic notation and the main technical concepts.

\subsection{Notation and Computational Model}
Let $G=(V,E)$ be a simple graph and let $n=|V|$.
For a subset $U \subseteq V$ of vertices, $G[U]$ denotes the subgraph of $G$ induced by $U$, and $G-U$ denotes the subgraph $G[V \setminus U]$.
For a total order $\prec$ over $V$, we let $G(v)$ be the subgraph $G[\{u: v \preceq u\}]$.
For $u \in V$, we denote by $d_u$ the degree of $u$ in $G$; if $H$ is a subgraph of $G$ and $u \in V(H)$, then we denote by $d(u|H)$ the degree of $u$ in $H$.
This applies in particular to $H=G(v)$, so we may write $d(u|G(v))$.
We denote by $\Delta$ the maximum degree of $G$.

The smallest value for which the problem is non-trivial is
$k \ge 3$.
A \emph{$k$-graphlet} of $G$ is an induced, connected, $k$-vertex subgraph of $G$.
We denote by $\graphlet_1, \cdots, \graphlet_{m_k}$, the representatives of the isomorphism classes of connected graphs on $k$ vertices, in any order.
Thus, every $k$-graphlet of $G$ is isomorphic to precisely one $\graphlet_i$.
For each $i=1,\ldots,m_k$ let $l_i$ be the number of induced subgraphs of $G$ isomorphic to $\graphlet_i$, and let $L=\sum_{i=1}^{m_k} l_i$.
The \emph{$k$-graphlet distribution} of $G$ is the vector
\begin{equation*}
\bmu_k \triangleq (\mu_1, \cdots, \mu_{m_k})
\end{equation*}
where $\mu_i = \frac{l_i}{L}$ for each $i$.
Our goal is to estimate $\bmu_k$. An example of the $k$-graphlet distribution is presented in \Cref{fig:graphletdistrib}.

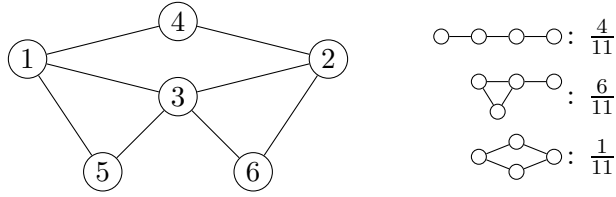
\begin{figure}
  \centering
  \begin{tikzpicture}
    \tikzstyle{vertex}=[circle,draw, inner sep=2pt]
    \tikzstyle{anonvertex}=[circle,draw, inner sep=2pt]
    \node[vertex] (v1) at (0,1.5) {$1$};
    \node[vertex] (v2) at (4,1.5) {$2$};
    \node[vertex] (v3) at (2,1) {$3$};
    \node[vertex] (v4) at (2,2) {$4$};
    \node[vertex] (v5) at (1,0) {$5$};
    \node[vertex] (v6) at (3,0) {$6$};

    \path[-] (v1) edge (v3);
    \path[-] (v1) edge (v4);
    \path[-] (v1) edge (v5);
    \path[-] (v2) edge (v3);
    \path[-] (v2) edge (v4);
    \path[-] (v2) edge (v6);
    \path[-] (v3) edge (v5);
    \path[-] (v3) edge (v6);

    \node[anonvertex] (a1) at (5.5,1.8) {};
    \node[anonvertex] (a2) at (6,1.8) {};
    \node[anonvertex] (a3) at (6.5,1.8) {};
    \node[anonvertex] (a4) at (7,1.8) {};
    \node (a) at (7.5,1.8) {: $\frac{4}{11}$};
    \path[-] (a1) edge (a2);
    \path[-] (a2) edge (a3);
    \path[-] (a3) edge (a4);
    \node[anonvertex] (b1) at (6,1.2) {};
    \node[anonvertex] (b2) at (6.5,1.2) {};
    \node[anonvertex] (b3) at (6.25,0.8) {};
    \node[anonvertex] (b4) at (7,1.2) {};
    \node (b) at (7.5,1) {: $\frac{6}{11}$};
    \path[-] (b1) edge (b2);
    \path[-] (b2) edge (b3);
    \path[-] (b1) edge (b3);
    \path[-] (b2) edge (b4);
    \node[anonvertex] (c1) at (6,0.2) {};
    \node[anonvertex] (c2) at (6.5,0.4) {};
    \node[anonvertex] (c3) at (6.5,0) {};
    \node[anonvertex] (c4) at (7,0.2) {};
    \node (b) at (7.5,0.2) {: $\frac{1}{11}$};
    \path[-] (c1) edge (c2);
    \path[-] (c1) edge (c3);
    \path[-] (c3) edge (c4);
    \path[-] (c2) edge (c4);
  \end{tikzpicture}
  \caption{A toy graph (left) and its $4$-graphlet distribution (right).
  $4$-graphlets with probability $0$ are not shown.}\label{fig:graphletdistrib}
\end{figure}

A notion central to this approach is the \emph{degree-dominating} ordering (DD-ordering) introduced in~\cite{bressan2021efficient},
as well as its approximate version from~\cite{Bourreau2024}.

\begin{definition}[(Approximate) DD-ordering]\label{def:ddorder}
Given a simple graph $G=(V,E)$,
a total order $\prec$ over $V$ is degree-dominating if, for each $v \in V$, the maximum degree in $G(v)$ is achieved by $v$ itself; in other words, $d(v|G(v)) \ge d(u|G(v))$ for all $u \succeq v$.
For $\vartheta > 0$,
a \emph{$\vartheta$-DD order} $\prec$ is a total order over $V$ such that, for all $v \in V$ with $d(v|G(v)) \ge 1$ and all $u \succeq v$, it holds that $d(v|G(v)) \ge \vartheta \cdot  d(u|G(v))$.

\end{definition}
We often use $\vartheta = \frac{1}{1+\epsilon}$ for some $\epsilon>0$.
An example of DD-ordering is shown in \Cref{fig:ddorder}. A possible $\frac{1}{1.5}$-DD order of the graph is $1 \prec 3 \prec 2 \prec 5 \prec 6 \prec 4$. Note that for the first position of the order, we have $d(1|G(1))=3$ and $\max_{u \succeq 1} d(u|G(1)) = d(3|G(1)) = 4$, which satisfies $d(1|G(1)) \ge \frac{1}{1.5} \cdot \max_{u \succeq 1} d(u|G(1))$.

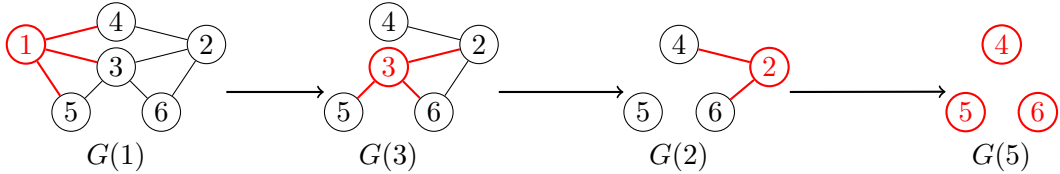
\begin{figure}
  \centering
  \begin{subfigure}{0.23\textwidth}
    \centering
    \begin{tikzpicture}[scale=0.6, remember picture]
      \tikzstyle{vertex}=[circle,draw, inner sep=2pt]
      \node[red,thick,vertex] (v1) at (0,1.5) {$1$};
      \node[vertex] (v2) at (4,1.5) {$2$};
      \node[vertex] (v3) at (2,1) {$3$};
      \node[vertex] (v4) at (2,2) {$4$};
      \node[vertex] (v5) at (1,0) {$5$};
      \node[vertex] (v6) at (3,0) {$6$};

      \path[red,thick,-] (v1) edge (v3);
      \path[red,thick,-] (v1) edge (v4);
      \path[red,thick,-] (v1) edge (v5);
      \path[-] (v2) edge (v3);
      \path[-] (v2) edge (v4);
      \path[-] (v2) edge (v6);
      \path[-] (v3) edge (v5);
      \path[-] (v3) edge (v6);

      \node at (2,-1) {$G(1)$};

      \coordinate (r1) at (current bounding box.east);
    \end{tikzpicture}
  \end{subfigure}
  \begin{subfigure}{0.23\textwidth}
    \centering
    \begin{tikzpicture}[scale=0.6, remember picture]
      \tikzstyle{vertex}=[circle,draw, inner sep=2pt]
      \node[vertex] (v2) at (4,1.5) {$2$};
      \node[red,thick,vertex] (v3) at (2,1) {$3$};
      \node[vertex] (v4) at (2,2) {$4$};
      \node[vertex] (v5) at (1,0) {$5$};
      \node[vertex] (v6) at (3,0) {$6$};

      \path[red,thick,-] (v2) edge (v3);
      \path[-] (v2) edge (v4);
      \path[-] (v2) edge (v6);
      \path[red,thick,-] (v3) edge (v5);
      \path[red,thick,-] (v3) edge (v6);

      \node at (2,-1) {$G(3)$};

      \coordinate (l2) at (current bounding box.west);
      \coordinate (r2) at (current bounding box.east);
    \end{tikzpicture}
  \end{subfigure}
  \begin{subfigure}{0.23\textwidth}
    \centering
    \begin{tikzpicture}[scale=0.6, remember picture]
      \tikzstyle{vertex}=[circle,draw, inner sep=2pt]
      \node[red,thick,vertex] (v2) at (4,1) {$2$};
      \node[vertex] (v4) at (2,1.5) {$4$};
      \node[vertex] (v5) at (1.2,0) {$5$};
      \node[vertex] (v6) at (2.8,0) {$6$};

      \path[red,thick,-] (v2) edge (v4);
      \path[red,thick,-] (v2) edge (v6);

      \node at (2,-1) {$G(2)$};

      \coordinate (l3) at (current bounding box.west);
      \coordinate (r3) at (current bounding box.east);
    \end{tikzpicture}
  \end{subfigure}
  \begin{subfigure}{0.23\textwidth}
    \centering
    \begin{tikzpicture}[scale=0.6, remember picture]
      \tikzstyle{vertex}=[circle,draw, inner sep=2pt]
      \node[red,thick,vertex] (v4) at (2,1.5) {$4$};
      \node[red,thick,vertex] (v5) at (1.2,0) {$5$};
      \node[red,thick,vertex] (v6) at (2.8,0) {$6$};

      \node at (2,-1) {$G(5)$};

      \coordinate (l4) at (current bounding box.west);
      \coordinate (r4) at (current bounding box.east);
    \end{tikzpicture}
  \end{subfigure}

  \begin{tikzpicture}[remember picture, overlay]
    \draw[->, thick] (r1) -- (l2);
    \draw[->, thick] (r2) -- ([yshift=1.5mm]l3);
    \draw[->, thick] ([yshift=1.5mm]r3) -- ([yshift=1.5mm]l4);
  \end{tikzpicture}
  \caption{Computation of a $\frac{1}{1.5}$-DD ordering of a toy graph, by iterative removal of a high-degree vertex. The resulting $\frac{1}{1.5}$-DD ordering is $1 \prec 3 \prec 2 \prec 5 \prec 6 \prec 4$.}
  \label{fig:ddorder}
\end{figure}

We consider the \emph{semi-streaming} model of \cite{feigenbaum_graph_2005}.
In this model, the memory consists of $M=\tilde O(n^{1+c})$ words of $O(\log n)$ bits for some constant $c>0$ (for example, $c=0.1$).
The algorithm can access the graph by scanning its edges, sequentially, in what is called a \emph{pass}.
The order in which edges appear in each pass is arbitrary (that is, decided by an adversary), and the algorithm can make any number of passes.
We assume $m=\Omega(M)$ to avoid trivialities. All computation will be in the standard RAM model.

\subsection{Uniform Graphlet Sampling}\label{sec:ugs}
Our work builds on \StreamAlgo, the streaming algorithm of Bourreau et al.~\cite{Bourreau2024}.
\StreamAlgo\ consists of a preprocessing phase and a sampling phase.
The preprocessing phase is run once, and computes a $\vartheta$-DD order $\ddorder$,
where $\vartheta=\frac{1}{1+\epsilon}$ for some desired $\epsilon>0$, together with an initial distribution $\bp$ over $V$.
With $\ddorder$ and $\bp$ one can then implement a $k$-pass sampling routine ensuring that any given $k$-graphlet of $G$ has the same sampling probability except for a multiplicative factor of $((1+\epsilon)k)^{O(k)}$.
Using rejection sampling one can remove that factor and get truly uniform sampling, with an overhead of $((1+\epsilon)k)^{O(k)}$ expected trials per sample.
The crucial part is therefore computing $\ddorder$ and $\bp$ with as few passes as possible.

\smallskip
\noindentparagraph{Computing the $\vartheta$-DD order}
As shown in \cite[Theorem 3.1]{Bourreau2024}, one can compute $\ddorder$ with high probability in $O\left(\frac{\log n}{\epsilon^2} \log \frac{\log n}{\epsilon}\right)$ passes by using memory $O(n)$, or in $O\left(\frac{\log n}{\epsilon^2}\right)$ passes by using memory $O\left(\frac{n \log n}{\epsilon}\right)$.
The idea is to select a subset $S$ of vertices that have degree at least $\frac{\Delta}{1+\epsilon}$, where $\Delta$ is the maximum degree of $G$.
This can be done with high probability with a constant number of passes.
The set $S$ is then appended to $\ddorder$ in arbitrary order, and $S$ is removed from $G$ (meaning that \emph{the algorithm} keeps track, in some way, of the subgraph obtained by removing $S$).
Therefore, every $O(1)$ passes the maximum degree of $G$ decreases by a factor $\frac{1}{1+\epsilon}$.
As a consequence, the algorithm terminates within $\log_{\frac{1}{1+\epsilon}} n$ passes.



\smallskip
\noindentparagraph{Initial Distribution}
Given a $\frac{1}{1+\epsilon}$-DD order, let $N_v$ be the number of $k$-graphlets in $G(v)$ containing $v$. Computing $d(v|G(v))$ and deciding whether $N_v>0$ for all vertices in $V$ can be done in 1 pass and $O(kn)$ memory by \cite[Lemma 3.3]{Bourreau2024}. After that, the \emph{normalization parameter}
\begin{equation*}
  Z = \sum_{v \in V, \ N_v>0} d(v|G(v))^{k-1}
\end{equation*}
as well as the \emph{initial distribution} $\bp = \{p(v)|v \in V\}$ such that
\begin{equation*}
  p(v) = \begin{cases}
    \frac{d(v|G(v))^{k-1}}{Z}, & N_v > 0, \\
    0, & N_v = 0
  \end{cases}
\end{equation*}
can be computed without any pass. This subroutine is formally described in \Cref{alg:InitDistrib}. The initial distribution will be used later to select the starting vertex and compute the relevant sampling probability in the sampling phase.

\begin{algorithm}
	\caption{Initial Distribution \cite[Algorithm 2]{Bourreau2024}}
	\label{alg:InitDistrib}
	\SetKwProg{Function}{Function}{}{}
	\Function{\InitDistrib($G, \prec$)}{
        Compute $d(v|G(v))$ and decide whether $N_v>0$ for each $v \in V$ \\
        $Z \gets \sum_{v \in V, \ N_v>0} d(v|G(v))^{k-1}$ \\
        $p(v) \gets \begin{cases}
            \frac{d(v|G(v))^{k-1}}{Z}, & N_v > 0, \\
            0, & N_v = 0
        \end{cases}$ \\
        \Return $\bp = \{p(v)|v \in V\}$
	}
\end{algorithm}

\noindentparagraph{Rejection Sampling}
At a high level, we first sample a $k$-graphlet of $G$ and then accept it with some probability so that the probability of each $k$-graphlet being sampled and accepted is the same. The first sampling step goes as follows. We sample a vertex $v$ according to the initial distribution $\bp = \{p(v) | v \in V\}$ and let $S = \{v\}$. Then we ``grow'' $S$ for $k-1$ times, each time by selecting uniformly at random an edge in $G(v)$ with exactly one endpoint in $S$, and adding the other endpoint to $S$. In this way, we obtain a $k$-graphlet $S$ with probability $p(S) = p(v) \cdot p(S | v)$, where $p(S|v)$ is the probability of obtaining $S$ from $v$. All those probabilities can be computed efficiently.
This sampling algorithm is formally described in \Cref{alg:SampleGraphlet}.

\begin{algorithm}
	\caption{Sampling A Graphlet \cite[Algorithm 3]{Bourreau2024}}
	\label{alg:SampleGraphlet}
	\SetKwProg{Function}{Function}{}{}
	\Function{\SampleGraphlet($G, \prec, k, \bp$)}{
        Sample $v$ from the distribution $\bp$ \\
        $S \gets \{v\}$ \\
        \For{$i \gets 1$ \KwTo $k-1$}{
            Uniformly sample an edge $(x,y)$ in $G(v)$ such that $x \in S$ and $y \not\in S$ \\
            $S \gets S \cup \{y\}$
        }
        \Return $S$
	}
\end{algorithm}

To achieve uniform sampling, let $\Gamma$ be some constant such that $0 < \Gamma \le \min_{S} p(S)$, where $S$ ranges over all $k$-tuples of vertices such that $p(S) > 0$.
We then accept $S$ with probability $\frac{\Gamma}{p(S)}$, otherwise reject $S$ and sample again.
By \cite[Lemma C.1]{Bourreau2024}, with high constant probability (say $0.99$), after $(k(1+\epsilon))^{O(k)}$ trials of sampling we accept a $k$-graphlet, and that $k$-graphlet is drawn from the uniform distribution.

\noindentparagraph{Parallel sampling}
As \Cref{alg:SampleGraphlet} uses memory $O(k^2)$, by using memory $M$ one can run $\Theta \left(\frac{M}{k^2}\right)$ instances  in parallel, and produce with high probability $\Theta\left(\frac{M}{(k(1+\epsilon))^{O(k)}}\right)$ independent uniformly random $k$-graphlets.
Moreover, \cite[Theorem 3.4]{Bourreau2024} gives a trade-off between running time and number of passes: the parallel sampling can be achieved either in $k$ passes and $O(M |E|)$ time, or in $2k-1$ passes and $O(M2^k + k|E| \log n)$ time.

\ignore{
\noindent\textbf{Computational model.} The algorithm has a memory of $M$ words of $\Theta(\log n)$ bits each. We require $M=\Omega(n)$ or $M=\Omega(n \log n)$, depending on the case. The graph is stored as a list of edges in arbitrary (i.e., adversarial) order. With one \emph{pass}, the algorithm can read the list sequentially. This is called \emph{semi-streaming} model~\cite{feigenbaum_graph_2005}.
To avoid trivialities we assume $m=\Omega(M)$; otherwise, in one pass one can store $G$ and run the algorithm of~\cite{B23}. For computation we assume the standard RAM model. We also assume that in time $O(1)$ one can draw a random uniform integer in $\{1,\ldots,c\}$ for any $c = \poly(n)$, or a Bernoulli random variable $B(p)$ for $p = \Omega(n^{-k})$. When we run multiple instances of a subroutine in parallel, the running time is understood to be the total number of operations executed by all those instances.

\noindent \textbf{Graphlet size.}
Our algorithms are designed primarily for $k = O(1)$, but they yield nontrivial guarantees also for $k=\omega(1)$; for instance, for $k=\sqrt{\log n}$. Our full statements make the dependence on $k$ clear.
}

\subsection{Concentration Inequalities}

Let $B(m,p)$ denote the binomial distribution with parameters $m,p$.
We use the following Chernoff-style bounds in our analysis.

\begin{fact}[Chernoff Bound Variant \Rmnum{1}]
\label{fact:Chernoff}
Let $X \sim B(m, p)$ and $0 < \epsilon < 1$. Then, we have the following:
\begin{compactitem}
\item If $m \geq N$, $\Pr{X \leq (1 - \epsilon) N p}
\leq \exp\left(- \frac{1}{3} \epsilon^2 N p\right)$.

\item If $m \leq N$, $\Pr{X \geq (1 + \epsilon) N p}
\leq \exp\left(- \frac{1}{3} \epsilon^2 N p\right)$.
\end{compactitem}
\end{fact}

Fact~\ref{fact:Chernoff} follows from stochastic dominance together with another variant of the Chernoff Bound.

\begin{lemma}[Stochastic Dominance]
\label{lemma:coupling}
Let $n \le m$, let $X \sim B(n,p)$ and $Y \sim B(m,p)$. For any $c$, $\Pr{X > c} \le \Pr{Y > c}$.
\end{lemma}

\begin{proof}
Let $x_1, \cdots, x_m$ be Bernoulli variables with positive probability $p$. Let $A = \sum_{i=1}^n x_i$, $B = \sum_{i=n+1}^m x_i$, and $C = A + B$. Thus $A$ and $X$ have the same distribution, and $C$ and $Y$ have the same distribution.
Now observe that, for any $c$,
\begin{equation*}
\Pr{X > c} = \Pr{A > c} \le \Pr{C > c} = \Pr{Y > c}. \qedhere
\end{equation*}
\end{proof}

\begin{fact}[Chernoff Bound Variant \Rmnum{2}]
\label{fact:Chernoff2}
Let $D$ be the distribution given by $\Pr{X=r_i} = p_i$ for $i = 1, \cdots, m$ where $\sum_{i=1}^m p_i = 1$. A variable following $D$ has expectation $E = \sum_{i=1}^m p_ir_i$. Let $X_1, \cdots, X_T$ be i.i.d. variables following $D$. For any $0 < \epsilon < 1$, we have the following inequality:
\begin{equation*}
  \Pr{\left| \frac{\sum_{i=1}^T X_i}{T} - E \right| \ge \epsilon E } \le 2\exp\left( -\frac{\epsilon^2TE}{3\max_{i=1}^m r_i} \right).
\end{equation*}
\end{fact}

We also use Hoeffding's inequality.

\begin{fact}[Hoeffding's inequality]
\label{fact:Hoeffding}
Let $D$ be the distribution given by $\Pr{X=r_i} = p_i$ for $i = 1, \cdots, m$ where $\sum_{i=1}^m p_i = 1$. A variable following $D$ has expectation $E = \sum_{i=1}^m p_ir_i$. Let $X_1, \cdots, X_T$ be i.i.d. variables following $D$. For any $0 < \alpha < 1$, we have the following inequality:
\begin{equation*}
  \Pr{\left| \frac{\sum_{i=1}^T X_i}{T} - E \right| \ge \alpha } \le 2\exp\left( -\frac{2T\alpha^2}{\left(\max_{i=1}^m r_i\right)^2} \right).
\end{equation*}
\end{fact} 

\section{Improved Algorithms to Compute the Degree-Dominating Order}
\label{sec:fast-ddorder}

As recalled above, the sampling phase of the algorithm of \cite{Bourreau2024} requires just $k$ or $2k-1$ passes, depending on the implementation.
The preprocessing phase is the one that dominates the streaming complexity, with $O(\log n)$ passes or $O(\log n \log \log n)$ passes depending on the desired memory usage.
The idea of \cite{Bourreau2024}, described in \Cref{sec:ugs}, is to use every pass to detect and remove the subset of vertices with degree at least $\frac{\Delta}{1+\epsilon}$, where $\Delta$ is the current maximum degree of the graph.
This leads the maximum degree of $G$ to decrease by a multiplicative factor of $(1+\epsilon)$ at each pass, making the algorithm terminate within $O(\log_{1+\epsilon} n)$ passes, and at the same time the removal sequence yields precisely a $\frac{1}{1+\eps}$-DD ordering.

Our key challenge is to break through the $O(\log n)$ passes bound.
In order to do that we must, at each pass, make more progress than just removing the vertices with degree at least $\frac{\Delta}{1+\epsilon}$.
Note that, whenever \emph{any} vertex is removed, the degrees of the remaining vertices can change.
Therefore there is a delicate tradeoff between (a) how many vertices one removes, and (b) how good the current degree estimates still are.
In the approach of \cite{Bourreau2024} described above, this tradeoff still works, because after each vertex-removal pass the degrees are recomputed exactly using one additional pass.
In our case, however, we can make such a pass only $O(1/c)$ times.


Using similar ideas to~\cite{k-core-mpc}, we modify the algorithm of \cite{Bourreau2024} by computing, in a single pass, several approximations of the input graph, each one obtained by sub-sampling edges with a different probability.
We then employ these sub-sampled graphs to provide the degree estimates required by several subsequent vertex-removal passes of \cite{Bourreau2024}.
In this way we are essentially ``accelerating'' the algorithm of \cite{Bourreau2024} by performing several passes at once.
%

For the sake of readability, we start by introducing in \Cref{sec:fast-ddorder-1}  a ``warm-up'' algorithm which still uses $O\left(\frac{\log n}{\epsilon}\right)$ passes and $O\left(\frac{n}{\epsilon^2} \log \frac{n}{\epsilon \delta}\right)$ memory.
Then, in \Cref{sec:fast-ddorder-2}, we show how for any desired $c>0$ we can bring the passes down to $O\left(\frac{1}{c}\right)$ passes at the price of using $O\left(\frac{n^{1+c}}{\epsilon^3} \log \frac{n}{\epsilon \delta}\right)$ memory.

\subsection{A Simple Approach Using $O(\log n)$ Passes}
\label{sec:fast-ddorder-1}

The idea of the algorithm is as follows. We sample a directed graph $\tilde G$ by adding each direction of each edge of $G$ into $\tilde G$ with some probability $p$. By picking an appropriate $p$, we have that: (1) the number of edges in $\tilde G$ is small with high probability; and (2) the out-going degree of each vertex is a good estimate of its degree. As we sample the two directions of the original edge independently, the estimations of vertex degrees are independent, even if vertices are removed sequentially.



We state a warm-up algorithm in Algorithm \ref{alg:approxdd-fewer-1}. Let $U$ be the set of vertices of $G$, $\epsilon$ be the parameter of the DD order, and $\ddorder$ be an empty list. Define a probability $p$ and we sample $\tilde G$ from $G[U]$ as described above using one pass such that with high probability, $\tilde G$ contains $\tilde O\left(\frac{n}{\epsilon^2}\right)$ edges. The degree of a vertex $v$ in the subgraph $G[H]$ for any $H \subseteq U$ is then estimated by $\frac{1}{p}$ times the out-degree of $v$ in $\tilde G[H]$.
Then we execute a subroutine \PeelN to append vertices with large estimated degree to $\ddorder$ and remove them from $U$, where a vertex is said to have large estimated degree if its estimated degree is at least $\frac{\Delta}{1+\Theta(\epsilon)}$ for the maximum degree $\Delta$ of the current graph~$G$ (or its estimation).
Vertices are scanned in an arbitrary order. When a vertex~$v$ is considered, if its estimated degree in the current remaining graph is at least~$\frac{\Delta}{1+\Theta(\epsilon)}$, then we are supposed to add $v$ to $\ddorder$ and remove~$v$ together with its incident edges.
In this way, we claim that after an execution of \PeelN, only vertices with actual degree larger than $\frac{\Delta}{1+\epsilon}$ will be appended to $\ddorder$ and removed from $U$ and thus $\ddorder$ is always a $\frac{1}{1+\epsilon}$-DD order. Moreover, the remaining vertices have actual degree at most $\frac{\Delta}{1+\Theta(\epsilon)}$ so that the maximum degree of $G$ decreases by a factor $\frac{1}{1+\Theta(\epsilon)}$.

\begin{algorithm}
\caption{Approximate DD Order via Edge Sampling \Rmnum{1}}
\label{alg:approxdd-fewer-1}
\SetKwProg{Function}{Function}{}{}
\Function{\algoname{ApproxDD-Warmup}($G=(V,E),\epsilon, \delta$)}{
    $\ddorder \gets $ empty list \\
    $U \gets V$ \\
    $\hat \epsilon \gets \frac{\epsilon}{4+3\epsilon}$ \\
    $\Delta \gets n-1$  \Comment{Max degree estimate; alternatively, computed with one pass.} \\
	$\et \gets 1 + \epsilon$ \Comment{A removed vertex should have current degree at least $\frac{\Delta}{K}$} \\
	$T \gets \log_{1+\frac{\epsilon}{2}} n$ \Comment{Estimated number of iterations in while loop} \\
    \While{$U \neq \emptyset$}{
        Using one pass on the edge list, construct a directed graph $\tilde G$ on $U$ by sampling each direction of each edge in $G[U]$ independently with probability $p$: $\quad \quad p = \min \left\{ \frac{3 \et}{\hat \epsilon^2 \Delta} \ln  \frac{2nT}{\delta}, 1 \right\}$ \\
        $(U, \ddorder) \gets \PeelN(U,\tilde G, p, \Delta,\ddorder)$ \\
        $\Delta \gets \frac{\Delta}{1+ \frac{\epsilon}{2}}$
    }
    \Return $\ddorder$
}

\Function{\PeelN($U, \tilde G, p, \Delta, \ddorder$)}{
    $\alpha \gets \frac{3\epsilon}{4}$ \\
	$H \gets U$ \\
    \ForEach{$v \in U$ in an arbitrary order}{
        $\hat d(v) \gets \frac{1}{p} \times ($Out-degree of $v$ in $\tilde G[H])$ \Comment{Unbiased degree estimator in $G[H]$} \\
        \If{$\hat d(v) \ge \frac{\Delta}{1+\alpha}$}{
            Append $v$ to $\ddorder$ \\
            $H \gets H \setminus \{v\}$ \Comment{Remove vertex $v$} \\
        }
    }
    \Return $(H,\ddorder)$
}
\end{algorithm}

\begin{theorem}\label{theorem:approxdd-fewer-1}
With probability at least $1 - \delta$,
 Algorithm \ref{alg:approxdd-fewer-1} returns a $\frac{1}{1+\epsilon}$-DD order
using $O\left(\frac{\log n}{\epsilon}\right)$ passes and $O\left(\frac{n}{\epsilon^2} \log \frac{n}{\epsilon \delta}\right)$
words of memory.
\end{theorem}

\begin{proof}

\noindent\textbf{Proof structure.}
The proof separates the \emph{probabilistic} part from the \emph{deterministic} part.
In the probabilistic part, we first define a small set of \emph{bad events} that, if they occur, could violate either
(i) the memory bound, (ii) the pass bound, or (iii) correctness of the returned order.
We then bound the probability of each bad event using Chernoff bounds (\Cref{fact:Chernoff}) and
take a union bound over all vertices and all iterations.
In the deterministic part, we condition on the event that none of these bad events occur in the first $T$ iterations,
and prove deterministically that the algorithm (a) makes geometric progress in $\Delta$ and thus terminates in $T$ iterations,
(b) outputs a $\frac{1}{1+\epsilon}$-DD order, and (c) never exceeds the claimed space.

\noindent \textbf{Bad Events in One Iteration.}
Suppose the maximum degree of $G[U]$ at the beginning of \PeelN is at most~$\Delta$.
To simplify the notation,
for a vertex $v \in U$ at the moment it is
considered in \PeelN,
we use $d_v$ to denote its true degree in the remaining graph $G[H]$
and $\hat{d}_v$ to denote the unbiased estimation of $d_v$,
which has the distribution $\frac{1}{p} \times B(d_v, p)$.

The algorithm makes decisions using only the sampled graph $\tilde G$ and the estimates $\hat d_v$.
Thus, there are two distinct failure modes to control:
(1) $\tilde G$ can be too large (breaking the memory budget), and
(2) $\hat d_v$ can deviate enough to change a peel/keep decision (breaking progress or correctness).
Accordingly, we consider the following bad events.

\begin{description}

\item[Type \Rmnum{1}:] The sampled directed graph $\tilde G$ has
more than $(1 + \hat{\epsilon})\cdot n \Delta p = \Theta\left(\frac{n}{\epsilon^2} \log \frac{n}{\epsilon \delta}\right)$ edges.
By Fact~\ref{fact:Chernoff},
this happens with probability at most $\exp\left(-\frac{1}{3} \hat{\epsilon}^2 n \Delta p\right) \leq
\frac{\delta}{2 T}$.

\noindent\emph{Intuition.}
The number of sampled directed edges is a sum of independent Bernoulli variables.
At the start of the iteration, the total number of directed edges in $G[U]$ is at most $n\Delta$,
and each is kept with probability $p$, so the expectation is at most $n\Delta p$.
A multiplicative Chernoff bound then shows that exceeding $(1+\hat\epsilon)n\Delta p$ has exponentially small probability,
which is why choosing $p$ so that $n\Delta p=\Theta\!\left(\frac{n}{\epsilon^2}\log\frac{n}{\epsilon\delta}\right)$
suffices for a union bound over $T$ iterations.

\item[Type \Rmnum{2}:] For a vertex~$v$,
if $d_v \geq \frac{\Delta}{K}$, the bad event is
$\hat{d}_v < (1 - \hat{\epsilon}) \cdot d_v$;
if $d_v < \frac{\Delta}{K}$,
the bad event is
$\hat{d}_v \geq (1 + \hat{\epsilon}) \cdot \frac{\Delta}{K}$.
By Fact~\ref{fact:Chernoff},
the bad event for each vertex is at most
$\exp\left(-\frac{1}{3} \hat{\epsilon}^2 \frac{\Delta p}{K}\right)  \leq \frac{\delta}{2nT}$.

\noindent\emph{Intuition.}
We split Type~\Rmnum{2} into two cases because the harmful direction of estimation error depends on the true degree.
If $d_v$ is large ($d_v\ge \Delta/K$), then underestimating it could prevent peeling a vertex that should be removed,
slowing the decrease of $\Delta$.
If $d_v$ is small ($d_v< \Delta/K$), then overestimating it could wrongly peel $v$,
which would violate the DD property.
The Chernoff bound is applied to the binomial $B(d_v,p)$, and we upper bound the failure probability uniformly by
lower bounding the relevant mean by $(\Delta/K)p$, enabling a union bound over all vertices and all iterations.
\Cref{fig:bad-events} helps to illustrate the area of type \Rmnum{2} bad events.

\end{description}

\begin{figure}
  \centering
  \begin{tikzpicture}[scale=2.5]
    \pgfmathsetmacro\border{1.7}
    \pgfmathsetmacro\xA{0.6}
    \pgfmathsetmacro\xC{1.4}
    \pgfmathsetmacro\yB{1}
    \pgfmathsetmacro\k{\yB/\xC}
    \pgfmathsetmacro\eps{0.02}

    \fill[gray!30] (0,\yB) rectangle (\xA,\border-\eps);
    \fill[gray!30] (\xA,0) -- (\xA,\k*\xA) -- plot[domain=\xA:\border-\eps] (\x,{\k*\x}) -- (\border-\eps,0) -- cycle;
    \fill[orange!20] (0,0) rectangle (\xA,\yB);
    \fill[blue!20] (\xA,\yB) -- (\xA,{\k*\xA}) -- plot[domain=\xA:\xC] (\x,{\k*\x}) -- cycle;

    \draw[->,thick] (0,0) -- (\border,0) node[below] {$d_v$};
    \draw[->,thick] (0,0) -- (0,\border) node[left] {$\hat d_v$};
    \draw (\xA,0.05) -- (\xA,-0.05) node[below] {$\dfrac{\Delta}{K}$};
    \draw (\xC,0.05) -- (\xC,-0.05) node[below] {$\dfrac{\Delta}{1+\frac{\epsilon}{2}}$};
    \draw (0.05,\yB) -- (-0.05,\yB) node[left] {$\dfrac{\Delta}{1+\alpha}$};

    \draw[dashed,domain=0:\border-\eps] plot (\x,{\k*\x});

    \filldraw (\xA,\yB) circle (0.5pt);
    \filldraw (\xC,\yB) circle (0.5pt);

    \draw[dashed] (0,\yB) -- (\xC,\yB);
    \draw[dashed] (\xC,\yB) -- (\xC,0);
    \draw[dashed] (\xA,\border-\eps) -- (\xA,0);

    \node at ({\border+0.2},{\k*\border+0.1}) {$\hat d_v = (1-\hat\epsilon)d_v$};
    \node at ({\xA/2}, {(\yB+\border)/2}) {bad};
    \node at (\xC, {\yB/2-0.1}) {bad};
    \node at ({\xA/2}, {\yB/2+0.1}) {must};
    \node at ({\xA/2}, {\yB/2}) {remain};
    \node at (\xA+0.2, \yB-0.2) {small};
  \end{tikzpicture}
  \caption{The relation between the true degree $d_v$ and the estimated degree $\hat{d}_v$. The area of bad events is filled in gray. The blue triangle area indicates that if $v$ is not removed in \PeelN, we have $d_v \le \frac{\Delta}{1+\frac{\epsilon}{2}}$. The orange rectangle area indicates that if $d_v < \frac{\Delta}{K}$, $v$ must remain after \PeelN.} \label{fig:bad-events}
\end{figure}
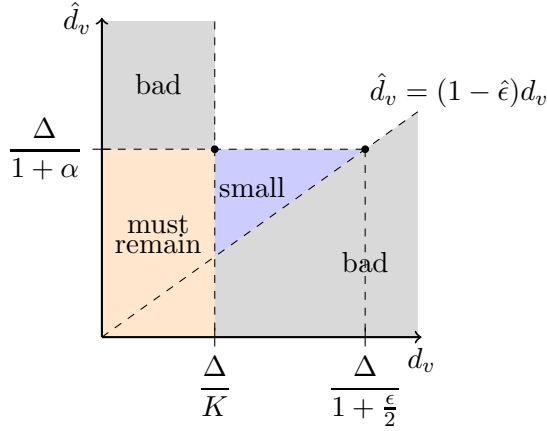

The union of all bad events within one iteration has probability at most $\frac{\delta}{T}$.
Therefore, it suffices to show that if no bad event happens within the first $T$ iterations, then the algorithm must terminate and return a correct result within the required memory.

\noindent \textbf{The Number of Passes.} We consider one execution of \PeelN.
If no type \Rmnum{2} bad event happens, we argue that if a vertex~$v$ is not removed,
it must be the case that at the moment it is considered,
we have $d_v \leq \frac{\Delta}{1 + \frac{\epsilon}{2}}$,
whose degree may drop even further as other vertices are removed later.

We do not need to worry about the case $d_v < \frac{\Delta}{K} < \frac{\Delta}{1 + \frac{\epsilon}{2}}$.
For the case $d_v \geq \frac{\Delta}{K}$, since $v$ is not removed, we have $\hat{d}_v < \frac{\Delta}{1 + \alpha}$.
Moreover, the type \Rmnum{2} bad event at $v$ does not happen, and so we have $d_v \leq \frac{\hat{d}_v}{1 - \hat{\epsilon}} <\frac{\Delta}{(1+\alpha)(1 - \hat{\epsilon})} = \frac{\Delta}{1 + \frac{\epsilon}{2}}$. This corresponds to the blue triangle area in \Cref{fig:bad-events}.

Therefore, If no type \Rmnum{2} bad event happens within the first $T$ iterations, any vertex that survives the scan already has degree at most $\frac{\Delta}{1 + \frac{\epsilon}{2}}$ at the moment it is considered.
This means that $\Delta$ will drop below 1 after $T = \log_{1 + \frac{\epsilon}{2}} n = O\left(\frac{\log n}{\epsilon}\right)$ iterations at which point the algorithm must terminate.
Since each iteration uses 1 pass to build $\tilde G$, the total passes are $T$.

\noindent\textbf{Correctness.}
Within one iteration of \PeelN, it suffices to show that only vertices~$v$ satisfying $d_v \geq \frac{\Delta}{K}$ could possibly be removed, since for any other vertex $u$ remaining before this iteration, we have $d(u|G(v)) \le \Delta$ and hence $d(v|G(v)) = d_v \ge \frac{\Delta}{K} \ge \frac{1}{1+\epsilon} d(u|G(v))$.

Since the type \Rmnum{2} bad event for a vertex $v$ satisfying $d_v < \frac{\Delta}{K}$ does not happen,
we must have $\hat{d}_v < (1 + \hat{\epsilon}) \cdot \frac{\Delta}{K} = \frac{\Delta}{1 + \alpha}$,
which means such a vertex~$v$ will not be appended to $\ddorder$ when being considered.
This corresponds to the orange rectangle area in \Cref{fig:bad-events}.

Therefore, if no type \Rmnum{2} bad event happens within the first $T$ iterations, then $\ddorder$ is a $\frac{1}{1+\epsilon}$-DD order.

\noindent\textbf{Memory.}
The memory usage is dominated by the number of edges in the sampled graph $\tilde G$ in each iteration.
Since the type \Rmnum{1} bad event does not happen within the first $T$ iterations, we conclude that the number of edges sampled in each iteration is at most $O\left(\frac{n}{\epsilon^2} \log \frac{n}{\epsilon \delta}\right)$.
\end{proof}

\subsection{From $O(\log n)$ to $O(1/c)$ Passes}
\label{sec:fast-ddorder-2}

One idea to decrease the number of passes in Algorithm~\ref{alg:approxdd-fewer-1}
is to generate and remember several sampled directed graphs during a single pass,
at the cost of using more memory.
In \Cref{alg:approxdd-fewer-2}, instead of having one pass per iteration, a pass is made every $q = \lfloor \log_{1+\frac{\epsilon}{2}} n^c \rfloor$ iterations for some constant $c>0$. The sampling probabilities are computed by the estimated maximum degree for each graph assuming its dropping by a factor of $1+\frac{\epsilon}{2}$. The same subroutine \PeelN is executed on the sampled graphs sequentially.
By choosing the appropriate $c$, we ensure that
it decreases the number of passes to $O\left(\frac{1}{c}\right)$,
and each pass can lead to a factor~$n^c$ decrease in the maximum degree~$\Delta$.

\begin{algorithm}[H]
\caption{Approximate DD Order via Edge Sampling \Rmnum{2}}
\label{alg:approxdd-fewer-2}
\SetKwProg{Function}{Function}{}{}
\Function{\ApproxDDb($G=(V,E),\epsilon, \delta, c$)}{
    $\ddorder \gets $ empty list \\
    $U \gets V$ \\
    $\hat \epsilon \gets \frac{\epsilon}{4+3\epsilon}$ \\
    $\Delta \gets n-1$  \Comment{Max degree estimate; alternatively, computed with one pass.} \\
	$\et \gets 1 + \epsilon$ \\
	$T \gets \log_{1+\frac{\epsilon}{2}} n$ \Comment{Estimated number of sampled directed graphs} \\
    $q, i \gets \lfloor \log_{1+\frac{\epsilon}{2}} n^c \rfloor$ \\
		
    \While{$U \neq \emptyset$}{
        \If{$i=q$}{
            \For(\Comment{Using one pass on edge list, sample $q$ directed graphs.}){$j$ from $0$ to $q-1$}{
				$\Delta_j \gets \frac{\Delta}{(1 + \frac{\epsilon}{2})^j}$ \\
				$p_j = \min \left\{ \frac{3 \et}{\hat \epsilon^2 \Delta_j} \ln  \frac{2nT}{\delta}, 1\right\}$ \\
				Construct a directed graph $\tilde G_j$ on $U$ by sampling each direction of each edge in $G[U]$ independently with probability $p_j$
			}
			$i \gets 0$
        }
        $(U, \ddorder) \gets \PeelN(U,\tilde G_i, p_i, \Delta_i, \ddorder)$ \Comment{Same subroutine as in \Cref{alg:approxdd-fewer-1}} \\
        $i \gets i+1$ \\
		\If{$i = q$} {
			$\Delta \gets \frac{\Delta_{q-1}}{1 + \frac{\epsilon}{2}}$
		}
    }
    \Return $\ddorder$
}
\end{algorithm}

\begin{theorem}\label{theorem:approxdd-fewer-2}
With probability at least $1 - \delta$,
 Algorithm~\ref{alg:approxdd-fewer-2} returns a $\frac{1}{1+\epsilon}$-DD order
using $O\left(\frac{1}{c}\right)$ passes and $O\left(\frac{n^{1+c}}{\epsilon^3} \log \frac{n}{\epsilon \delta}\right)$
words of memory.
\end{theorem}

\begin{proof}
Observe that the difference between Algorithms~\ref{alg:approxdd-fewer-1} and~\ref{alg:approxdd-fewer-2} is that in Algorithm~\ref{alg:approxdd-fewer-2}, during one pass of the edge lists, random directed graphs from several iterations are sampled and saved.

\noindent\textbf{Correctness and the Number of Passes.}
The key idea is a coupling argument via the directed edge sampling in each iteration: even though \Cref{alg:approxdd-fewer-2} samples multiple graphs in one pass, it can be seen as using exactly the same randomness that \Cref{alg:approxdd-fewer-1} would have used over the corresponding iterations.

Therefore, the two algorithms produce exactly the same order~$\ddorder$, and the correctness follows from \Cref{theorem:approxdd-fewer-1}.
As before, given that there is no bad event in the first $T$ iterations, the algorithm must terminate and uses $O(\frac{T}{q}) = O(\frac{1}{c})$ passes.

\noindent\textbf{Memory.}
Recall that in each pass, we sample and save $q$ directed subgraphs.
Assuming that no bad event has occurred in previous passes, the maximum degree of $G[U]$ is at most $\Delta$ at the beginning of the current pass.
For the $i$-th graph, the expected number of sampled directed edges is on the order of $n\Delta \cdot p_i$.
Hence, in this pass, the expected number of sampled edges in all the $q$ subgraphs is at most:

\begin{equation*}
\sum_{i=0}^{q-1} n \Delta p_i \le
O\left(\frac{n^{1+c}}{\epsilon^3} \log \frac{n}{\epsilon \delta}\right).
\end{equation*}

To get a high probability statement for memory usage,
observe that we are no worse than the situation in \Cref{theorem:approxdd-fewer-1}.
First, the expected number
of saved edges is larger in each pass, which means Chernoff Bound will give a smaller
failure probability for each pass.
Second, the number of passes are fewer, which means there are fewer terms
in the union bound.

In conclusion, with probability at least~$1-\delta$,
all the required properties as stated are satisfied.
\end{proof}

\ignore{
\section{Estimating the bucket sizes}

Let $S$ be a set. Suppose one can (i) sample from $S$ according to some distribution $p$, and (ii) compute $p(g)$ for any given $g \in S$. Let $X$ be a random variable distributed according to $p$, and define $N=\frac{1}{p(X)}$; clearly one can obtain a realization of $N$ through (i) and (ii) above. Note that:
\begin{align}
    \E[N] = \sum_{g \in S} \Pr{X=g} \cdot \E[N|X=g] = \sum_{g \in S} p(g) \cdot \frac{1}{p(g)} = |S|
\end{align}
Thus $N$ is an unbiased estimator of $|S|$. By this argument, after our preprocessing phase one can easily estimate the size of each bucket by using our sampling routines.

In the same way one can estimate the number of graphlets isomorphic to a connected $k$-vertex graph $H_i$. To this end, let set $S$ be the set of all graphlets, and define:
\begin{align}
    N_i =\frac{\Ind{X \simeq H_i}}{p(X)}
\end{align}
It is immediate to see that:
\begin{align}
    \E[N_i] = \sum_{g \in S_i} p(g) \cdot \frac{1}{p(g)} = |S_i|
\end{align}
where $S_i$ is the subset of graphlets isomorphic to $H_i$.
}

\section{Approximating the $k$-Graphlet Distribution}
\label{sec:graphlet-distribution}

As stated in \Cref{sec:ugs}, provided a $\frac{1}{1+\epsilon}$-DD order $\ddorder$ of a graph $G$, one can compute an initial distribution $\bp$ and use it to sample a $k$-graphlet $S$ of $G$ with probability $p(S)$. In this section, we complete the task of approximating the $k$-graphlet distribution $\bmu_k = (\mu_1, \cdots, \mu_{m_k})$ of $G$ based on this $k$-graphlet sampling.

\subsection{Sampling with Counters}
\label{sec:gldist-counter}

The approximation to the $k$-graphlet distribution is described in \Cref{alg:gd}. Using \Cref{alg:SampleGraphlet}, a graphlet $S$ can be sampled with some probability $p(S)$ computed by the initial distribution $\bp$. The original algorithm in \cite{Bourreau2024} further achieves uniform sampling by accepting the sampling of $S$ with probability $\frac{\Gamma}{p(S)}$ for a constant $\Gamma$ smaller than every $p(S)$. In our algorithm, we do not reject any sampling,
but employ the technique of Horvitz–Thompson estimators.
We maintain a collection of counters $c_1, \cdots, c_{m_k}$ for isomorphism classes $1, \cdots, m_k$ respectively. When we sample a graphlet $S$, recall that there is precisely one integer $r \in [1, m_k]$ such that $g_r$ is isomorphic to $S$. We then add $\frac{1}{p(S)}$ to $c_r$. The graphlet distribution is then obtained by normalizing these counters.


\begin{algorithm}
	\caption{Approximate $k$-Graphlet Distribution via Counters (Abstract Description)}
	\label{alg:gd}
	\SetKwProg{Function}{Function}{}{}
	\Function{\gd($G, \ddorder, k, \bp, T$)}{ 
        $c_1, \cdots, c_{m_k} \gets 0$ \\
        \For(\textbf{(sequentially/in parallel)}){$i \gets 1$ \KwTo $T$}{
            $S \gets$ \SampleGraphlet($G, \ddorder, k, \bp$) \Comment{Sample a $k$-graphlet $S$ with probability $p(S)$} \\
            $r \gets$ the index such that $S$ is isomorphic to $\graphlet_r$ \\
            $c_r \gets c_r + \frac{1}{p(S)}$  
        }
        $\hat C \gets {\sum_{i=1}^{m_k} c_i}$ \\
        $\hat \mu_i \gets \frac{c_i}{\hat C}$ \textbf{for each} $1 \le i \le m_k$ \\
        \Return $\hat \mu_1, \cdots, \hat \mu_{m_k}$
	}
\end{algorithm}

Recall that given a $\frac{1}{1+\epsilon}$-DD order $\prec$, $Z$ defined in \Cref{sec:ugs} is the sum of $d(v|G(v))^{k-1}$ over those vertices $v$ with $N_v>0$. We will use the property that both $\frac{1}{p(S)}$ and the number $L$ of graphlets in $G$ are bounded by $Z$.

\begin{fact}{(\cite[Lemma C.1]{Bourreau2024})}\label{lemma:bound-on-pS}
For each $k$-graphlet $S$, it holds that
\begin{equation*}
    \frac{1}{p(S)} \le (k-1)! (1+\epsilon)^{k-1} Z.
\end{equation*}
\end{fact}

\begin{lemma}\label{lemma:bound-on-L}
It holds that
\begin{equation*}
    L \ge \frac{Z}{(k-1)^{k-1}}.
\end{equation*}
\end{lemma}

To prove \Cref{lemma:bound-on-L}, we need the following lemma from \cite{B23}.

\begin{fact}{(\cite[Lemma 24]{B23})}\label{lemma:b23}
For every $v \in V$ such that $N_v > 0$, we have
\begin{equation*}
  N_v \ge \frac{d(v|G(v))^{k-1}}{(k-1)^{k-1}}.
\end{equation*}
\end{fact}

\begin{proof}[Proof of \Cref{lemma:bound-on-L}]
By \Cref{lemma:b23}, we have
\begin{equation*}
   L= \sum_{v \in V, \ N_v>0} N_v \ge \sum_{v \in V, \ N_v>0} \frac{d(v|G(v))^{k-1}}{(k-1)^{k-1}} = \frac{Z}{(k-1)^{k-1}}. \qedhere
\end{equation*}
\end{proof}

We give the analysis of \Cref{alg:gd}.

\begin{lemma}\label{lemma:bound-on-hatC}
For any $\frac{1}{1+\epsilon}$-DD order, $k$ and $0 < \epsilon_0, \delta < 1$, let $T$ be a positive integer satisfying
\begin{equation*}
  T \ge \frac{3 (k-1)! (k-1)^{k-1} (1+\epsilon)^{k-1}}{\epsilon_0^2} \ln{\frac{2}{\delta}},
\end{equation*}
then \Cref{alg:gd} computes $\hat C$ such that $\left|\frac{\hat C}{T} - L\right| < \epsilon_0 L$ with probability at least $1-\delta$.
\end{lemma}

\begin{proof}
Let $X_1, \cdots, X_T$ be the value of $\frac{1}{p(S)}$ of each iteration. Then we have $\hat C = \sum_{i=1}^T X_i$. In fact, $X_1, \cdots, X_T$ are independent random variables such that $X_i$ takes value $\frac{1}{p(S)}$ with probability $p(S)$ for each $k$-graphlet $S$ in $G$. Each $X_i$ is at most $(k-1)!(1+\epsilon)^{k-1}Z$ by \Cref{lemma:bound-on-pS} and has expectation $L$. By \Cref{fact:Chernoff2,lemma:bound-on-L} we have
\begin{equation*}
  \begin{aligned}
    \Pr{\left|\frac{\hat C}{T} - L\right| \ge \epsilon_0 L} &\le 2\exp\left( -\frac{\epsilon_0^2 L T}{3 (k-1)! (1+\epsilon)^{k-1} Z} \right) \\
    &\le 2\exp\left( -\frac{\epsilon_0^2 T}{3 (k-1)! (k-1)^{k-1} (1+\epsilon)^{k-1}} \right) \\
    &= \delta.
  \end{aligned}
\end{equation*}
\end{proof}

\begin{theorem}\label{thm:bound-on-T}
For any $\frac{1}{1+\epsilon}$-DD order, $k$ and $0 < \alpha, \delta < 1$, let $T$ be a positive integer satisfying
\begin{equation*}
  T \ge \frac{12 ((k-1)!)^2 (k-1)^{2k-2} (1+\epsilon)^{2k-2}}{\alpha^2 (1 - \alpha)^2} \ln{\frac{4m_k}{\delta}},
\end{equation*}
then \Cref{alg:gd} computes $\hat \mu_1, \cdots, \hat \mu_{m_k}$ such that with probability at least $1-\delta$, $|\hat \mu_i - \mu_i| \le \alpha$ for each $1 \le i \le m_k$.
\end{theorem}

\begin{proof}
We want to bound the deviation of each normalized counter $\hat \mu_i = c_i/\hat C$ from the true proportion $\mu_i = l_i/L$, uniformly over all isomorphism classes $i\in[m_k]$.
The main technical issue is that $\hat \mu_i$ is a ratio with random denominator $\hat C$.
To handle this cleanly, we first condition on the event that $\hat C$ concentrates around its mean (so the denominator is well-controlled), and then apply a concentration bound to each numerator $c_i$.
Finally we union bound over the $m_k$ classes and combine with the probability that the conditioning event holds.

Assume that the condition $|\frac{\hat C}{T} - L| < \frac{\alpha}{2} L$ holds. Then $\hat C$ lies in the interval $[(1-\alpha/2)LT,\ (1+\alpha/2)LT]$.
For each $1 \le i \le m_k$, recall that $\mu_i = \frac{l_i}{L}$ where $l_i$ is the number of $k$-graphlets in $G$ which are isomorphic to $\graphlet_i$.
We have
\begin{equation}\label{eq:mui}
  \begin{aligned}
    \Pr{|\hat \mu_i - \mu_i| \ge \alpha}
    = \Pr{\left|\frac{c_i}{\hat C} - \frac{l_i}{L}\right| \ge \alpha}
    \le \Pr{\frac{c_i}{\left(1-\frac{\alpha}{2}\right)LT} - \frac{l_i}{L} \ge \alpha} + \Pr{\frac{l_i}{L} - \frac{c_i}{\left(1+\frac{\alpha}{2}\right)LT} \ge \alpha}.
  \end{aligned}
\end{equation}

We analyze $\Pr{\frac{c_i}{\left(1-\frac{\alpha}{2}\right)LT} - \frac{l_i}{L} \ge \alpha}$, and the other side is similar.
We rewrite the deviation in terms of $\frac{c_i}{T}$:

\begin{equation*}
  \begin{aligned}
    \Pr{\frac{c_i}{\left(1-\frac{\alpha}{2}\right)LT} - \frac{l_i}{L} \ge \alpha} &= \Pr{\frac{c_i}{T} - \left(1-\frac{\alpha}{2}\right)l_i \ge \alpha \left(1-\frac{\alpha}{2}\right) L} \\
    &= \Pr{\frac{c_i}{T} - l_i \ge \alpha \left(1-\frac{\alpha}{2}\right) L - \frac{\alpha}{2} l_i} \\
    &\le \Pr{\frac{c_i}{T} - l_i \ge \frac{\alpha}{2}(1 - \alpha) L},
  \end{aligned}
\end{equation*}
where the last inequality is due to $l_i \le L$.
Similarly to the proof of \Cref{lemma:bound-on-hatC}, let $X_1, \cdots, X_T$ be the increment of $c_i$ after each iteration, i.e., $c_i = \sum_{j=1}^T X_j$.
In fact, $X_1, \cdots, X_T$ are independent random variables such that $X_j$ takes value $\frac{1}{p(S)}$ with probability $p(S)$ for each $k$-graphlet $S$ in $G$ that is isomorphic to $\graphlet_i$, and $0$ with the rest probability.
Each $X_i$ is at most $(k-1)!(1+\epsilon)^{k-1}Z$ by \Cref{lemma:bound-on-pS} and has expectation $l_i$.
By \Cref{fact:Hoeffding,lemma:bound-on-L},
\begin{equation*}
  \begin{aligned}
    \Pr{\frac{c_i}{T} - l_i \ge \frac{\alpha}{2}(1 - \alpha) L}
    \le& 2\exp\left( -\frac{T\alpha^2 (1 - \alpha)^2 L^2}{2 ((k-1)! (1+\epsilon)^{k-1} Z)^2} \right) \\
    \le& 2\exp\left( -\frac{T\alpha^2 (1 - \alpha)^2}{2 ((k-1)!)^2 (1+\epsilon)^{2k-2} (k-1)^{2k-2} } \right)
    \le \frac{\delta}{4m_k}.
  \end{aligned}
\end{equation*}

One can argue similarly that
\begin{equation*}
  \Pr{\frac{l_i}{L} - \frac{c_i}{\left(1+\frac{\alpha}{2}\right)LT} \ge \alpha} \le \frac{\delta}{4m_k}.
\end{equation*}

Recall that the computation above is conditioned on $|\frac{\hat C}{T} - L| < \frac{\alpha}{2} L$.
We denote this condition by $\mathcal C$.
Applying a union bound to both sides of \Cref{eq:mui} for each $1 \le i \le m_k$, we obtain that
\begin{equation*}
  \Pr{\bigcup_{i \in [m_k]} |\hat \mu_i - \mu_i| > \alpha \ \bigg| \ \mathcal C} \le \frac{\delta}{2}.
\end{equation*}

Also note that
\begin{equation*}
  T \ge \frac{12 (k-1)! (k-1)^{k-1} (1+\epsilon)^{k-1}}{\alpha^2} \ln{\frac{4}{\delta}},
\end{equation*}
which gives $\Pr{\bar{\mathcal C}} \le \frac{\delta}{2}$ by \Cref{lemma:bound-on-hatC}.
Therefore, we have the upper bound for overall failure probability:
\begin{equation*}
  \begin{aligned}
    \Pr{\bigcup_{i \in [m_k]} \!\!\!|\hat \mu_i - \mu_i| > \alpha}
    \ \le\  \Pr{\bar{\mathcal C}} + \Pr{\bigcup_{i \in [m_k]} \!\!\!|\hat \mu_i - \mu_i| > \alpha \ \bigg| \ \mathcal C}
    \ \le\  \delta.
  \end{aligned}
\end{equation*}
\end{proof}

Now consider the number of passes of \Cref{alg:gd}. Recall that, with a memory of $M$ words, we can run \Cref{alg:gd} in parallel so that it returns $\Theta\left(\frac{M}{k^2}\right)$ samples using $O(k)$ passes as described in \Cref{sec:ugs}. Therefore, by \Cref{thm:bound-on-T}, $O\left(\frac{(k(1+\epsilon))^{O(k)}}{\alpha^4 M} \ln \frac{1}{\delta}\right)$ passes are sufficient to obtain an estimated $k$-graphlet distribution with probability at least $1-\delta$ in which each probability in the distribution differs at most $\alpha$ from the standard one.

\subsection{Counters vs. Rejection: a Variance View}

As mentioned before, the algorithm originally described in~\cite{Bourreau2024} uses \emph{rejection} in the sense that it accepts the sampling of $S$ with probability $\frac{\Gamma}{p(S)}$ where $\Gamma \le \min_S p(S)$. One can transform this algorithm into an approximation of the $k$-graphlet distribution: set a collection of counters $c_1, \cdots, c_{m_k}$, and whenever a graphlet $S$ is sampled and accepted with probability $\frac{\Gamma}{p(S)}$, add $1$ to $c_r$ where $g_r$ is isomorphic to $S$. The pseudocode of this algorithm is shown in \Cref{alg:gdr}.

\begin{algorithm}
	\caption{Approximate $k$-Graphlet Distribution via Rejection (Abstract Description)}
	\label{alg:gdr}
	\SetKwProg{Function}{Function}{}{}
	\Function{\gdr($k, \bp, T$)}{
        $c_1, \cdots, c_{m_k} \gets 0$ \\
        \For(\textbf{(sequentially/in parallel)}){$i \gets 1$ \KwTo $T$}{
    		Sample a $k$-graphlet $S$ uniformly \\
            $r \gets$ the index such that $S$ is isomorphic to $\graphlet_r$ \\
            $c_r \gets c_r + 1$
        }
        $\hat C \gets {\sum_{i=1}^{m_k} c_i}$ \\
        $\hat \mu_i \gets \frac{c_i}{\hat C}$ \textbf{for each} $1 \le i \le m_k$ \\
        \Return $\hat \mu_1, \cdots, \hat \mu_{m_k}$
	}
\end{algorithm}

We give an intuitive explanation of why using \gd in place of \gdr can only increase the concentration of the frequency estimates.
To this end, consider a modified version of \gd that, instead of adding $\frac{1}{p(S)}$ to the counter of the isomorphism class of $S$, adds $\frac{\Gamma}{p(S)}$, where $\Gamma$ is the same used by \gdr.
Note that this modification is immaterial as far as the output is concerned: since the output estimates are always rescaled in order to sum to $1$, multiplying all counters by a constant $\Gamma$ does not make any difference.
However, now one can see how \gdr is the randomized version of \gd.
Indeed, while \gdr increases the counter of $S$ by one with probability $\frac{\Gamma}{p(S)}$, \gd increases it deterministically by $\frac{\Gamma}{p(S)}$.
Thus, both algorithms increase the counter by the same amount \emph{in expectation}, but \gdr adds the variance of a coin toss.
This implies \gd can only yield higher concentration.

\subsection{The Complete Algorithm}
\Cref{alg:full} gives the complete pseudocode for approximating the $k$-graphlet distribution.
It first calls \Cref{alg:approxdd-fewer-2} to compute a vertex ordering that is a $\frac{1}{1+\epsilon}$-DD ordering with probability at least $1-\delta$.
Then it calls \Cref{alg:InitDistrib} to obtain the initial distribution (see \Cref{sec:ugs} for the definition).
Finally, it calls \Cref{alg:gd} to estimate the $k$-graphlet distribution.
By combining \Cref{theorem:approxdd-fewer-2,thm:bound-on-T}, one can easily verify that \Cref{alg:full} satisfies \Cref{thm:1}.

\begin{algorithm}
	\caption{Approximate $k$-Graphlet Distribution}
	\label{alg:full}
	\KwIn{an undirected graph $G$, the size of graphlets $k$, parameters $\epsilon, \delta, c, T$}
    \KwOut{the estimated $k$-graphlet distribution $\hat \mu_1, \cdots, \hat \mu_{m_k}$}

    $\ddorder \gets$ \ApproxDDb($G, \epsilon, \delta, c$) \\
    $\bp \gets$ \InitDistrib($G, \ddorder$) \\
    $\hat \mu_1, \cdots, \hat \mu_{m_k} \gets$ \gd($G, \ddorder, k, \bp, T$) \\
    \Return $\hat \mu_1, \cdots, \hat \mu_{m_k}$
\end{algorithm}

\section{Experiments}
\label{sec:experiments}
We conducted experiments to evaluate the practical performance of our algorithm and compare it to the algorithm of~\cite{Bourreau2024}.
Our experiments are designed to answer the following questions:
\begin{itemize}
  \item[(Q1)] How do the preprocessing phase (computing an approximate DD order) of \ApproxDDb and \ApproxDD compare in terms of number of passes, peak memory and the quality of DD order?
  \item[(Q2)] How is the performance of the overall algorithm (for estimating the $k$-graphlet distribution to achieve a target accuracy) in terms of number of passes and overall peak memory, and how do the ``counter'' variants compare to the ``rejection'' variant?
\end{itemize}

We organize the remainder of this section around these objectives: \Cref{exp:DD} addresses (Q1), and \Cref{exp:distr} addresses (Q2).

All algorithms were implemented in C++
(including our re-implementation of the baselines from~\cite{Bourreau2024} for a fair comparison)
and experiments were run on an Ubuntu server equipped with an Intel Xeon Silver 4108 CPU (1.80GHz) and 28GB of main memory.
\footnote{Our implementation is available at: \url{https://github.com/l2l7l9p/GD-Streaming}.}

\subsection{Datasets}

The computation of the graphlet distribution is relevant not only for sparse graphs, which are common in real-world data, but also for dense graphs such as the similarity graph among facial images. Our theoretical guarantee in \Cref{thm:1} holds regardless of graph density. Therefore, we ran experiments on both sparse and dense graphs from various sources to show that our algorithm scales well in both cases. Table~\ref{tab:datasize} reports all dataset statistics.

The datasets were prepared as follows:
\begin{itemize}
  \item NY Times, Twitter(WWW), Twitter(MPI) and Friendster are from the KONECT website \cite{dataset-konect}\footnote{\url{http://konect.cc/networks/}}. We removed edge directions, weights, self-loops, duplicate edges and any other irrelevant data, so as to retain only a list of undirected edges.
  \item  Sim-0 through Sim-7 are constructed from the database CelebA\footnote{\url{https://mmlab.ie.cuhk.edu.hk/projects/CelebA.html}}, which consists of 202599 human face images, each one tagged with 40 binary attributes. We represent each image by a vertex, and add an edge between two vertices if the corresponding images differ in at most $0,\ldots,7$ attributes.
  \item ER-0 through ER-6 are synthetic graphs, on the same number of vertices as Sim-*. Let $n_i$ and $m_i$ be the number of vertices and the number of edges of Sim-$i$ respectively. Each ER-$i$ is generated according to the Erd\H{o}s-R\'enyi model by drawing each edge independently with probability $\frac{m_i}{n_i(n_i-1)/2}$ so that the expected density of ER-$i$ is consistent with Sim-$i$.
\end{itemize}
This section gives plots for a representative subset of datasets;
the remaining ones are similar and are available in \Cref{sec:morefigures}.

\begin{table}
  \centering
  \caption{Dataset statistics. The first 4 graphs are real-world large graphs. Sim-* are the graphs indicating the similarity among images. ER-* are synthetic random graphs.}\label{tab:datasize}
  \begin{tabular}{r|rrr}
     \hline
     Dataset & $|V|$ & $|E|$ & $|E|/|V|$ \\
     \hline\hline
     NY Times \cite{dataset-konect}  & 401,388 & 69,654,798 & 173.53 \\
     Twitter(WWW) \cite{dataset-twitter} & 41,652,230 & 1,202,513,046 & 28.87 \\
     Twitter(MPI) \cite{dataset-twittermpi} & 52,579,682 & 1,614,106,187 & 30.70 \\
     Friendster \cite{dataset-konect} & 68,349,466 & 1,811,849,342 & 26.51 \\
     Sim-0 & 202,599 & 948,690 & 4.68 \\
     Sim-1 & " & 9,825,190 & 48.50 \\
     Sim-2 & " & 51,840,951 & 255.88 \\
     Sim-3 & " & 185,601,680 & 916.10 \\
     Sim-4 & " & 507,709,457 & 2505.98 \\
     Sim-5 & " & 1,134,892,538 & 5601.67 \\
     Sim-6 & " & 2,165,414,225 & 10688.18 \\
     Sim-7 & " & 3,637,218,904 & 17592.80 \\
     ER-0 & 202,599 & 949,312 & 4.69 \\
     ER-1 & " & 9,821,921 & 48.48 \\
     ER-2 & " & 51,843,003 & 255.89 \\
     ER-3 & " & 185,602,510 & 916.11 \\
     ER-4 & " & 507,715,524 & 2506.01 \\
     ER-5 & " & 1,134,903,860 & 5601.72 \\
     ER-6 & " & 2,165,448,192 & 10688.35 \\
     \hline
  \end{tabular}
\end{table}

\subsection{Computing the DD Order}
\label{exp:DD}

\noindentparagraph{Setup}
We evaluate \ApproxDDb (our algorithm) and \ApproxDD (the competitor from~\cite{Bourreau2024}) for computing a $\frac{1}{1+\eps}$-DD order, fixing $\eps=0.1$.
To test the passes-vs-memory tradeoff predicted by \Cref{thm:1}, we varied $c \in [0,0.5]$ (the performance of \ApproxDD does not depend on $c$ and is therefore constant across the single plot).
Each point shown in \Cref{fig:c} is the average of 5 executions; executions were terminated after 36 hours if still running.

\begin{figure}[H]
  \centering
  \begin{subfigure}{0.32\textwidth}
    \includegraphics[width=\textwidth,trim=0 20 0 0]{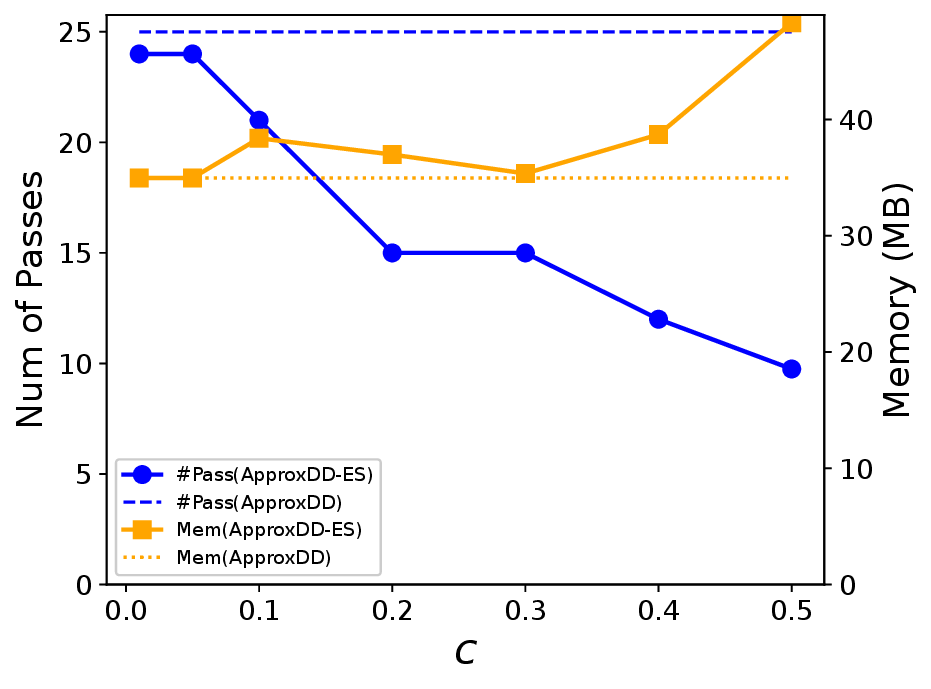}
    \caption{NY Times}
  \end{subfigure}
  \begin{subfigure}{0.32\textwidth}
    \includegraphics[width=\textwidth,trim=0 20 0 0]{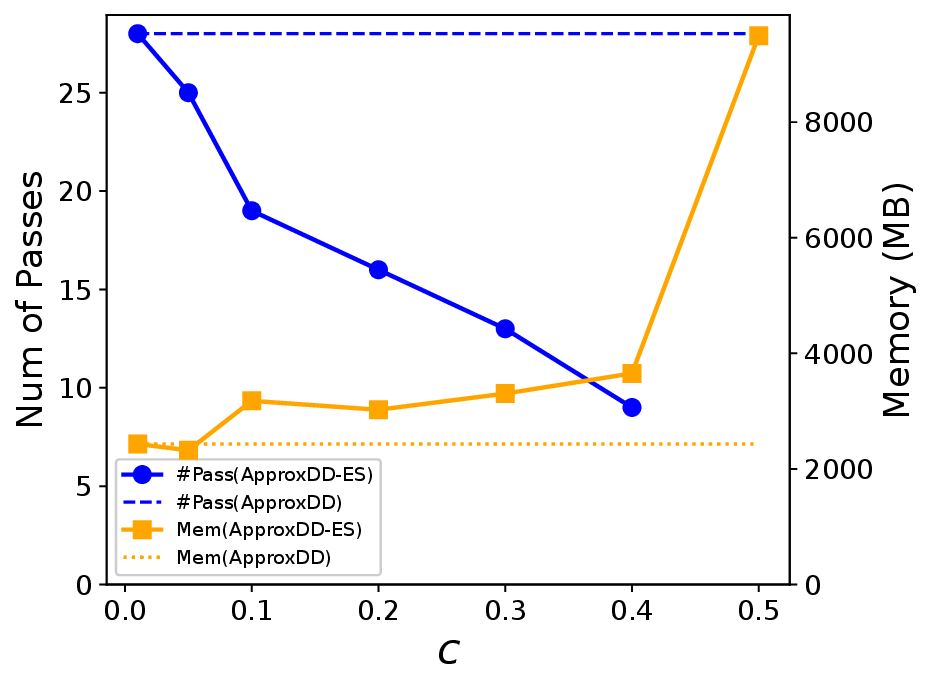}
    \caption{Twitter(WWW)}
  \end{subfigure}
  \begin{subfigure}{0.32\textwidth}
    \includegraphics[width=\textwidth,trim=0 20 0 0]{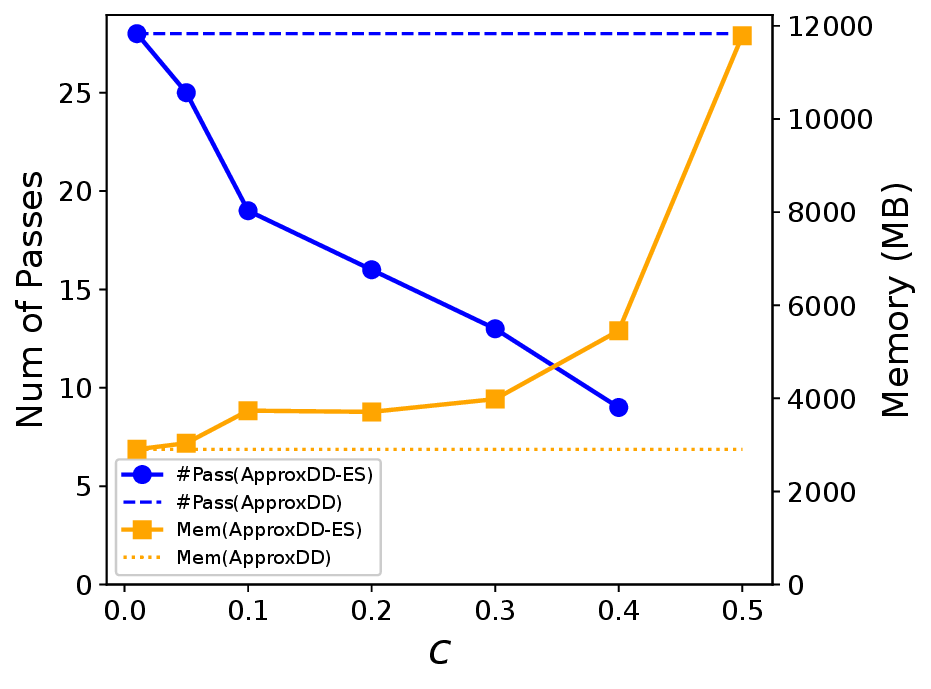}
    \caption{Twitter(MPI)}
  \end{subfigure}
  \begin{subfigure}{0.32\textwidth}
    \includegraphics[width=\textwidth,trim=0 20 0 0]{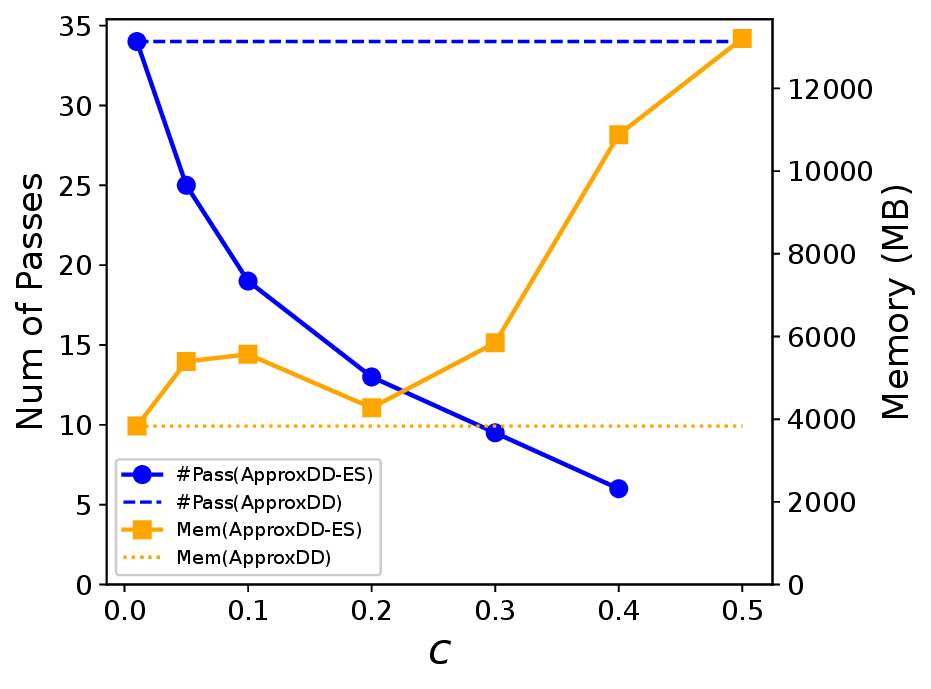}
    \caption{Friendster}
  \end{subfigure}
  \begin{subfigure}{0.32\textwidth}
    \includegraphics[width=\textwidth,trim=0 20 0 0]{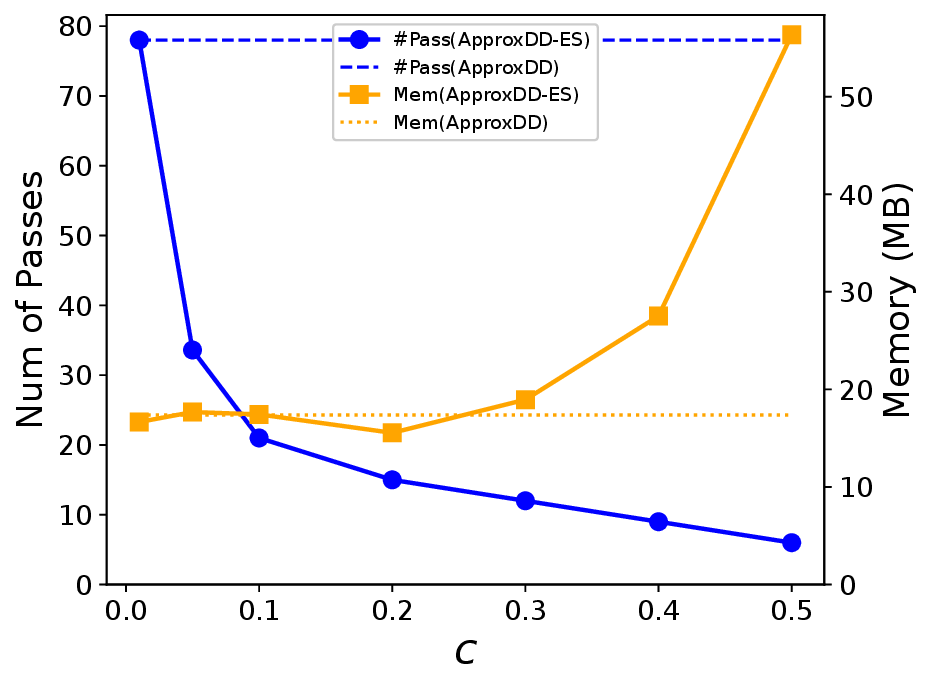}
    \caption{Sim-1}\label{fig:c-sim-1}
  \end{subfigure}
  \begin{subfigure}{0.32\textwidth}
    \includegraphics[width=\textwidth,trim=0 20 0 0]{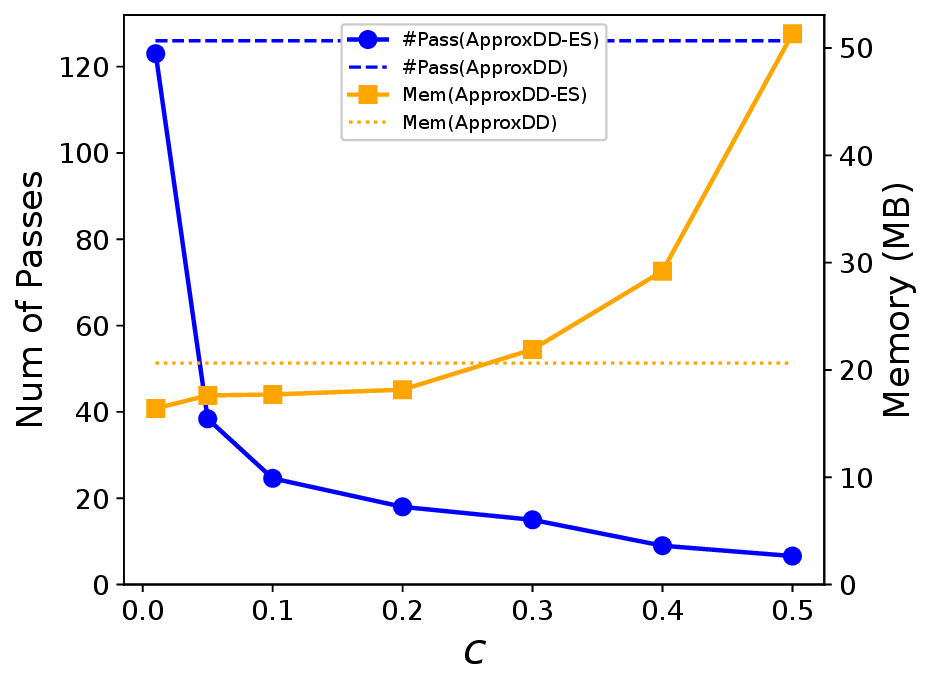}
    \caption{Sim-2}
  \end{subfigure}
  \begin{subfigure}{0.32\textwidth}
    \includegraphics[width=\textwidth,trim=0 20 0 0]{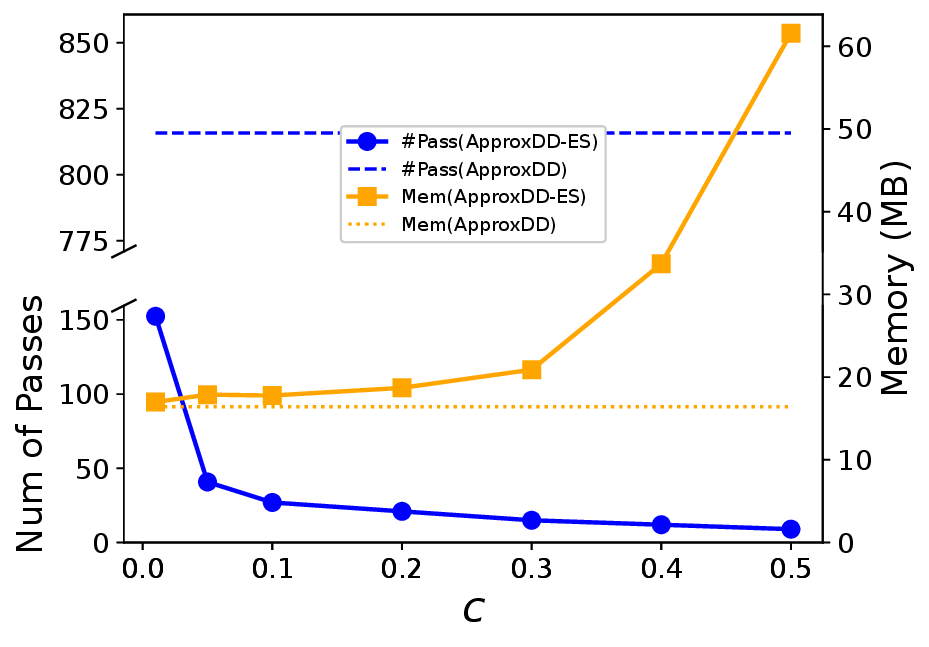}
    \caption{Sim-3}\label{fig:c-sim-3}
  \end{subfigure}
  \begin{subfigure}{0.32\textwidth}
    \includegraphics[width=\textwidth,trim=0 20 0 0]{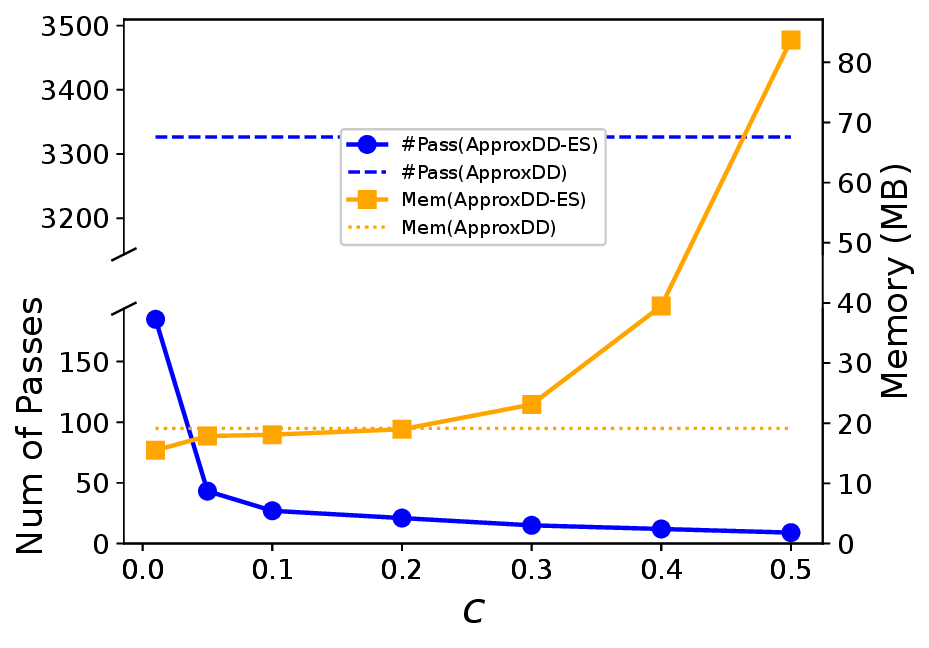}
    \caption{Sim-4}
  \end{subfigure}
  \begin{subfigure}{0.32\textwidth}
    \includegraphics[width=\textwidth,trim=0 20 0 0]{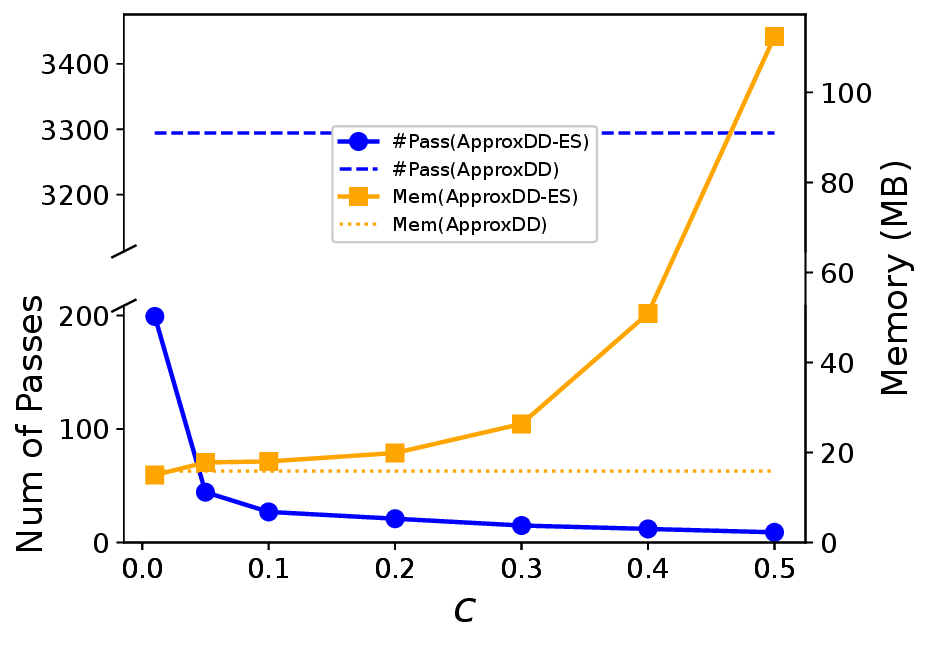}
    \caption{Sim-5}\label{fig:c-sim-5}
  \end{subfigure}
  \begin{subfigure}{0.32\textwidth}
    \includegraphics[width=\textwidth,trim=0 20 0 0]{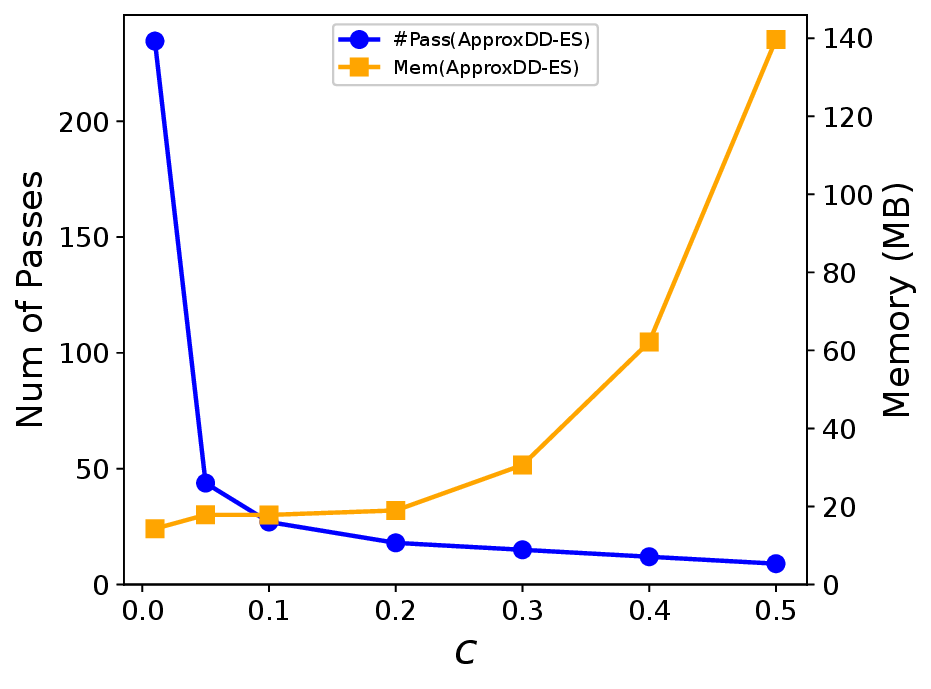}
    \caption{Sim-6}
  \end{subfigure}
  \begin{subfigure}{0.32\textwidth}
    \includegraphics[width=\textwidth,trim=0 20 0 0]{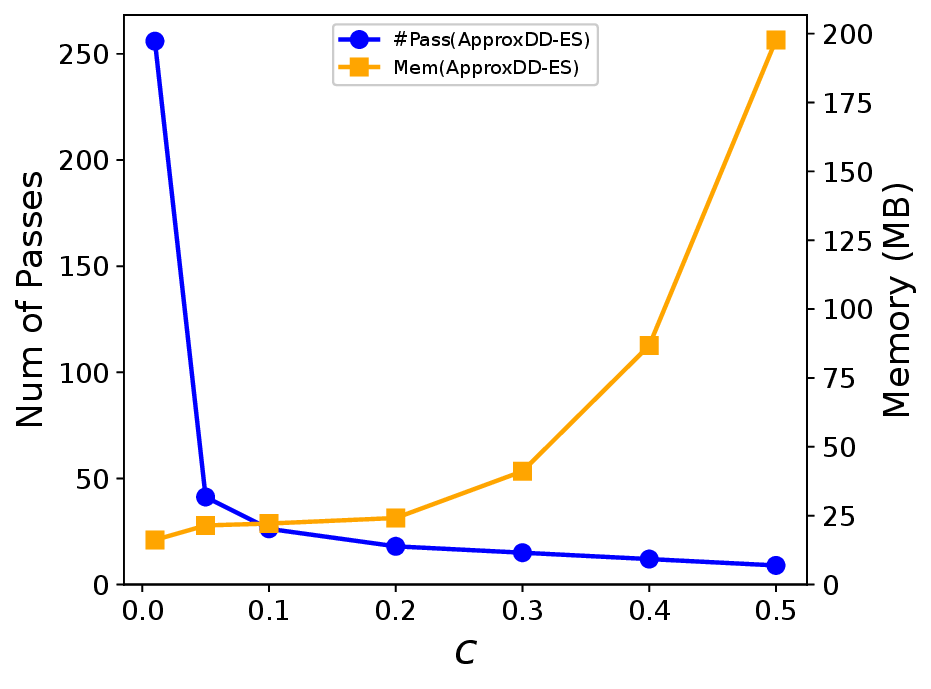}
    \caption{Sim-7}\label{fig:c-sim-7}
  \end{subfigure}
  \begin{subfigure}{0.32\textwidth}
    \includegraphics[width=\textwidth,trim=0 20 0 0]{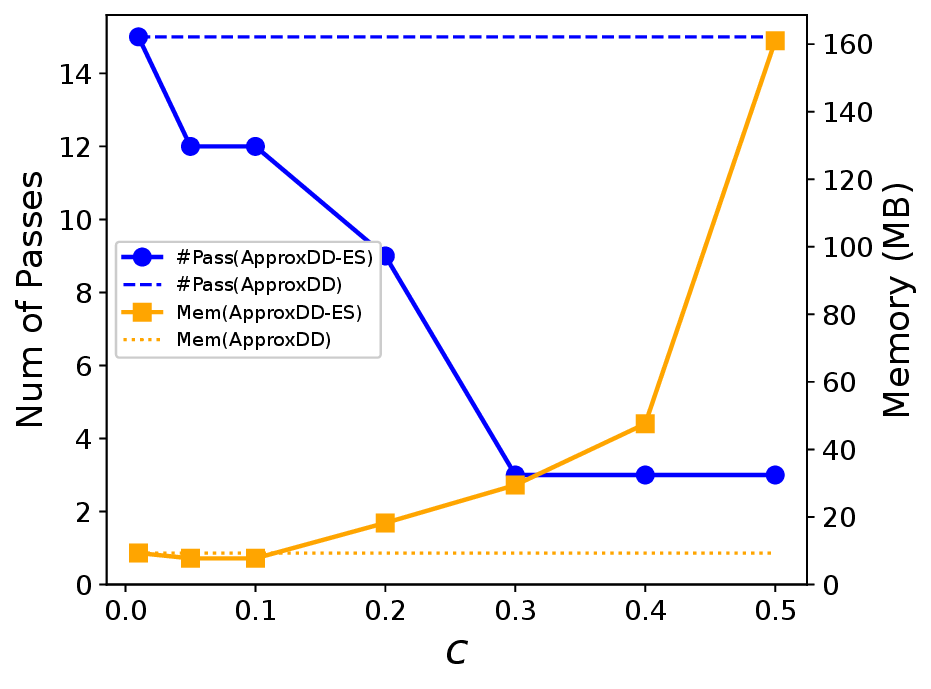}
    \caption{ER-0}
  \end{subfigure}
  \begin{subfigure}{0.32\textwidth}
    \includegraphics[width=\textwidth,trim=0 20 0 0]{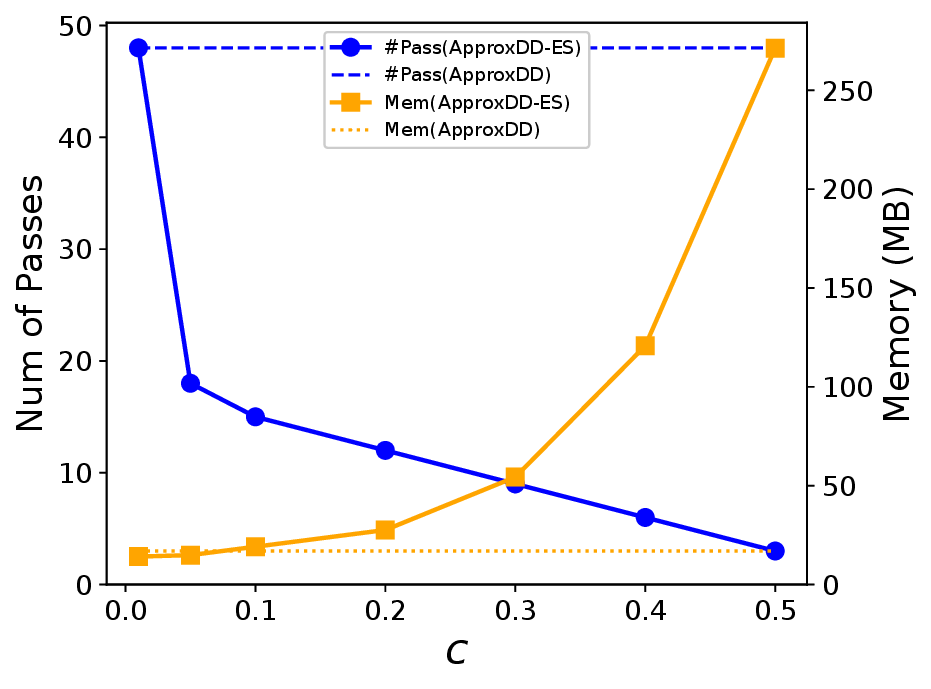}
    \caption{ER-1}
  \end{subfigure}
  \begin{subfigure}{0.32\textwidth}
    \includegraphics[width=\textwidth,trim=0 20 0 0]{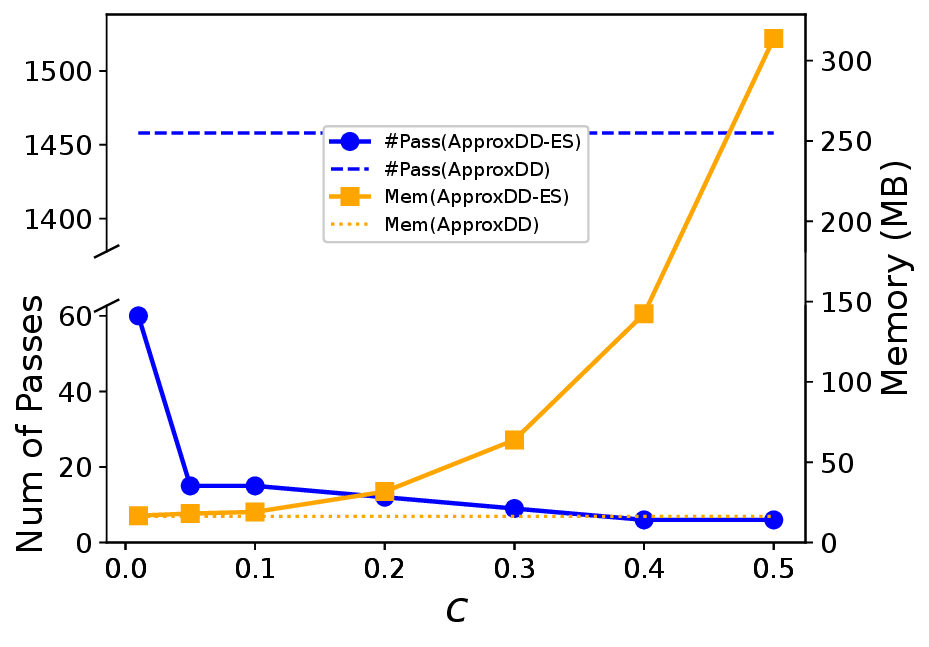}
    \caption{ER-2}
  \end{subfigure}
  \begin{subfigure}{0.32\textwidth}
    \includegraphics[width=\textwidth,trim=0 20 0 0]{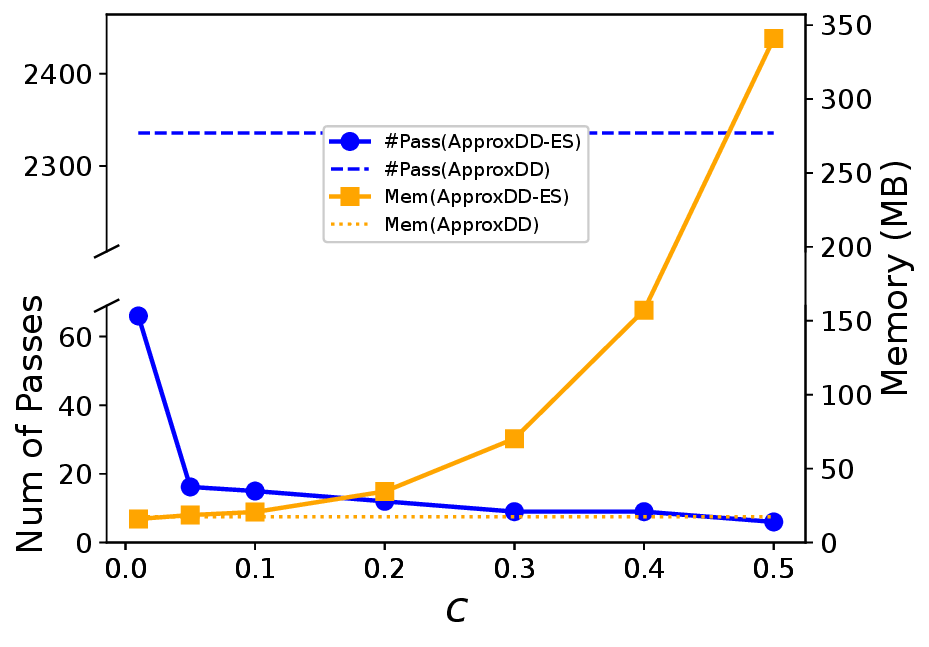}
    \caption{ER-3}
  \end{subfigure}
  \caption{Number of passes and peak memory usage, as a function of $c$, for \ApproxDD and \ApproxDDb. Each point shown is the average of 5 executions. Missing points for \ApproxDD mean it did not terminate within 36 hours.}\label{fig:c}
\end{figure}

\begin{figure}[h]\ContinuedFloat
  \centering
  \begin{subfigure}{0.32\textwidth}
    \includegraphics[width=\textwidth,trim=0 20 0 0]{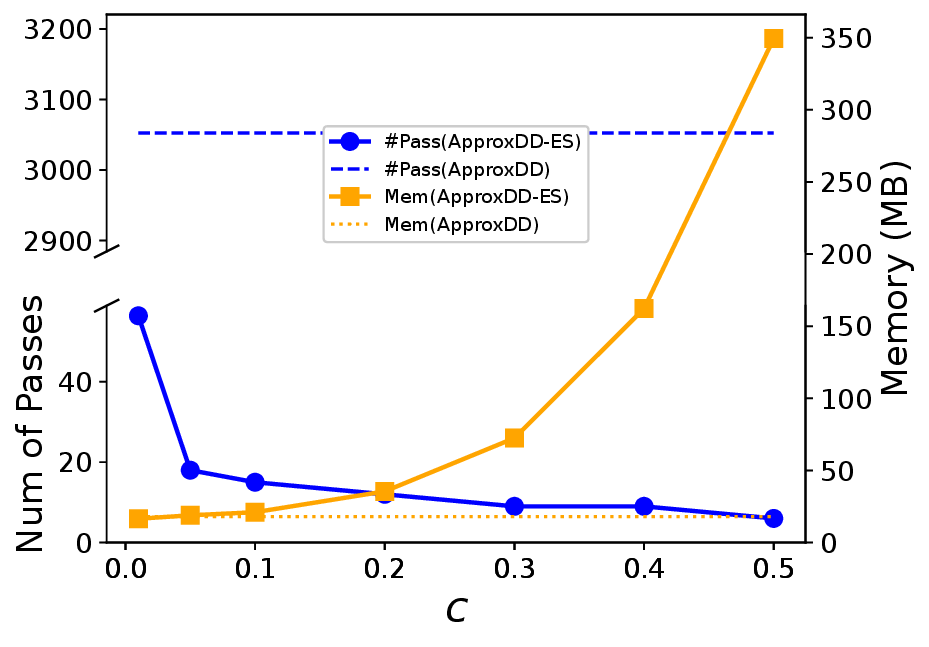}
    \caption{ER-4}
  \end{subfigure}
  \begin{subfigure}{0.32\textwidth}
    \includegraphics[width=\textwidth,trim=0 20 0 0]{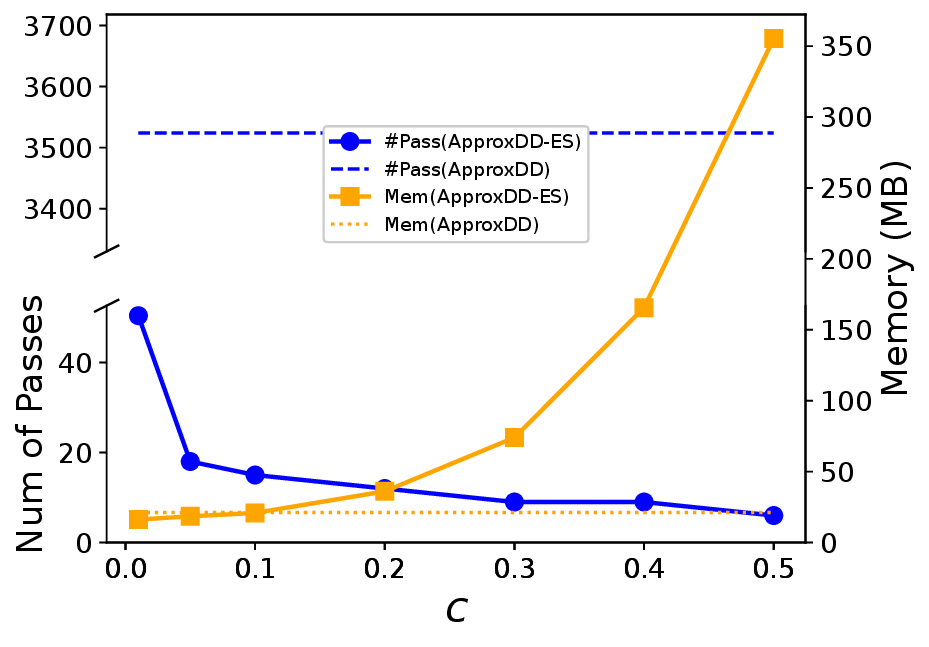}
    \caption{ER-5}
  \end{subfigure}
  \begin{subfigure}{0.32\textwidth}
    \includegraphics[width=\textwidth,trim=0 20 0 0]{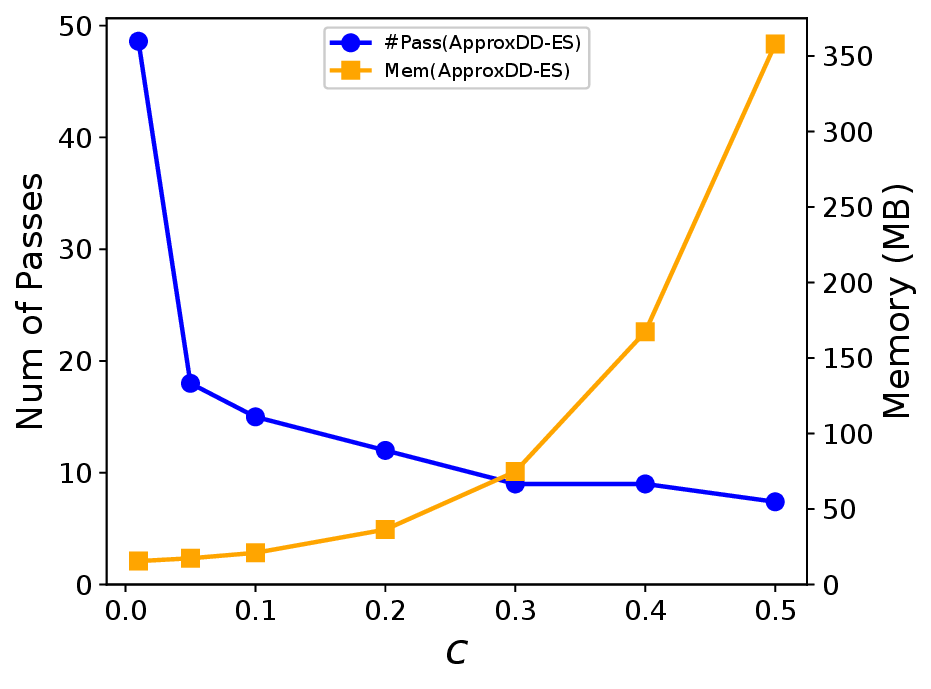}
    \caption{ER-6}\label{fig:c-er-6}
  \end{subfigure}
  \caption{(Continued) Number of passes and peak memory usage, as a function of $c$, for \ApproxDD and \ApproxDDb. Each point shown is the average of 5 executions. Missing points for \ApproxDD mean it did not terminate within 36 hours.}
\end{figure}

For both algorithms we implemented the heuristic in \cite[Appendix E]{Bourreau2024}, so as to improve the running time on large sparse graphs.
Roughly speaking, at the beginning of each iteration, the heuristic sorts the vertices in nonincreasing order of degree and selects the largest prefix whose induced graph fits in memory.
It then compares (i) running the next iteration of the peeling procedure (i.e., the while-loop in \Cref{alg:approxdd-fewer-2}) on this induced subgraph versus (ii) removing vertices in the prefix optimally, and keeps the option that leads to the largest drop in the current maximum degree.

We additionally evaluate the quality of the DD order produced by \ApproxDDb under different choices of $c$.
Given a graph $G=(V,E)$ and a vertex order $\prec$, for each vertex $v$ such that $d(v|G(v))>0$, define
\begin{equation*}
\vartheta_v := \frac{d(v|G(v))}{\Delta(G(v))}.
\end{equation*}

By definition, $\min_{v \in V \,:\, d(v|G(v))>0} \vartheta_v$ is exactly the largest value $\vartheta$ for which $\prec$ is a $\vartheta$-DD order.
To compute $\vartheta_v$ efficiently, write $\prec$ as $u_1,\ldots,u_n$ and index each edge $\{u_i,u_j\}$ ($i<j$) by $(i,j)$. Using an external sorting, we process edges in reverse-lexicographic order of $(i,j)$ while maintaining the current local degree of each vertex and the maximum local degree seen so far.
Whenever all edges of $G(v)$ have been processed for some $v$, we obtain $\vartheta_v = d(v|G(v))/\Delta(G(v))$.
For readability we report $\epsilon_v := \vartheta_v^{-1}-1$, i.e., $\prec$ is a $\frac{1}{1+\max_{v \in V \,:\, d(v|G(v))>0} \epsilon_v}$-DD order.

For each graph, we compute the empirical distribution of $\epsilon_v$ and take the average of 5 executions.

\noindentparagraph{Passes--Memory Tradeoff}
\Cref{fig:c} shows the passes--memory tradeoff of \ApproxDDb and compares it to \ApproxDD (which is insensitive to $c$).
Two remarks are in order.
The first one is that the results agree with \Cref{theorem:approxdd-fewer-2}: \ApproxDDb exhibits a clear memory-passes tradeoff, and the number of passes is very stable across different datasets (e.g., at most $25$ for $c=0.1$) in agreement with the $O(1/c)$ bound.
In contrast, the number of passes used by \ApproxDD goes with $\log n$, and therefore explodes on large datasets.

The second remark is that \ApproxDDb outperforms \ApproxDD. On most datasets, for $c=0.1$, \ApproxDDb uses at least 32\% fewer passes than \ApproxDD by using at most 45\% more memory. 
The advantage of \ApproxDDb is amplified on relatively dense graphs.
On the Similarity and ER families (see Figures \ref{fig:c-sim-1} to \ref{fig:c-er-6} and \Cref{fig:density}), increasing density causes \ApproxDD to require dramatically more passes (and to fail to terminate within 36 hours on Sim-6 and Sim-7), whereas \ApproxDDb remains within a small constant number of passes for $c=0.1$.
For example, on Sim-3 and Sim-5, \ApproxDDb for $c=0.1$ uses marginally more memory than \ApproxDD, and yet it makes $20\times$ and $110\times$ fewer passes.

\begin{figure}
\vspace{1em}
  \centering
  \begin{subfigure}{0.32\textwidth}
    \includegraphics[width=\textwidth]{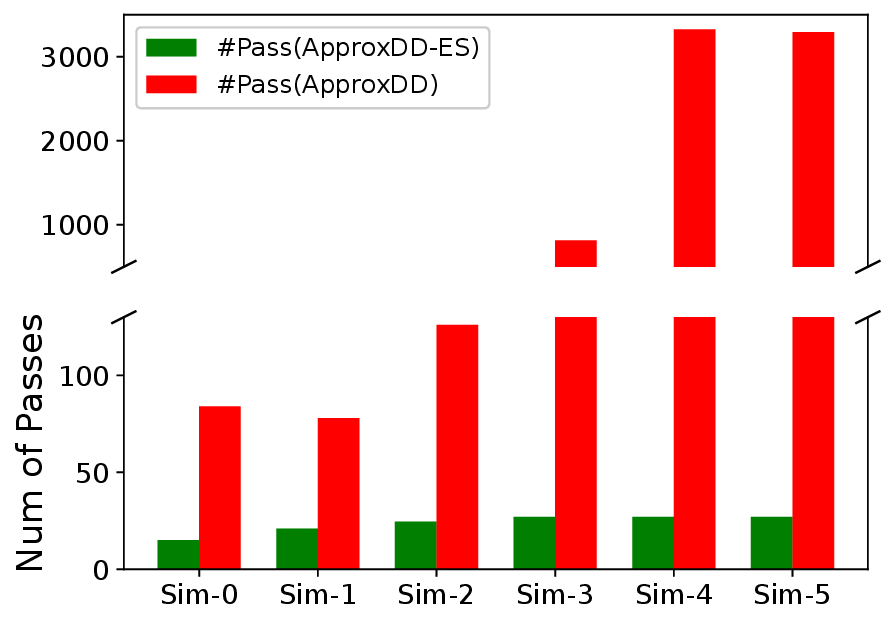}
  \end{subfigure}
  \hspace{5em}
  \begin{subfigure}{0.32\textwidth}
    \includegraphics[width=\textwidth]{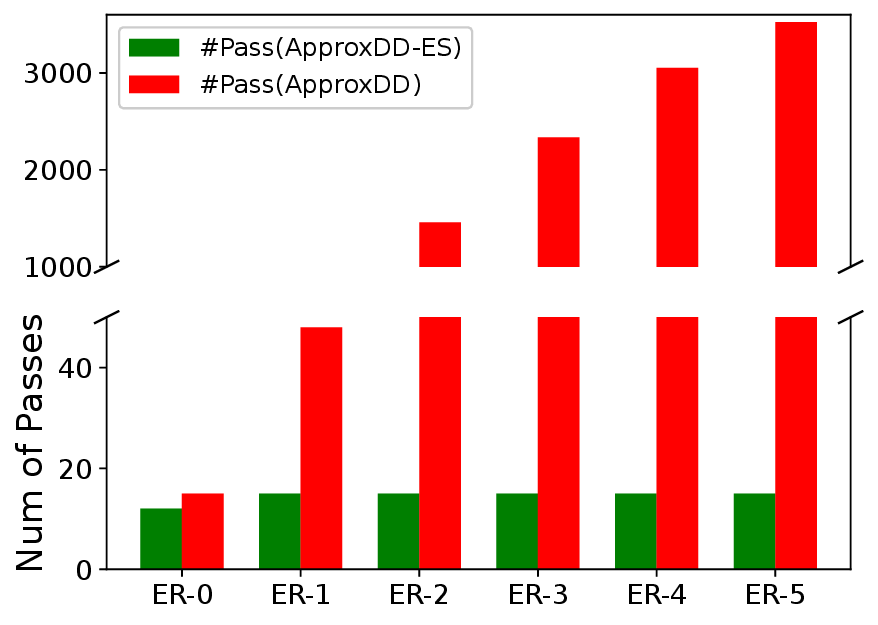}
  \end{subfigure}
  \caption{Number of passes on graphs of increasing densities using \ApproxDD and \ApproxDDb, when $c=0.1$.}\label{fig:density}
\end{figure}

Finally, \Cref{tab:mem-ddorder} reports the memory used for computing a $\frac{1}{1+\eps}$-DD order on all datasets when $c=0.1$.
Across datasets, \ApproxDDb uses moderately more memory than \ApproxDD, matching the tradeoff suggested by \Cref{thm:1}.

\begin{table}[ht]
  \caption{Memory used (in MB) by \ApproxDD and \ApproxDDb to compute a $\frac{1}{1+\eps}$-DD order, with $\eps=c=0.1$. The memory is expressed both in megabytes and (bracketed) as a fraction of the dataset size on disk (which uses 8 bytes per edge).}\label{tab:mem-ddorder}
  \centering
  \begin{tabular}{r|cc}
     \hline
     Dataset & \ApproxDD & \ApproxDDb \\
     \hline\hline
     NY Times & 34.96 (6.58\%) & 38.38 (7.22\%) \\
     Twitter(WWW) & 2432.10 (26.51\%) & 3177.75 (34.64\%) \\
     Twitter(MPI) & 2904.21 (23.58\%) & 3735.16 (30.33\%) \\
     Friendster & 3832.49 (27.72\%) & 5568.79 (40.29\%) \\
     Sim-0 & 17.10 (236.26\%) & 9.87 (136.36\%) \\
     Sim-1 & 17.37 (23.17\%) & 17.44 (23.27\%) \\
     Sim-2 & 20.64 (5.22\%) & 17.71 (4.48\%) \\
     Sim-3 & 16.42 (1.16\%) & 17.77 (1.26\%) \\
     Sim-4 & 19.14 (0.49\%) & 18.10 (0.47\%) \\
     Sim-5 & 15.86 (0.18\%) & 18.01 (0.21\%) \\
     Sim-6 & - & 17.84 (0.11\%) \\
     Sim-7 & - & 19.15 (0.07\%) \\
     ER-0 & 9.31 (128.54\%) & 7.76 (107.14\%) \\
     ER-1 & 16.98 (22.66\%) & 19.14 (25.54\%) \\
     ER-2 & 16.29 (4.12\%) & 19.07 (4.82\%) \\
     ER-3 & 17.43 (1.23\%) & 20.66 (1.46\%) \\
     ER-4 & 17.96 (0.46\%) & 21.03 (0.54\%) \\
     ER-5 & 21.16 (0.24\%) & 20.86 (0.24\%) \\
     ER-6 & - & 21.00 (0.13\%) \\
     \hline
   \end{tabular}
\end{table}

\smallskip
\noindentparagraph{DD Order Quality}
\Cref{fig:ddeval} shows the empirical distribution of $\epsilon_v$.
On large sparse graphs (NY Times, Twitter(WWW), Friendster), \ApproxDDb typically outputs orders that are extremely close to a perfect DD order: for small and moderate $c$, the mass is concentrated in the smallest $\epsilon$ bin (near $0$).
On dense graphs (ER and Sim), the distribution remains tightly concentrated on small $\epsilon$ (e.g., mostly $\epsilon \le 0.4$ for ER and mostly $\epsilon \le 4$ for Sim). 
Overall, these results suggest that \ApproxDDb produces high-quality DD orders in practice.

\begin{figure}[h!]
  \centering
  \begin{subfigure}{0.46\textwidth}
    \includegraphics[width=\textwidth]{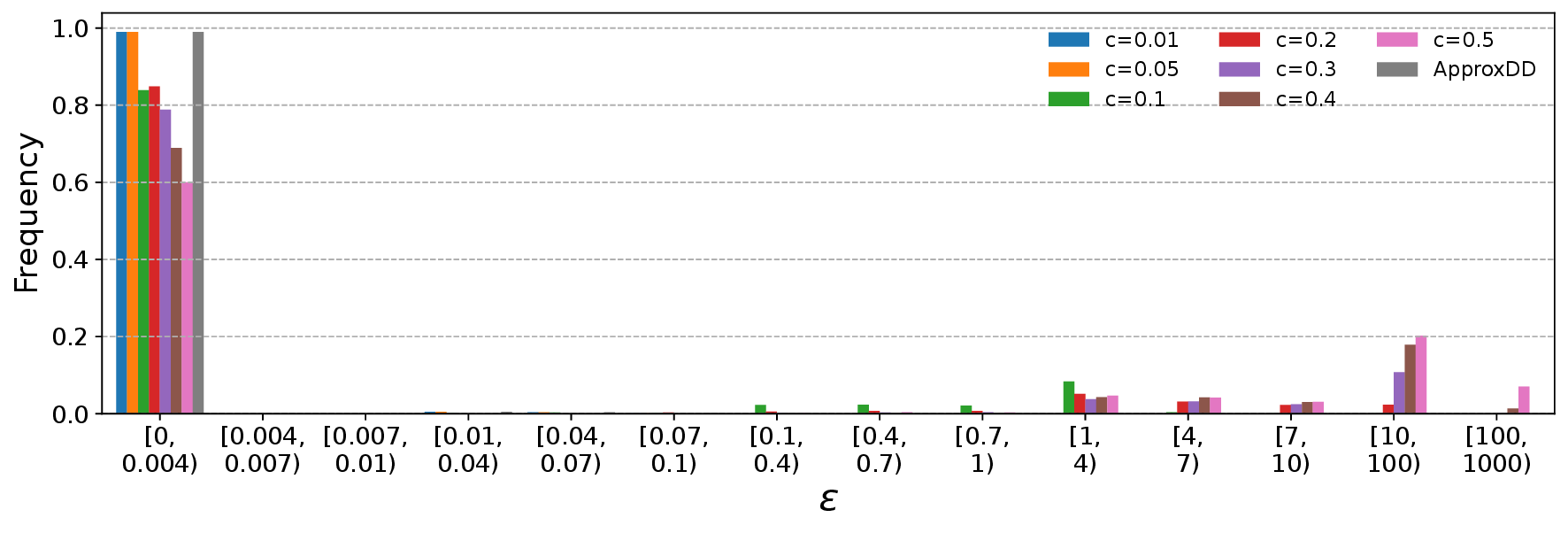}
    \caption{NY Times}
  \end{subfigure}
  \begin{subfigure}{0.46\textwidth}
    \includegraphics[width=\textwidth]{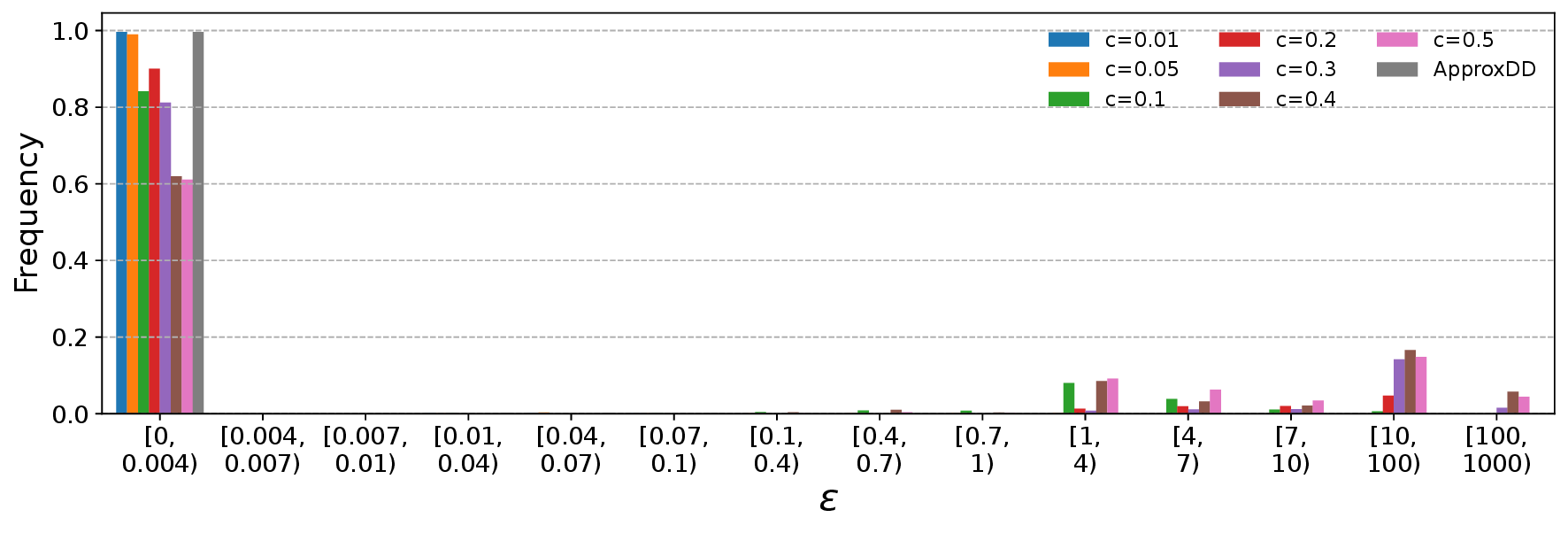}
    \caption{Twitter(WWW)}
  \end{subfigure}
  \begin{subfigure}{0.46\textwidth}
    \includegraphics[width=\textwidth]{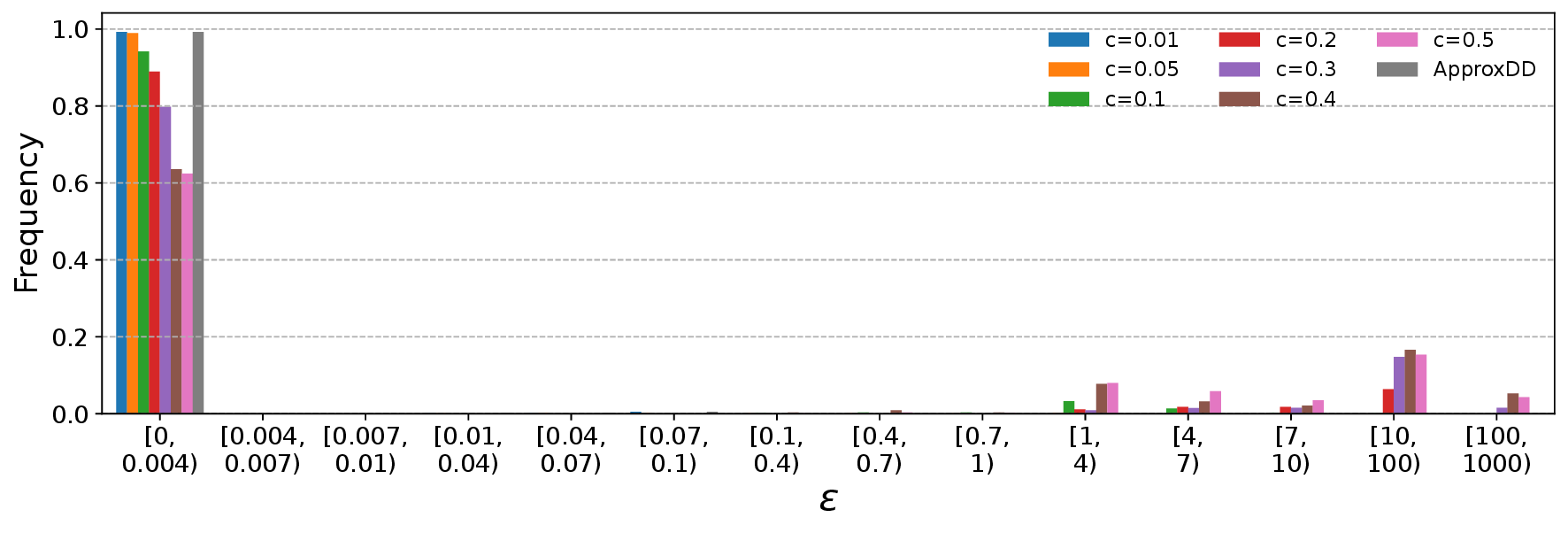}
    \caption{Twitter(MPI)}
  \end{subfigure}
  \begin{subfigure}{0.46\textwidth}
    \includegraphics[width=\textwidth]{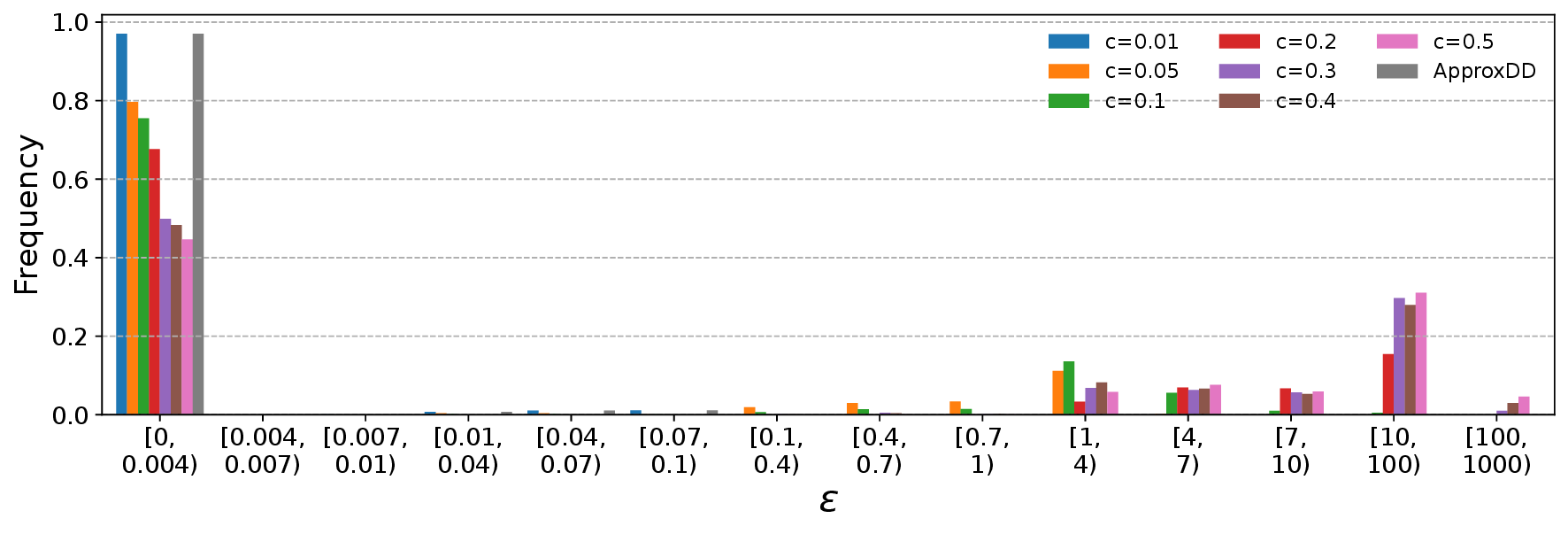}
    \caption{Friendster}
  \end{subfigure}
  \begin{subfigure}{0.46\textwidth}
    \includegraphics[width=\textwidth]{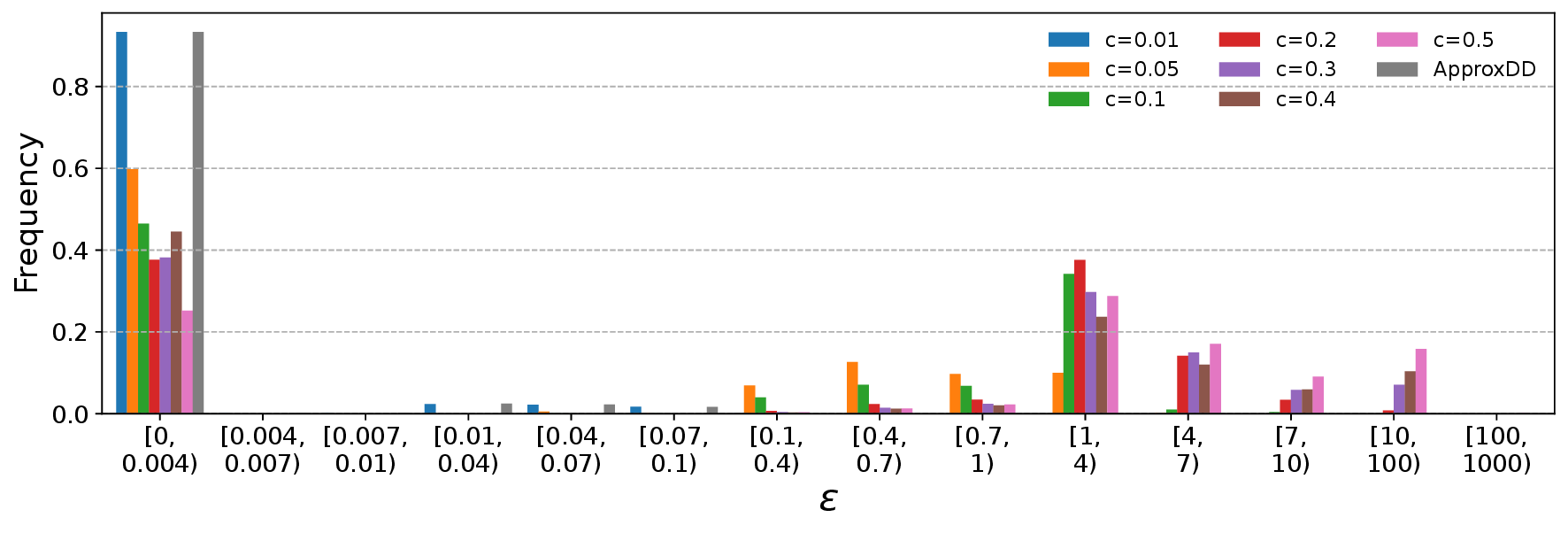}
    \caption{Sim-1}
  \end{subfigure}
  \begin{subfigure}{0.46\textwidth}
    \includegraphics[width=\textwidth]{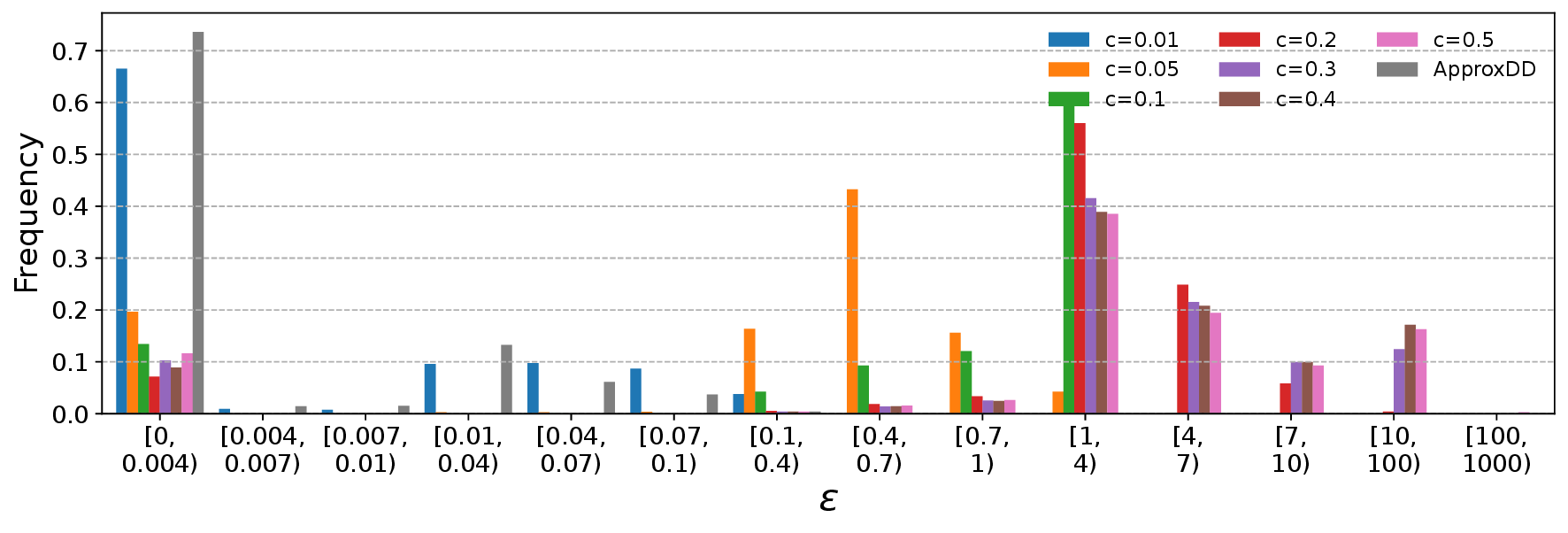}
    \caption{Sim-3}
  \end{subfigure}
  \begin{subfigure}{0.46\textwidth}
    \includegraphics[width=\textwidth]{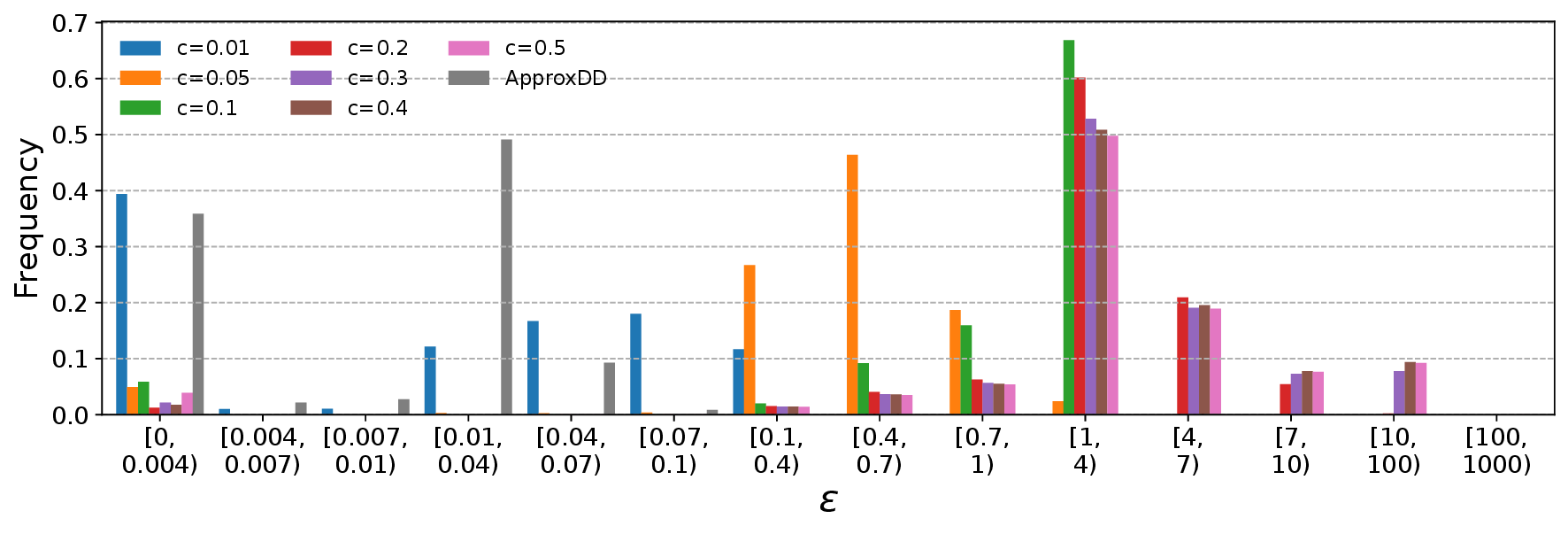}
    \caption{Sim-5}
  \end{subfigure}
  \begin{subfigure}{0.46\textwidth}
    \includegraphics[width=\textwidth]{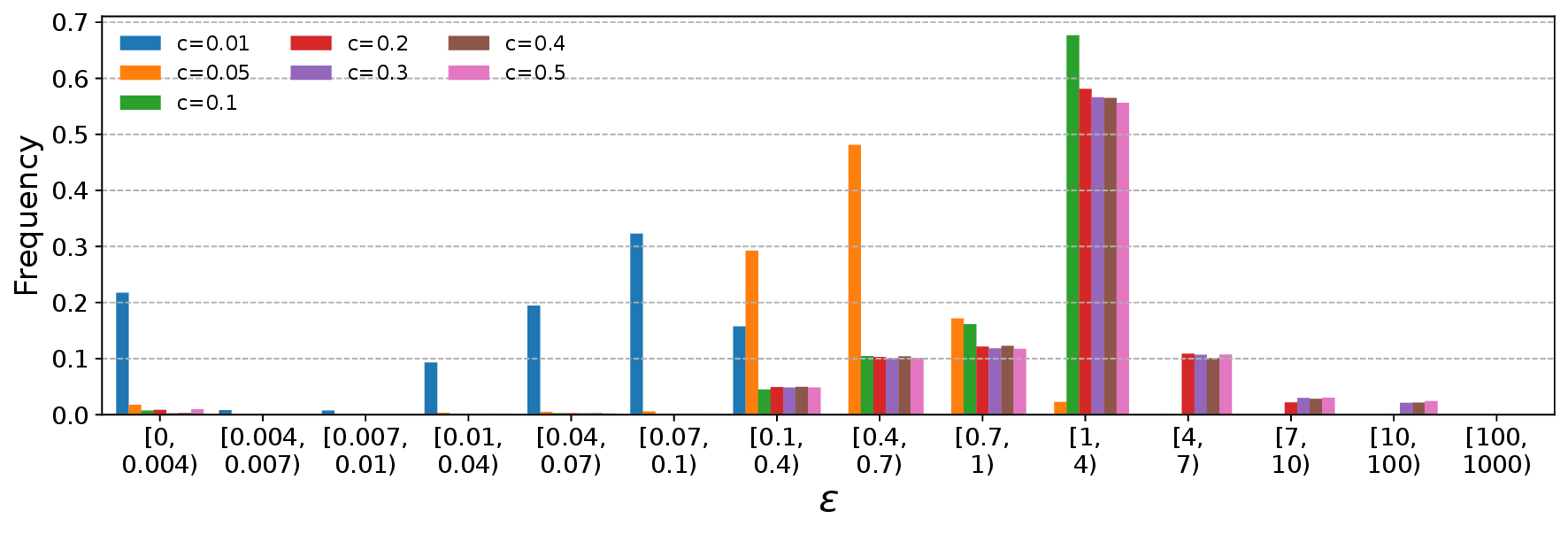}
    \caption{Sim-7}
  \end{subfigure}
  \begin{subfigure}{0.46\textwidth}
    \includegraphics[width=\textwidth]{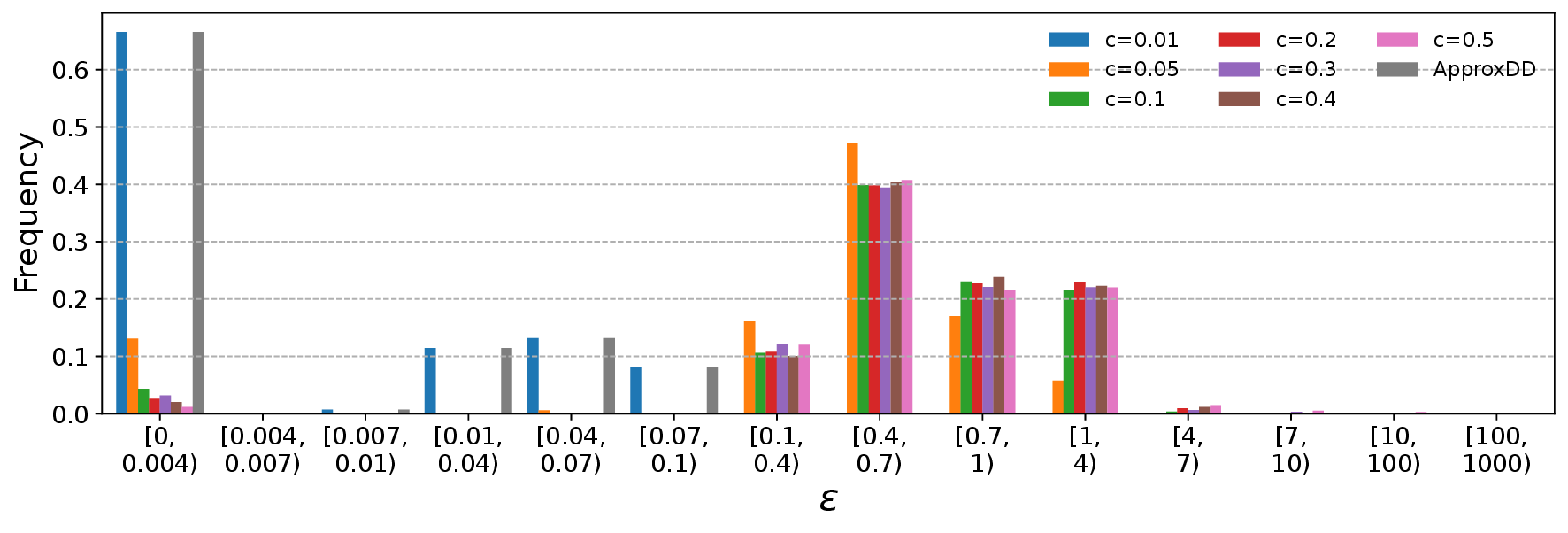}
    \caption{ER-1}
  \end{subfigure}
  \begin{subfigure}{0.46\textwidth}
    \includegraphics[width=\textwidth]{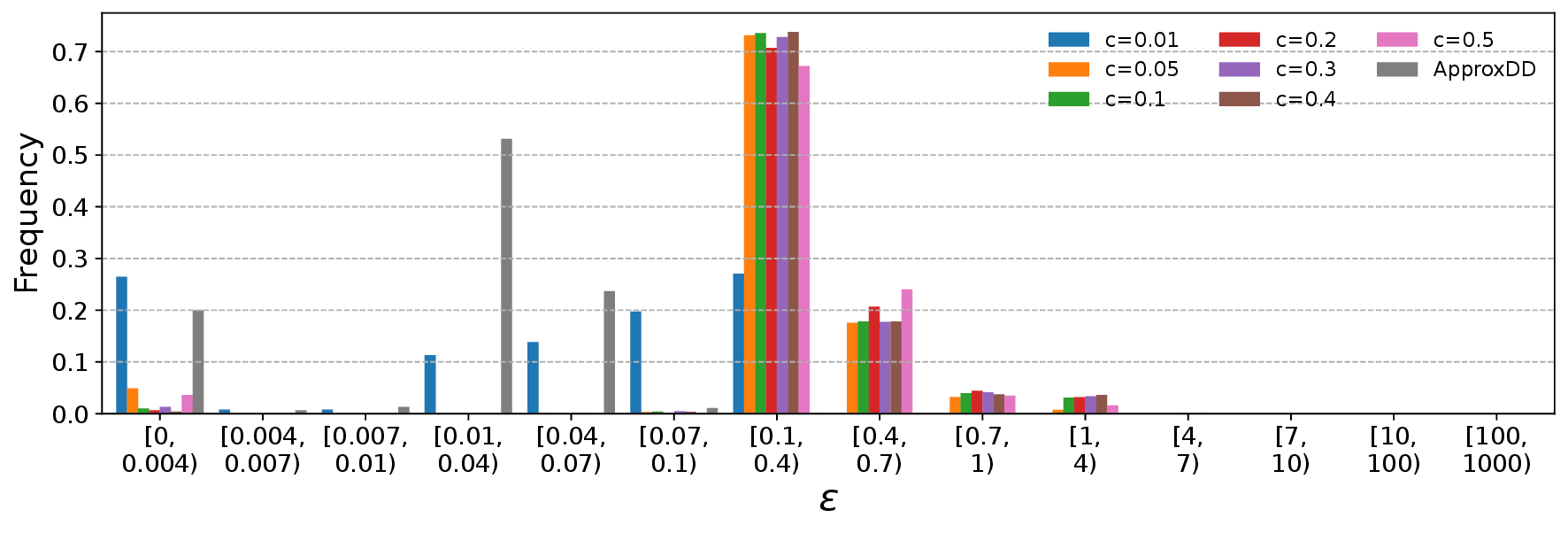}
    \caption{ER-2}
  \end{subfigure}
  \begin{subfigure}{0.46\textwidth}
    \includegraphics[width=\textwidth]{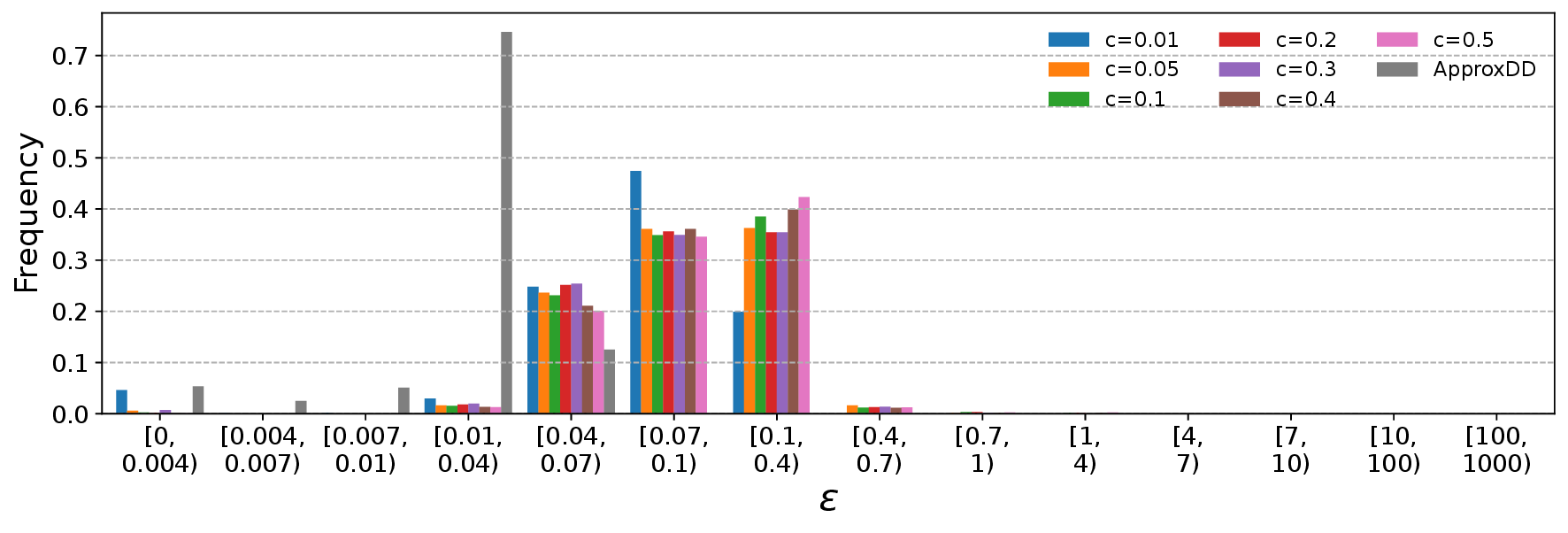}
    \caption{ER-4}
  \end{subfigure}
  \begin{subfigure}{0.46\textwidth}
    \includegraphics[width=\textwidth]{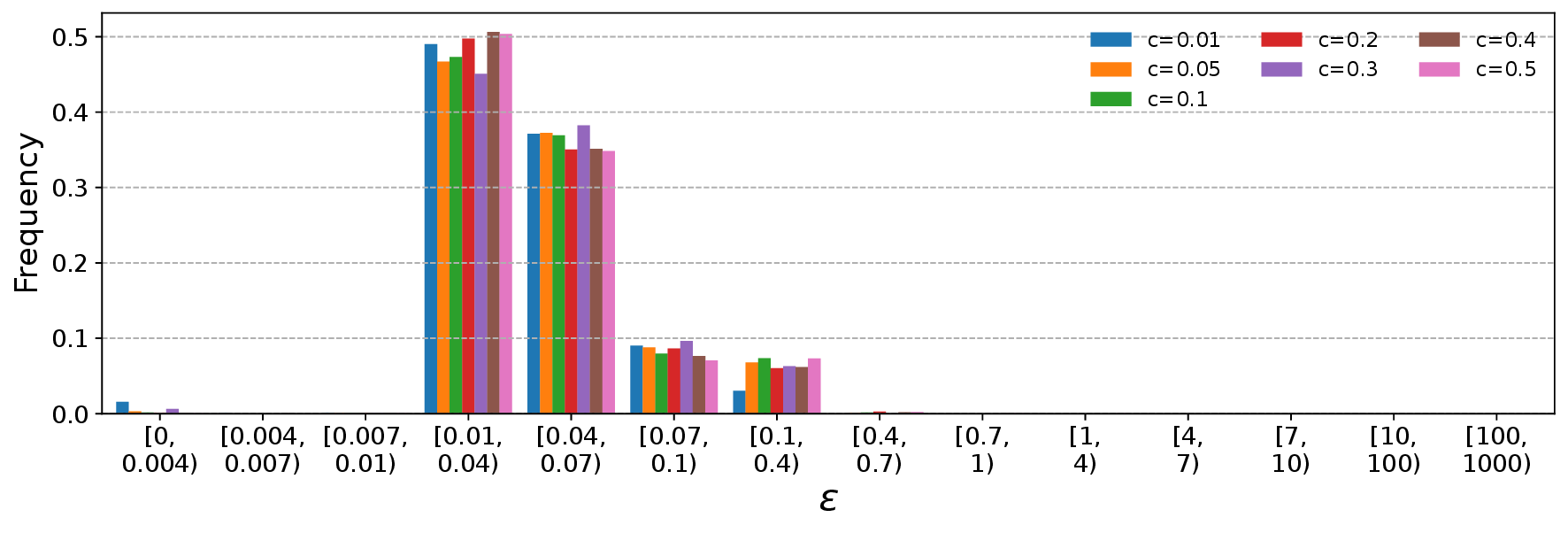}
    \caption{ER-6}
  \end{subfigure}
  \caption{Empirical distribution of $\epsilon_v$ under different $c$, using \ApproxDDb.}\label{fig:ddeval}
\end{figure}

\subsection{Computing the Graphlet Distribution}
\label{exp:distr}

We test the performance of our full algorithm for approximating the $k$-graphlet distribution, \ApproxDDb+\gd (\Cref{alg:full}), and of its competitor \ApproxDD+\gd (the algorithm of~\cite{Bourreau2024}).
Note that the algorithm originally described in~\cite{Bourreau2024} actually uses \emph{rejection}, rather than using \emph{counters}.
We also test this ``default'' competitor, \ApproxDD+\gdr, where \gdr is introduced in \Cref{alg:gdr}.
As expected, \ApproxDD+\gdr is much more inefficient than the ``counter'' variants, as it gains information only upon accepting a graphlet, which happens with probability only $k^{-O(k)}$.

We ran the algorithms for $k=4,5,6$, using as ground truth the average of 5 runs of Motivo \cite{Bressan2019-VLDB}.
We set the parameters of our algorithms to $\eps=0.1$, $c=0.1$, and $\delta=0.02$.
For each algorithm and each dataset, we start by computing a $\frac{1}{1+\eps}$-DD order.
We then ran 5 batches of parallel sampling.
After each batch, we computed the L$_\infty$ distance between (i) the ground truth and (ii) the $k$-graphlet distribution estimate obtained so far.
This experiment was then repeated 5 times, and we took the average.

\smallskip
\noindentparagraph{Accuracy versus passes}
\Cref{fig:Linfty} shows the number of passes versus L$_\infty$ distance achieved by the algorithms.
(Due to space limitations we show the plots only for a representative subset of datasets).
Note that the number of passes is the sum of the passes taken by preprocessing and of $2k-1$ passes for each sampling batch.
The first observation is that, on every dataset and for every value of $k$, save a few exceptions at $k=4$, \ApproxDDb+\gd and \ApproxDD+\gd yield an error smaller than \ApproxDD+\gdr.
The gap between the two errors increases with $k$, with \ApproxDD+\gdr lagging much behind \ApproxDDb+\gd and \ApproxDD+\gd for $k=6$.
This is explained by the aforementioned $k^{-O(k)}$ acceptance probability of \ApproxDD+\gdr, and confirms that the ``counter'' approach is significantly better than the ``rejection'' approach.
The second observation is that, again, \ApproxDDb+\gd uses fewer passes than \ApproxDD+\gd, outperforming it by orders of magnitude on dense graphs.
Note that the sampling phase is identical between the two algorithms (but is ran with different DD-orders), and seems to have a relatively small impact on the error-vs-passes tradeoff, if compared with the preprocessing phase.
This is coherent with the results of \Cref{exp:DD} and with our theoretical bounds.
Overall, \ApproxDDb+\gd ensures an L$_\infty$ error of at most $0.01$ with $<50$ passes for $k=4$, of at most $0.02$ with $<60$ passes for $k=5$, and at most $0.05$ with $<80$ passes for $k=6$.

\begin{figure}
  \centering
  \begin{subfigure}{0.32\textwidth}
    \includegraphics[width=\textwidth,trim=21 25 0 0,clip]{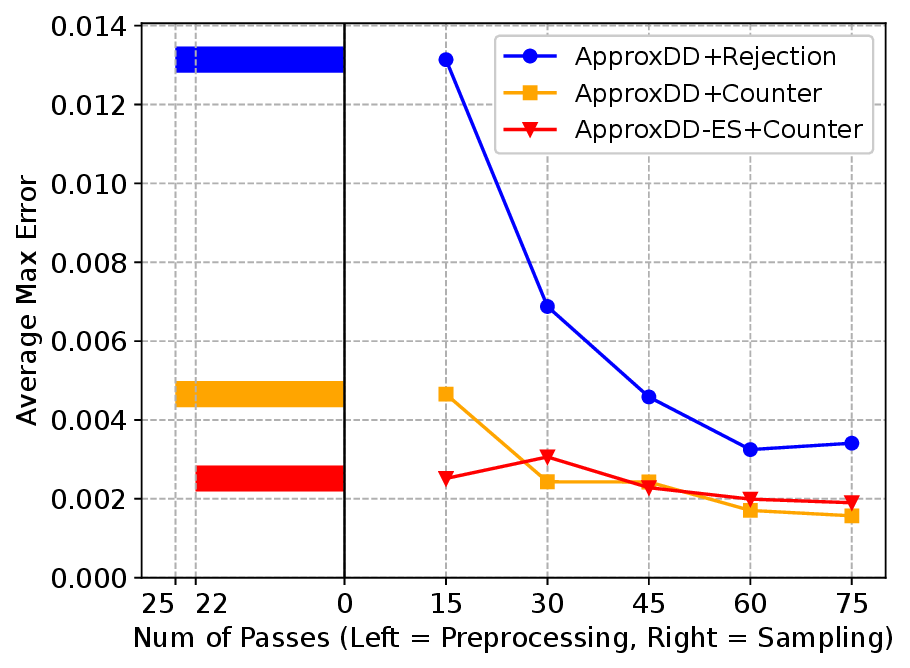}
    \caption{NY Times, $k=4$}
  \end{subfigure}
  \begin{subfigure}{0.32\textwidth}
    \includegraphics[width=\textwidth,trim=21 25 0 0,clip]{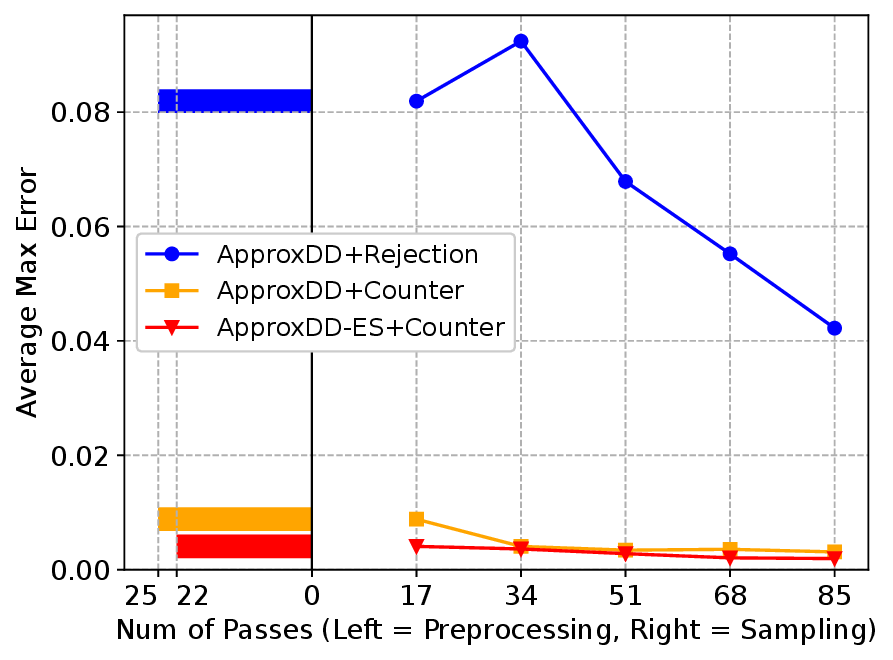}
    \caption{NY Times, $k=5$}
  \end{subfigure}
  \begin{subfigure}{0.32\textwidth}
    \includegraphics[width=\textwidth,trim=21 25 0 0,clip]{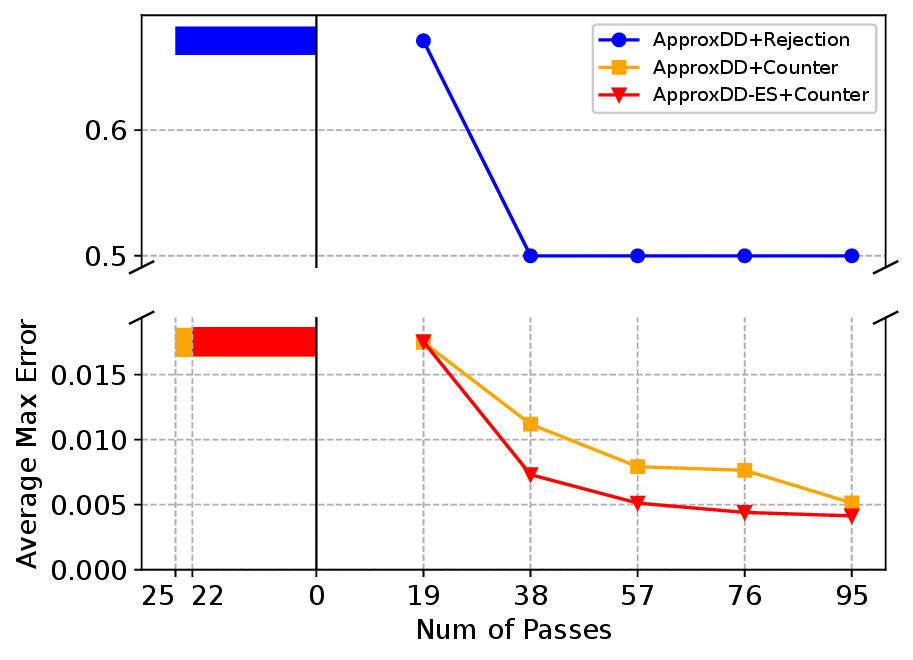}
    \caption{NY Times, $k=6$}
  \end{subfigure}
  \begin{subfigure}{0.32\textwidth}
    \includegraphics[width=\textwidth,trim=21 25 0 0,clip]{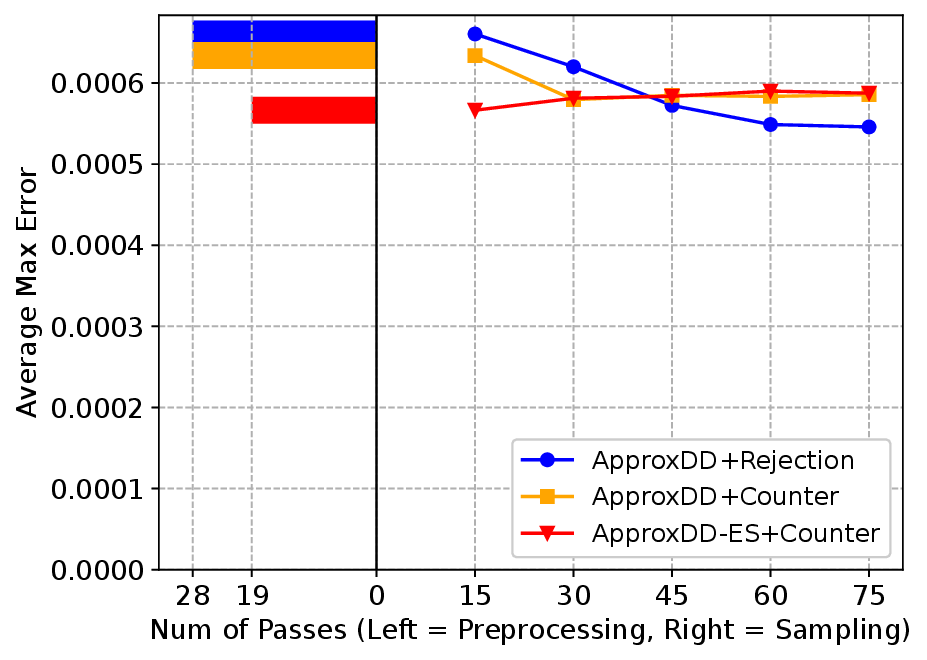}
    \caption{Twitter(WWW), $k=4$}
  \end{subfigure}
  \begin{subfigure}{0.32\textwidth}
    \includegraphics[width=\textwidth,trim=21 25 0 0,clip]{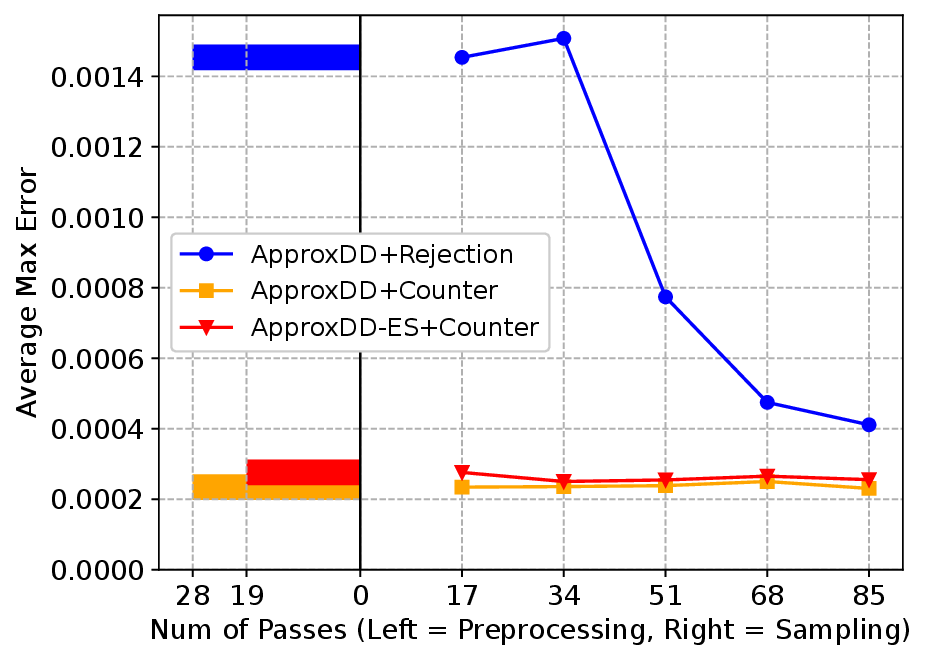}
    \caption{Twitter(WWW), $k=5$}
  \end{subfigure}
  \begin{subfigure}{0.32\textwidth}
    \includegraphics[width=\textwidth,trim=21 25 0 0,clip]{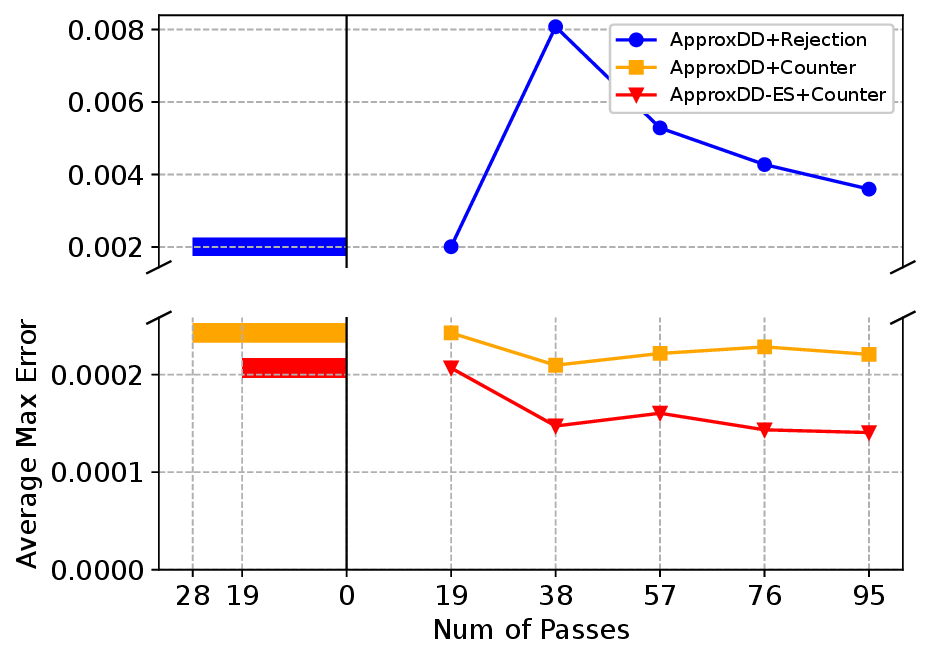}
    \caption{Twitter(WWW), $k=6$}
  \end{subfigure}
  \begin{subfigure}{0.32\textwidth}
    \includegraphics[width=\textwidth,trim=21 25 0 0,clip]{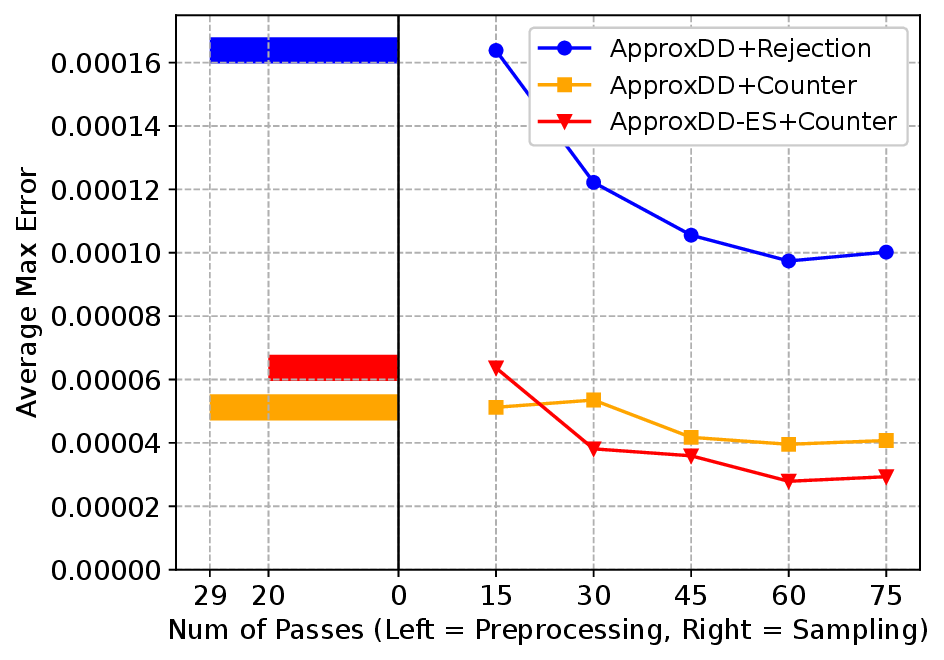}
    \caption{Twitter(MPI), $k=4$}
  \end{subfigure}
  \begin{subfigure}{0.32\textwidth}
    \includegraphics[width=\textwidth,trim=21 25 0 0,clip]{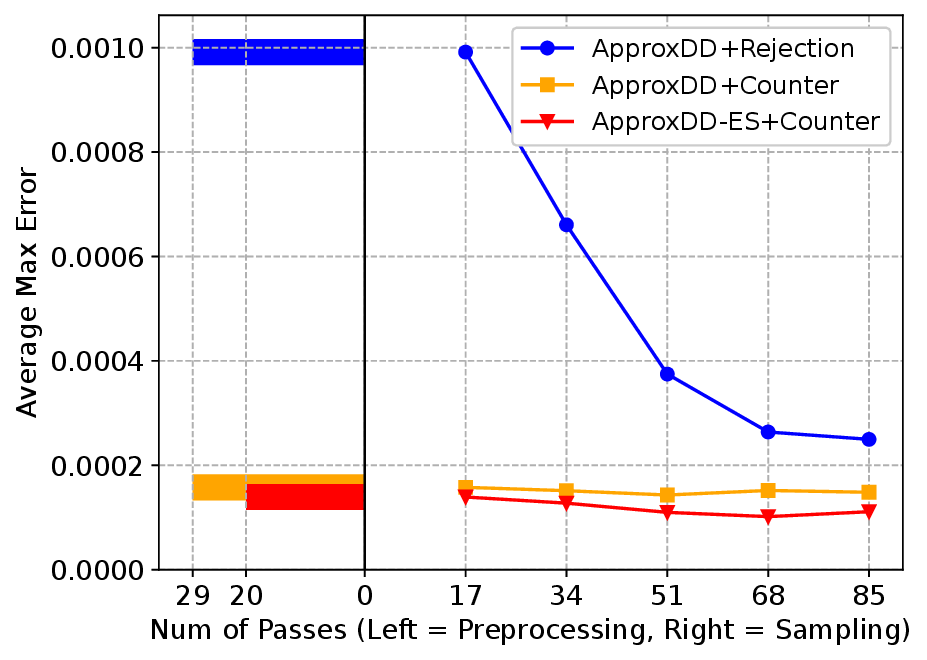}
    \caption{Twitter(MPI), $k=5$}
  \end{subfigure}
  \begin{subfigure}{0.32\textwidth}
    \includegraphics[width=\textwidth,trim=21 25 0 0,clip]{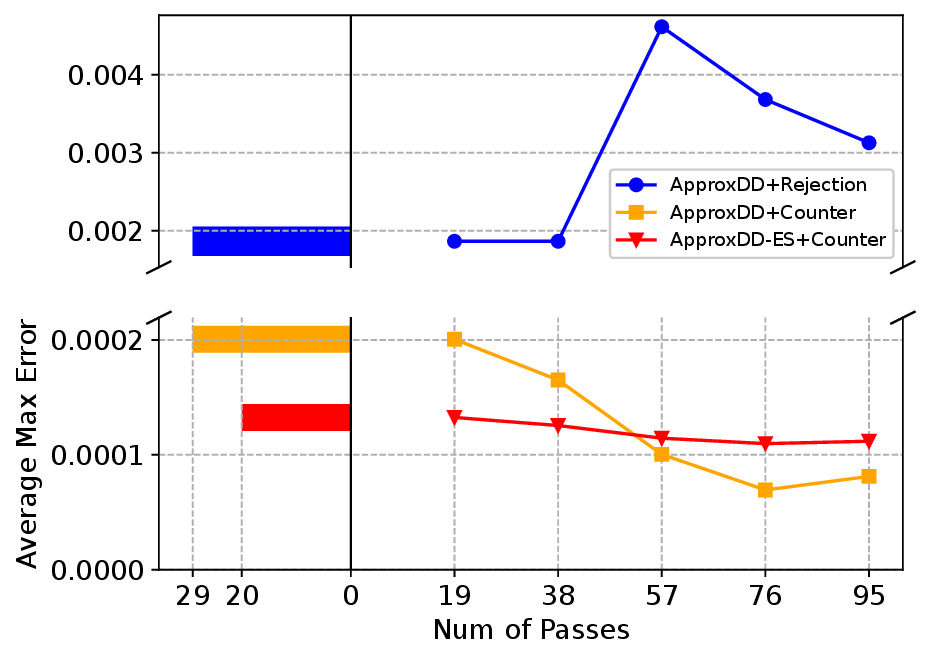}
    \caption{Twitter(MPI), $k=6$}
  \end{subfigure}
  \begin{subfigure}{0.32\textwidth}
    \includegraphics[width=\textwidth,trim=21 25 0 0,clip]{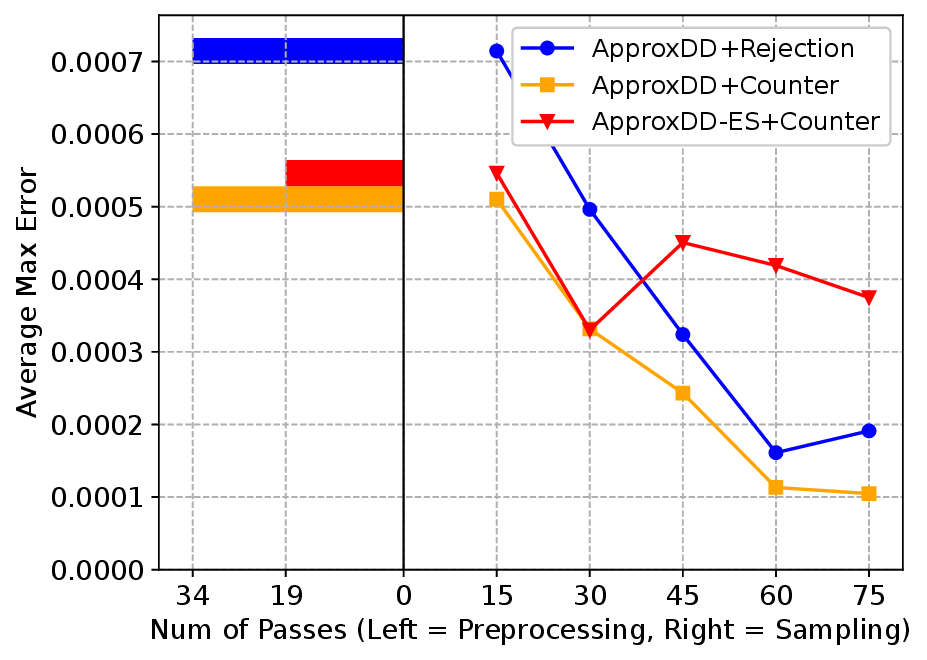}
    \caption{Friendster, $k=4$}
  \end{subfigure}
  \begin{subfigure}{0.32\textwidth}
    \includegraphics[width=\textwidth,trim=21 25 0 0,clip]{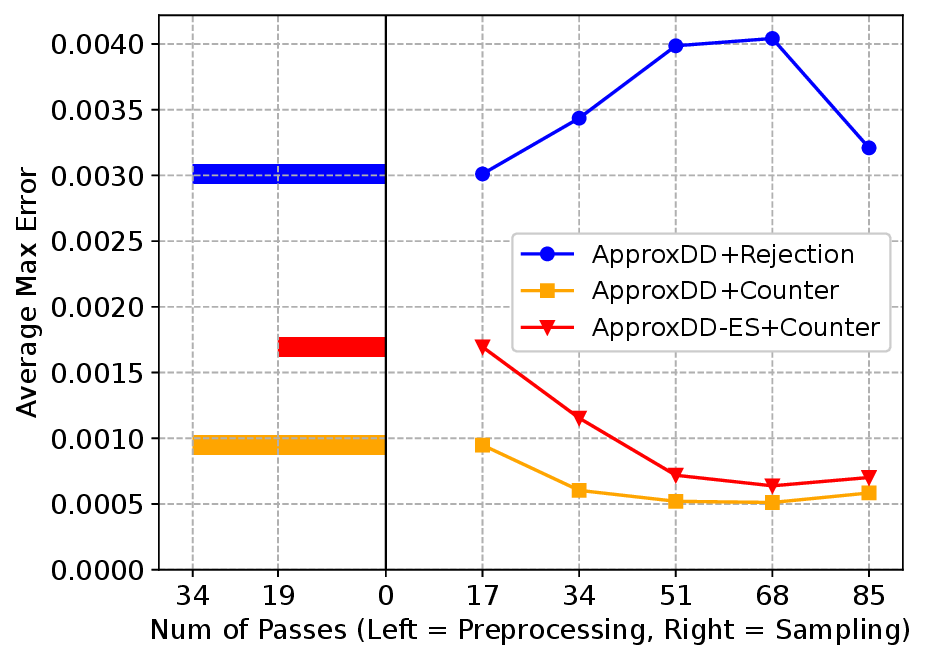}
    \caption{Friendster, $k=5$}
  \end{subfigure}
  \begin{subfigure}{0.32\textwidth}
    \includegraphics[width=\textwidth,trim=21 25 0 0,clip]{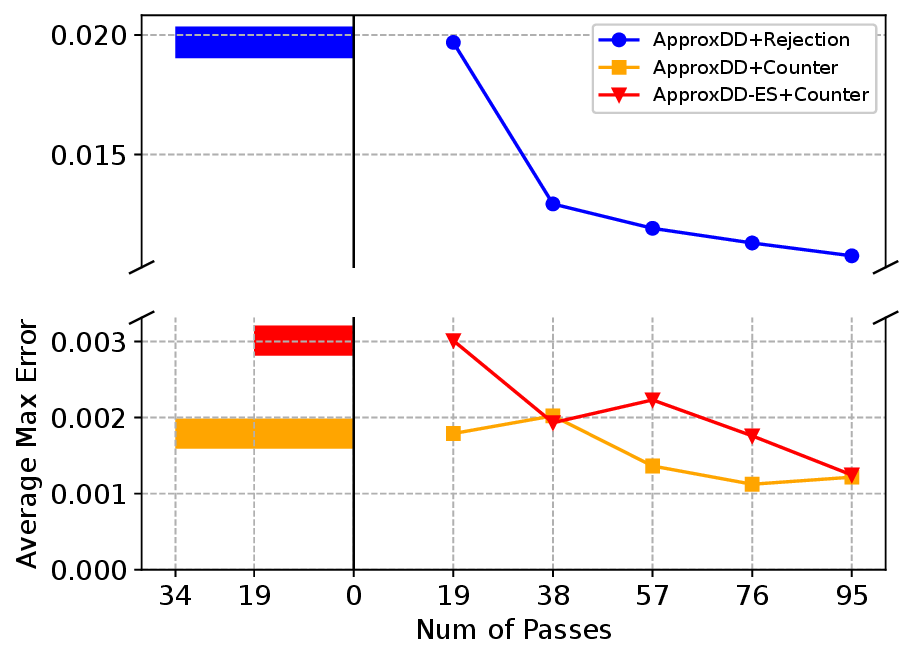}
    \caption{Friendster, $k=6$}
  \end{subfigure}
  \begin{subfigure}{0.32\textwidth}
    \includegraphics[width=\textwidth,trim=21 25 0 0,clip]{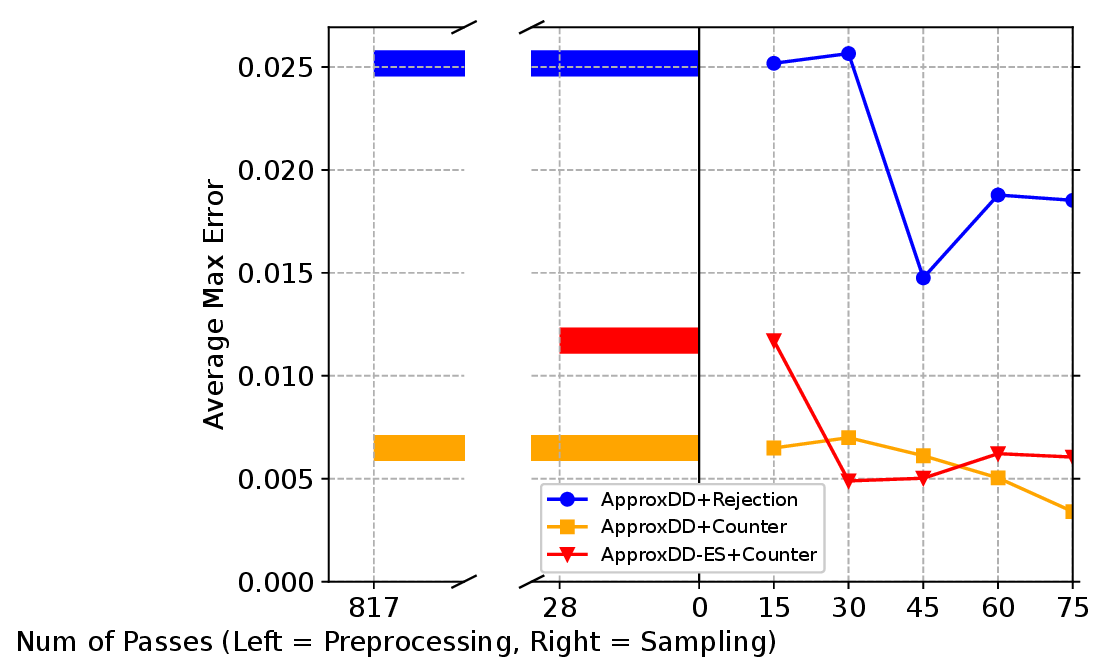}
    \caption{Sim-3, $k=4$}
  \end{subfigure}
  \begin{subfigure}{0.32\textwidth}
    \includegraphics[width=\textwidth,trim=21 25 0 0,clip]{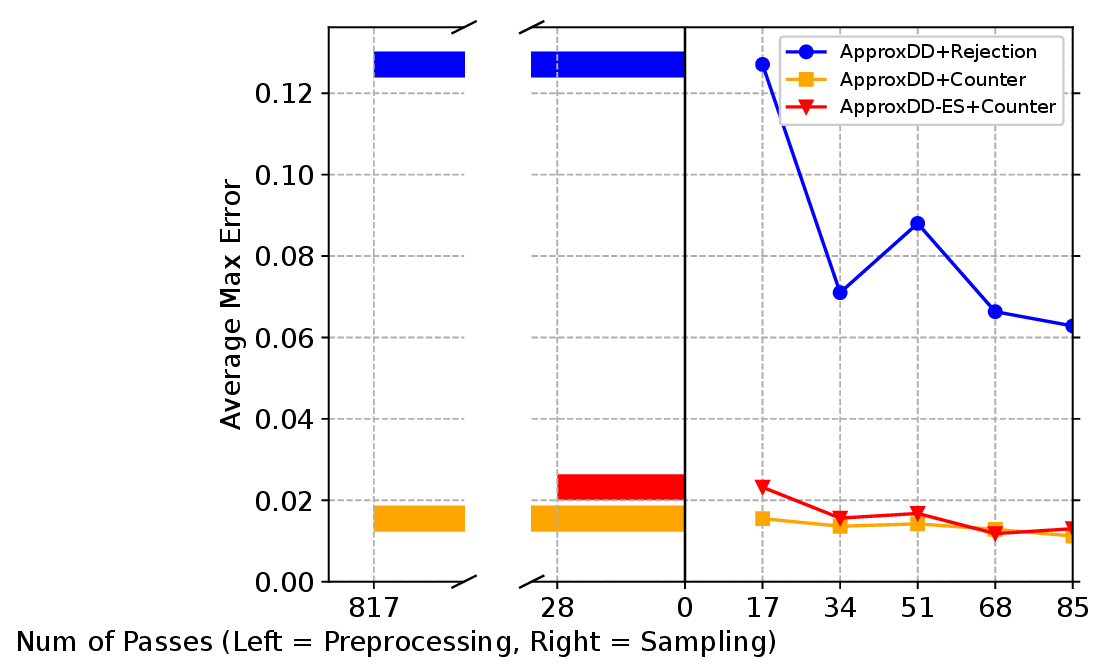}
    \caption{Sim-3, $k=5$}
  \end{subfigure}
  \begin{subfigure}{0.32\textwidth}
    \includegraphics[width=\textwidth,trim=21 25 0 0,clip]{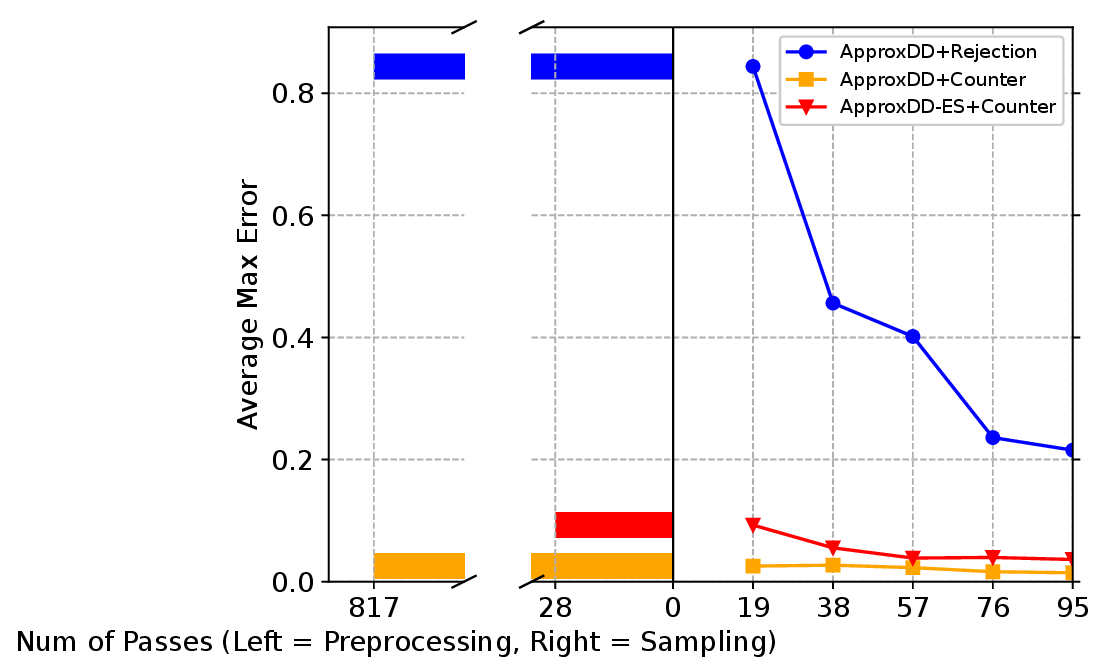}
    \caption{Sim-3, $k=6$}
  \end{subfigure}
    \caption{L$_{\infty}$ distance between the ground-truth and the estimated $k$-graphlet distribution as a function of the number of passes for sampling. The X-axis shows the number of passes (preprocessing on the left, sampling on the right). Missing points for \ApproxDD mean it did not terminate within 36 hours.}
    \label{fig:Linfty}
\end{figure}

\begin{figure}\ContinuedFloat
  \begin{subfigure}{0.32\textwidth}
    \includegraphics[width=\textwidth,trim=21 25 0 0,clip]{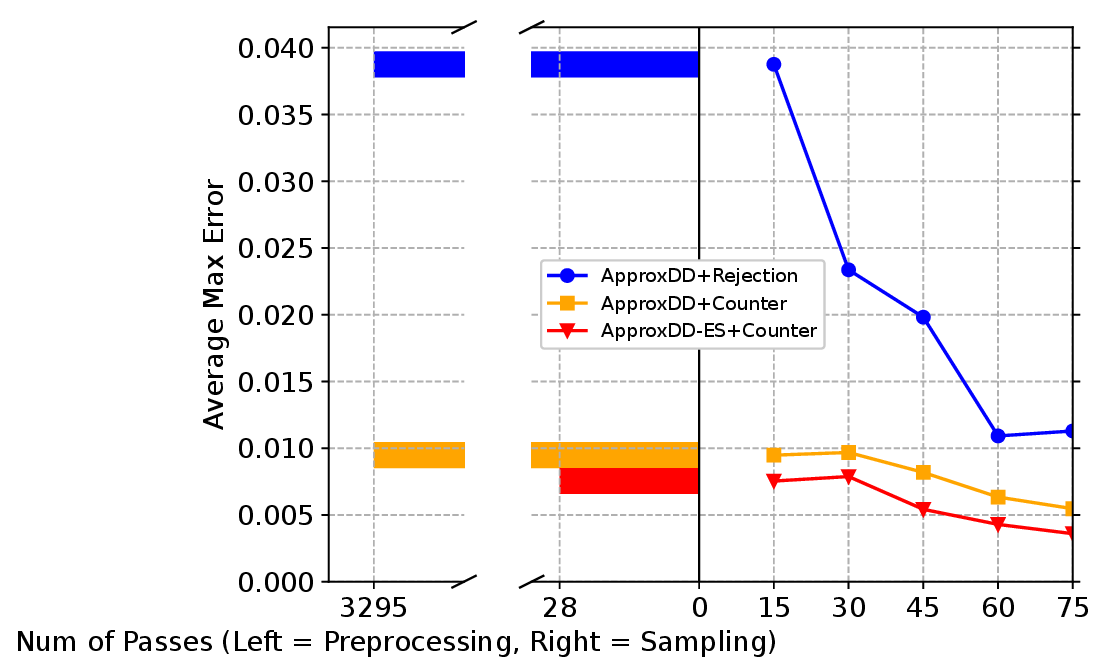}
    \caption{Sim-5, $k=4$}
  \end{subfigure}
  \begin{subfigure}{0.32\textwidth}
    \includegraphics[width=\textwidth,trim=21 25 0 0,clip]{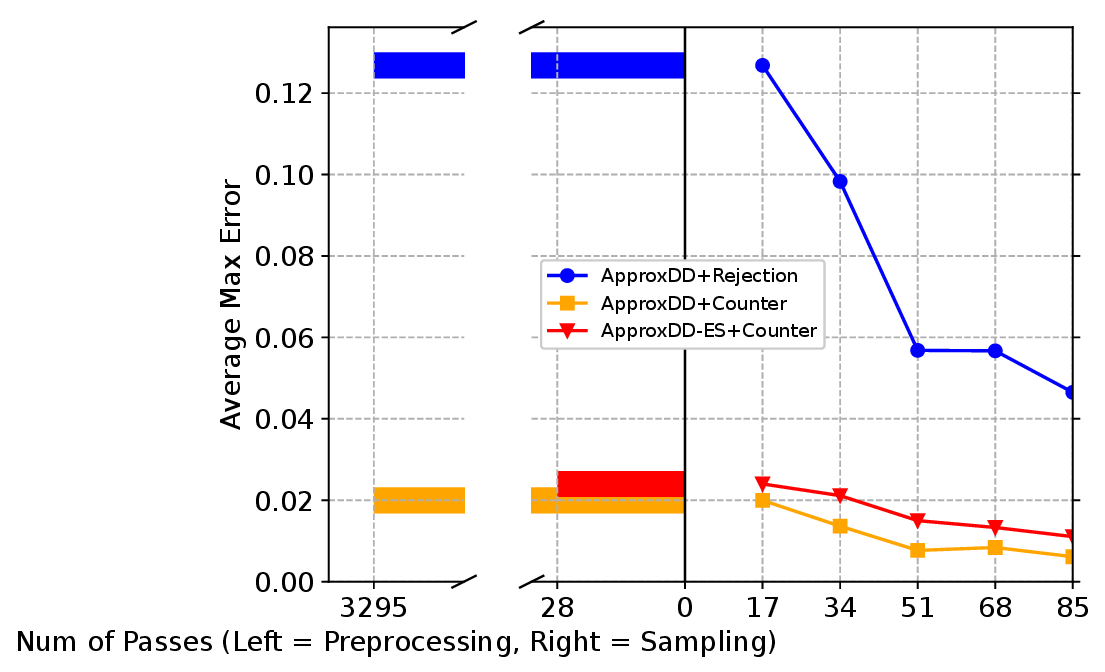}
    \caption{Sim-5, $k=5$}
  \end{subfigure}
  \begin{subfigure}{0.32\textwidth}
    \includegraphics[width=\textwidth,trim=21 25 0 0,clip]{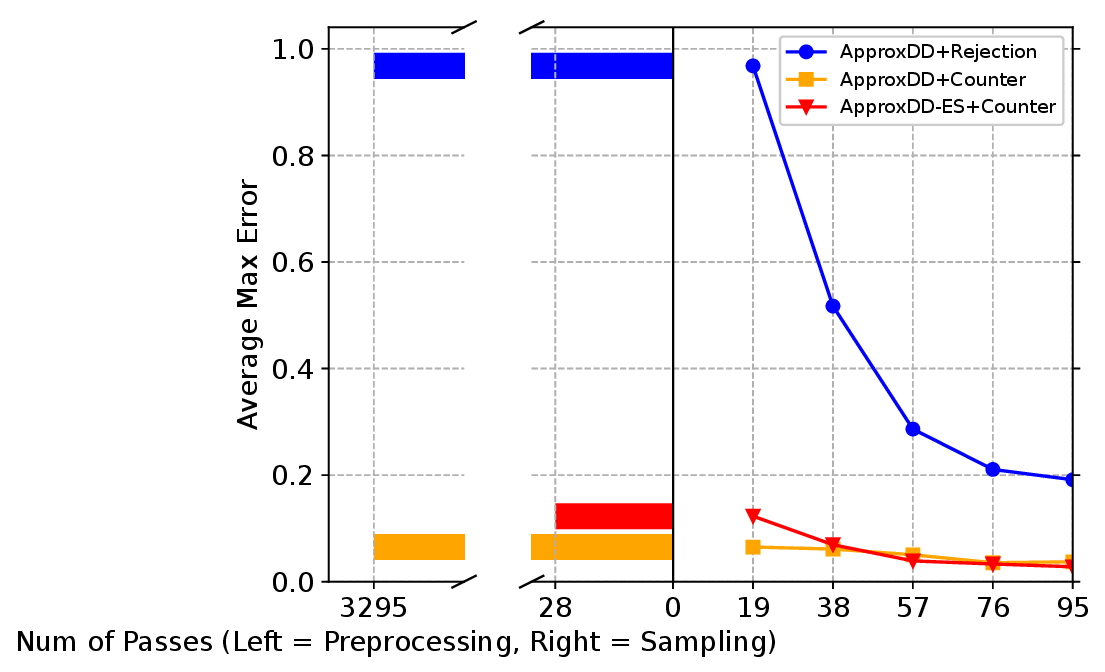}
    \caption{Sim-5, $k=6$}
  \end{subfigure}
  \begin{subfigure}{0.32\textwidth}
    \includegraphics[width=\textwidth,trim=21 25 0 0,clip]{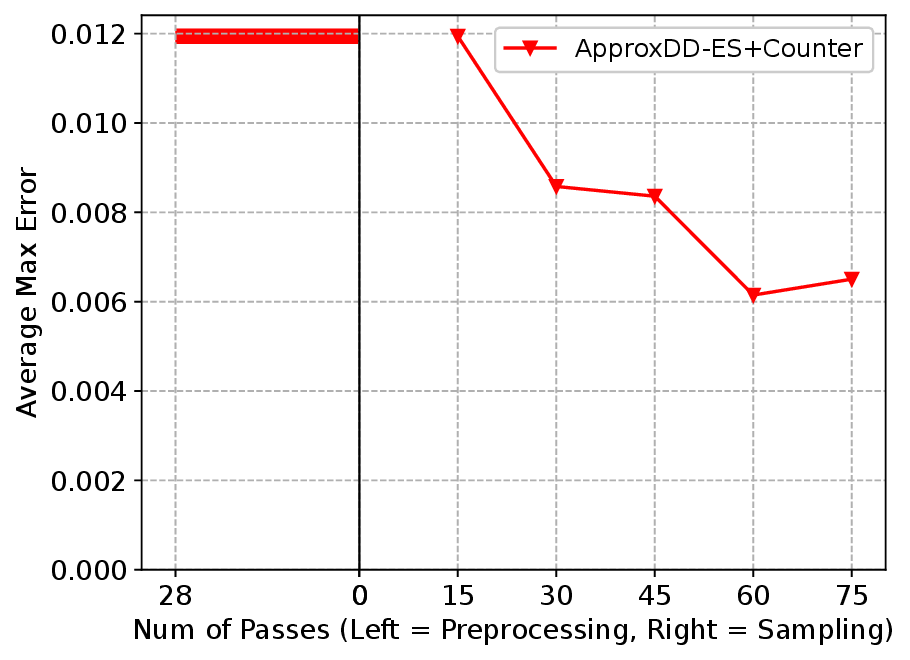}
    \caption{Sim-7, $k=4$}
  \end{subfigure}
  \begin{subfigure}{0.32\textwidth}
    \includegraphics[width=\textwidth,trim=21 25 0 0,clip]{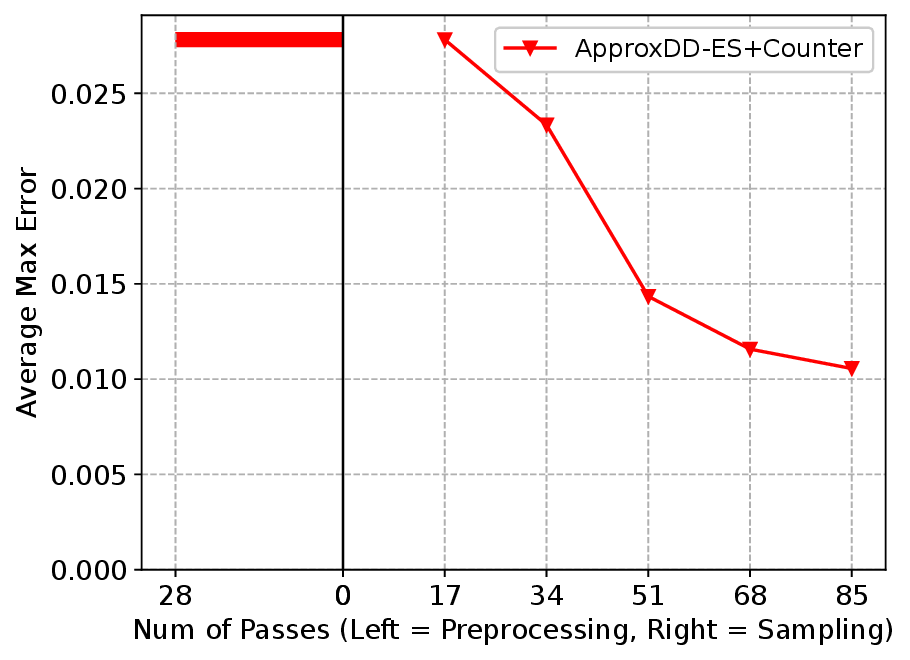}
    \caption{Sim-7, $k=5$}
  \end{subfigure}
  \begin{subfigure}{0.32\textwidth}
    \includegraphics[width=\textwidth,trim=21 25 0 0,clip]{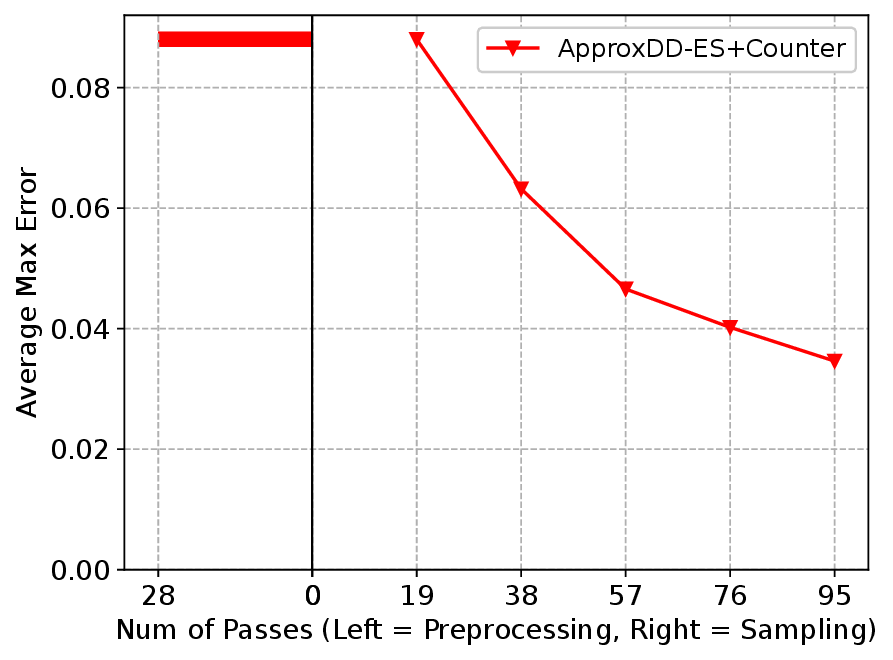}
    \caption{Sim-7, $k=6$}
  \end{subfigure}
  \begin{subfigure}{0.32\textwidth}
    \includegraphics[width=\textwidth,trim=21 25 0 0,clip]{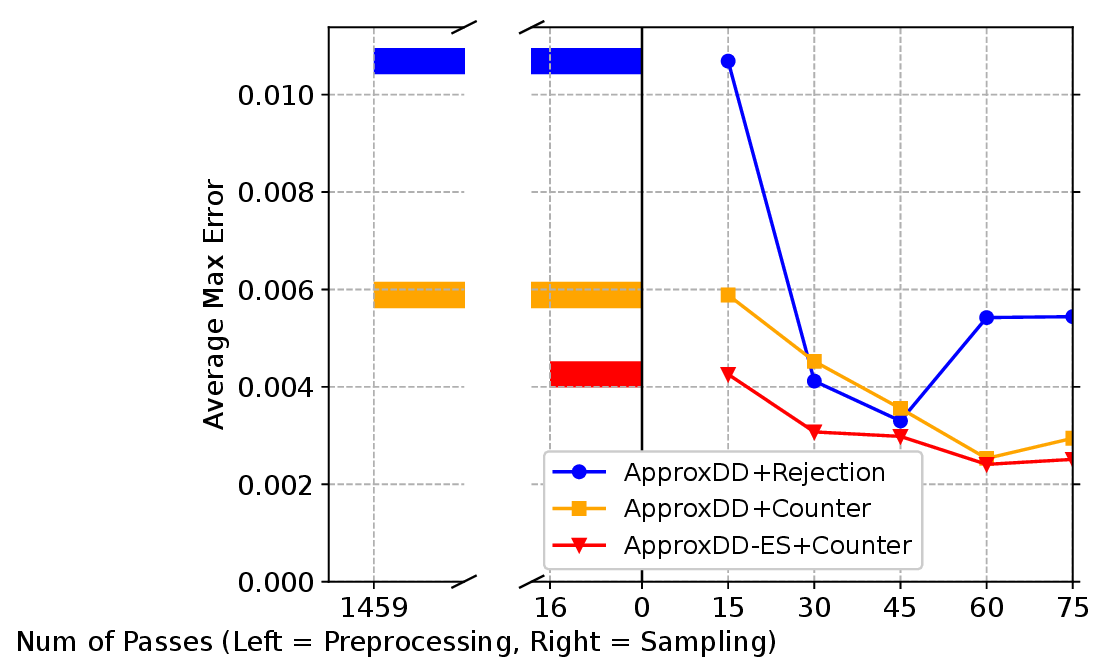}
    \caption{ER-2, $k=4$}
  \end{subfigure}
  \begin{subfigure}{0.32\textwidth}
    \includegraphics[width=\textwidth,trim=21 25 0 0,clip]{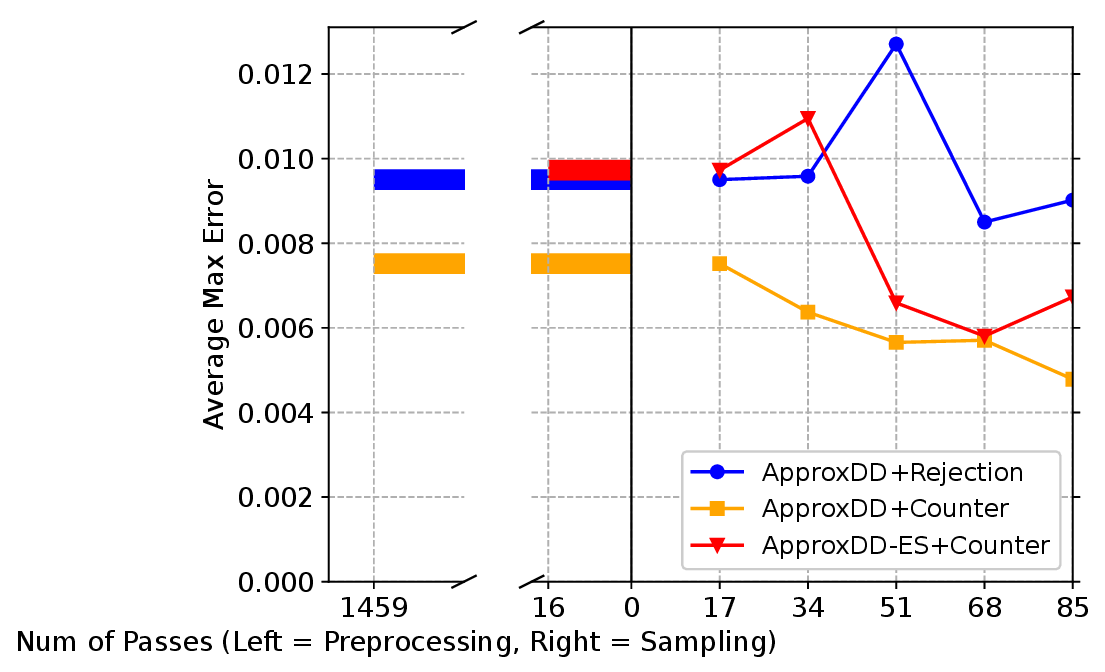}
    \caption{ER-2, $k=5$}
  \end{subfigure}
  \begin{subfigure}{0.32\textwidth}
    \includegraphics[width=\textwidth,trim=21 25 0 0,clip]{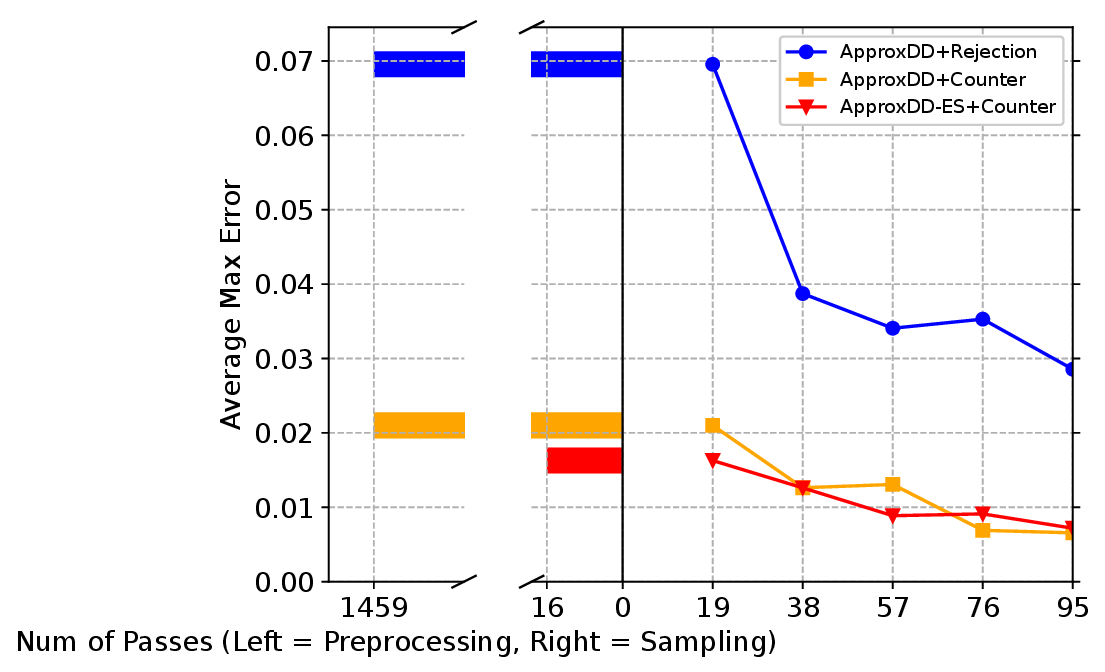}
    \caption{ER-2, $k=6$}
  \end{subfigure}
  \begin{subfigure}{0.32\textwidth}
    \includegraphics[width=\textwidth,trim=21 25 0 0,clip]{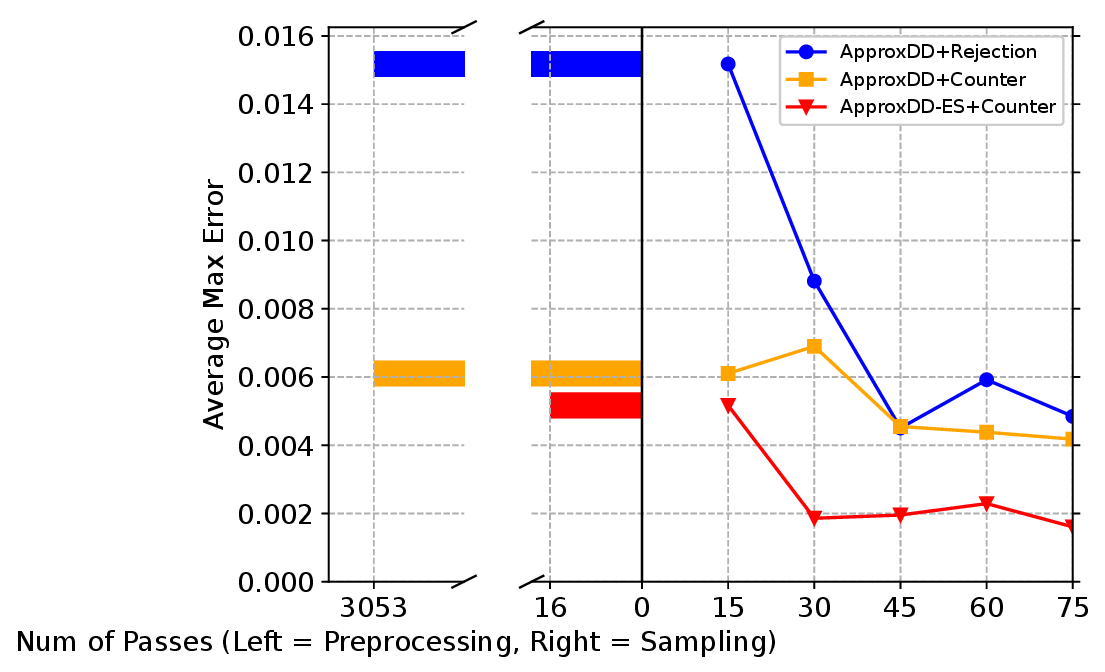}
    \caption{ER-4, $k=4$}
  \end{subfigure}
  \begin{subfigure}{0.32\textwidth}
    \includegraphics[width=\textwidth,trim=21 25 0 0,clip]{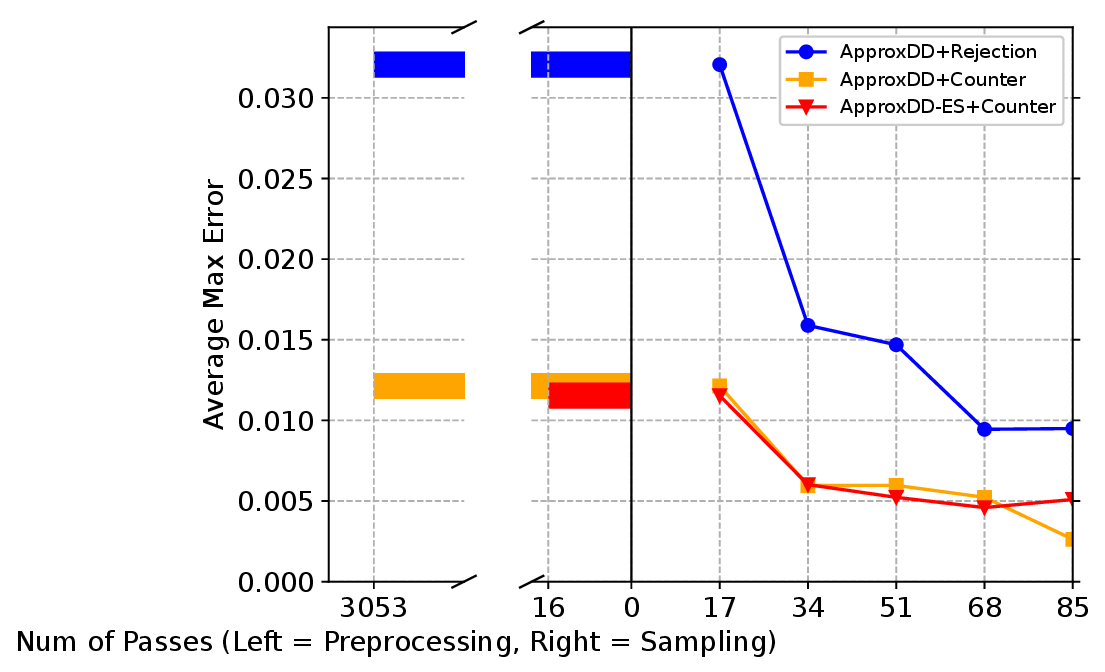}
    \caption{ER-4, $k=5$}
  \end{subfigure}
  \begin{subfigure}{0.32\textwidth}
    \includegraphics[width=\textwidth,trim=21 25 0 0,clip]{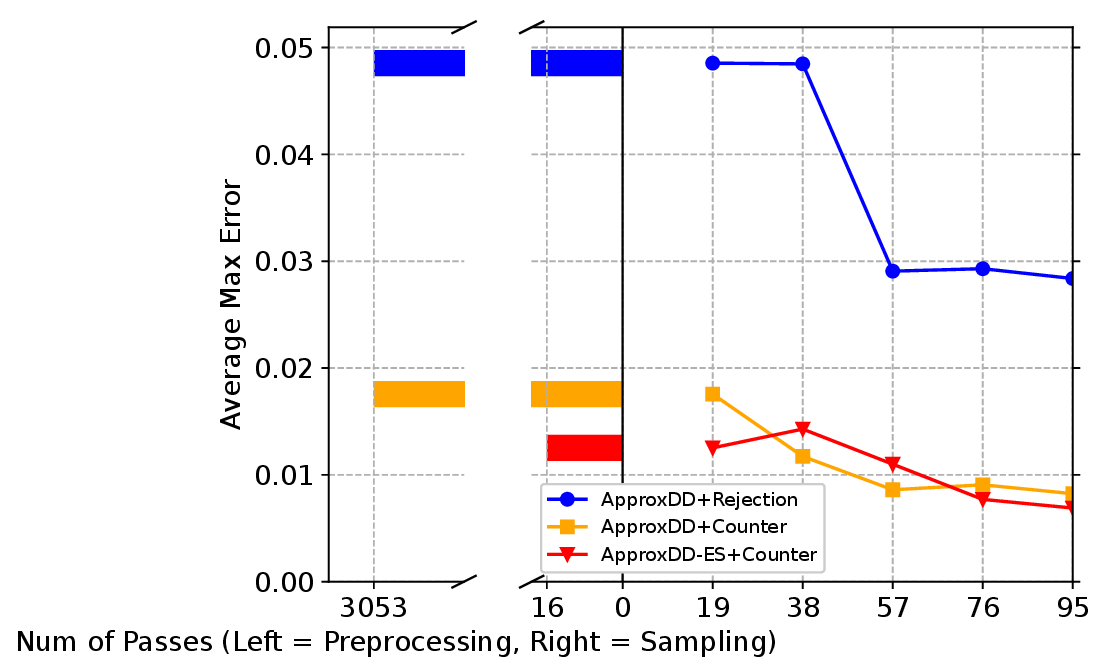}
    \caption{ER-4, $k=6$}
  \end{subfigure}
  \begin{subfigure}{0.32\textwidth}
    \includegraphics[width=\textwidth,trim=21 25 0 0,clip]{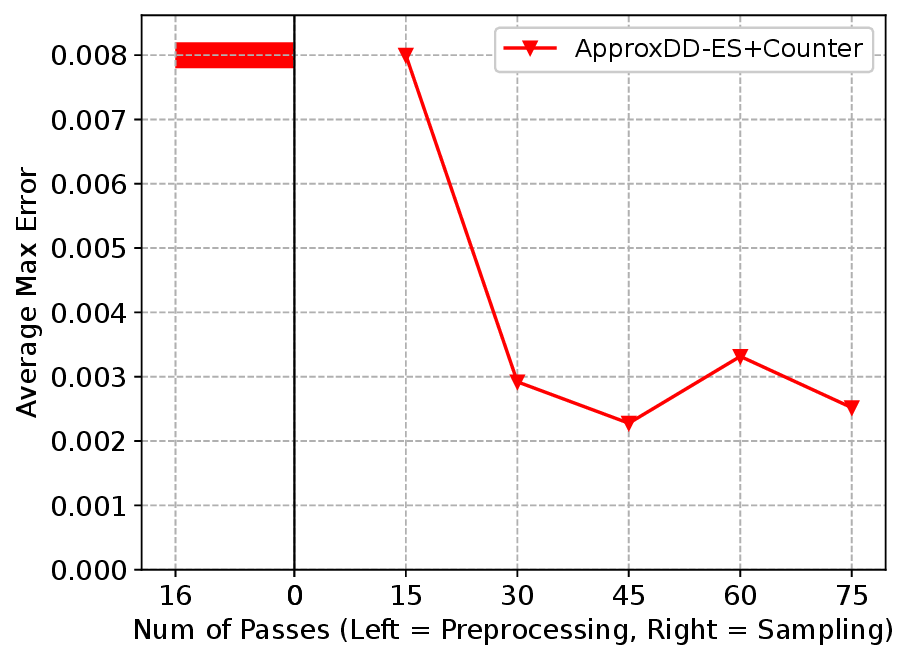}
    \caption{ER-6, $k=4$}
  \end{subfigure}
  \begin{subfigure}{0.32\textwidth}
    \includegraphics[width=\textwidth,trim=21 25 0 0,clip]{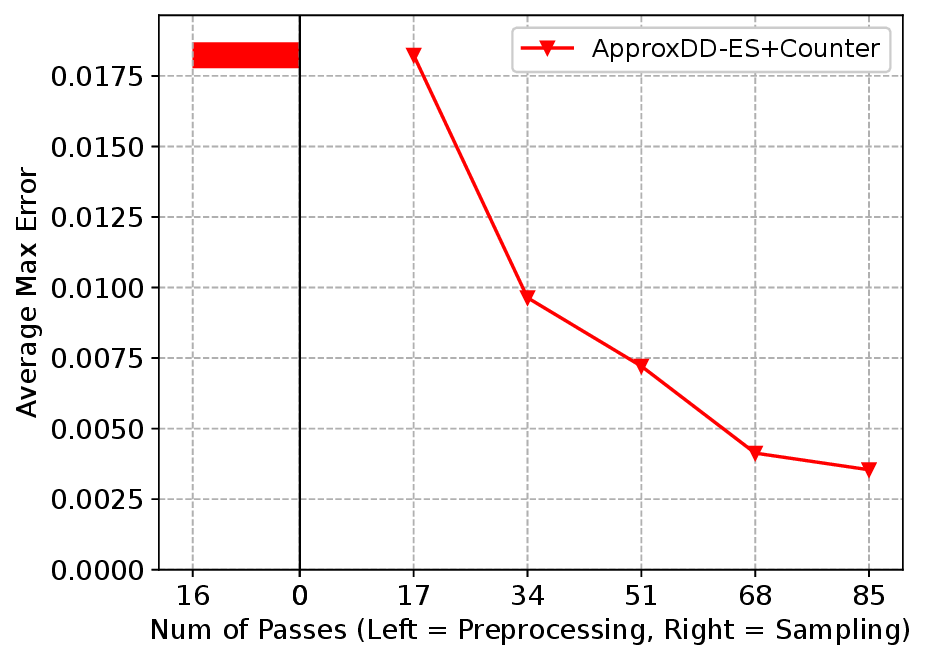}
    \caption{ER-6, $k=5$}
  \end{subfigure}
  \begin{subfigure}{0.32\textwidth}
    \includegraphics[width=\textwidth,trim=21 25 0 0,clip]{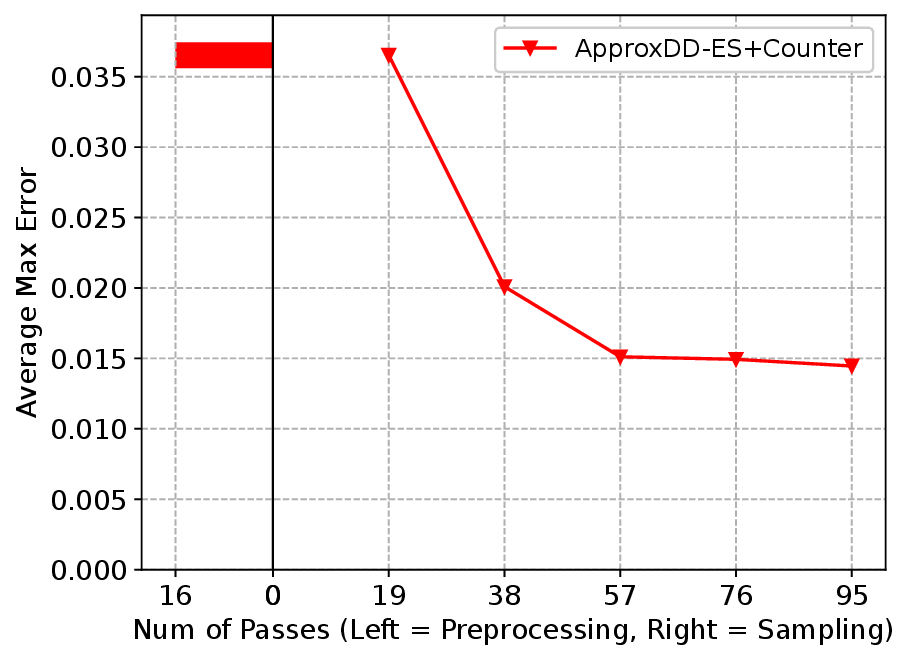}
    \caption{ER-6, $k=6$}
  \end{subfigure}
  \caption{(Continued) L$_{\infty}$ distance between the ground-truth and the estimated $k$-graphlet distribution as a function of the number of passes for sampling. The X-axis shows the number of passes (preprocessing on the left, sampling on the right). Missing points for \ApproxDD mean it did not terminate within 36 hours.}
\end{figure}

\smallskip
\noindentparagraph{Memory usage}
\Cref{tab:gd_memory} shows the memory usage of all algorithms.
In fact, in all our experiments, the total memory usage of \ApproxDDb+\gd and \ApproxDD+(\gdr or \gd) was dominated by the memory required for the preprocessing (which is already shown graphically in \Cref{exp:DD}).
Again, the memory usage of \ApproxDDb+\gd is slightly worse than \ApproxDD+\gdr and \ApproxDD+\gd, but not by large factors.
This confirms that \ApproxDDb+\gd is competitive.

\begin{table}
  \centering
  \caption{Overall memory usage of the algorithms (in MB). The percentages in brackets represent the ratio of memory usage over the disk size of the graph (each edge is represented in 8 bytes).}\label{tab:gd_memory}
  \scriptsize
  \begin{tabular}{r|c|c|ccc}
  \hline
  \multirow{3}{*}{Dataset} & \ApproxDDb & \ApproxDD & \multicolumn{3}{c}{Motivo} \\
   & +\gd & +(\gdr & $k=4$ & $k=5$ & $k=6$ \\
   &  & or \gd) &  &  &  \\
  \hline\hline
  NY Times & 38.83 (7.31\%) & 34.97 (6.58\%) & 856.60 (161.19\%) & 913.18 (171.84\%) & 1115.45 (209.90\%) \\
  Twitter (WWW) & 3370.27 (36.74\%) & 2691.84 (29.34\%) & 12407.77 (135.24\%) & 16687.46 (181.89\%) & 32576.35 (355.08\%) \\
  Twitter (MPI) & 4391.69 (35.66\%) & 3382.18 (27.46\%) & 17173.68 (139.46\%) & 24232.61 (196.78\%) & 48619.79 (394.81\%) \\
  Friendster & 5895.11 (42.65\%) & 4429.02 (32.04\%) & 20857.70 (150.89\%) & 26691.57 (193.09\%) & 45527.43 (329.35\%) \\
  Sim-0 & 16.50 (227.90\%) & 17.98 (248.37\%) & 99.21 (1370.71\%) & 122.89 (1697.89\%) & 192.33 (2657.24\%) \\
  Sim-1 & 17.46 (23.29\%) & 18.33 (24.45\%) & 181.23 (241.77\%) & 260.64 (347.71\%) & 513.17 (684.59\%) \\
  Sim-2 & 18.25 (4.61\%) & 19.17 (4.85\%) & 508.91 (128.67\%) & 611.73 (154.67\%) & 932.91 (235.87\%) \\
  Sim-3 & 18.84 (1.33\%) & 19.81 (1.40\%) & 1531.61 (108.16\%) & 1638.01 (115.68\%) & 1978.04 (139.69\%) \\
  Sim-4 & 19.16 (0.49\%) & 20.16 (0.52\%) & 3984.58 (102.87\%) & 4097.13 (105.77\%) & 4440.09 (114.63\%) \\
  Sim-5 & 19.39 (0.22\%) & 20.40 (0.24\%) & 8692.20 (100.39\%) & 8766.39 (101.25\%) & 9158.10 (105.77\%) \\
  Sim-6 & 19.82 (0.12\%) & - & 17176.55 (103.97\%) & 17297.12 (104.70\%) & 17656.74 (106.88\%) \\
  Sim-7 & 20.05 (0.07\%) & - & 28061.98 (101.13\%) & 28112.06 (101.31\%) & 28305.63 (102.00\%) \\
  ER-0 & 18.22 (251.50\%) & 17.78 (245.55\%) & 125.80 (1736.88\%) & 214.46 (2961.11\%) & 504.20 (6961.54\%) \\
  ER-1 & 20.96 (27.97\%) & 19.99 (26.67\%) & 194.11 (259.04\%) & 298.66 (398.56\%) & 643.80 (859.14\%) \\
  ER-2 & 20.55 (5.20\%) & 20.40 (5.16\%) & 511.20 (129.24\%) & 618.41 (156.35\%) & 964.53 (243.86\%) \\
  ER-3 & 21.18 (1.50\%) & 20.31 (1.43\%) & 1531.78 (108.17\%) & 1638.54 (115.71\%) & 1982.57 (140.01\%) \\
  ER-4 & 21.05 (0.54\%) & 20.16 (0.52\%) & 3985.01 (102.88\%) & 4101.33 (105.88\%) & 4441.28 (114.66\%) \\
  ER-5 & 21.66 (0.25\%) & 20.43 (0.24\%) & 8692.22 (100.39\%) & 8801.50 (101.65\%) & 9193.64 (106.18\%) \\
  ER-6 & 21.69 (0.13\%) & - & 16902.60 (102.31\%) & 16999.68 (102.90\%) & 17619.79 (106.65\%) \\
  \hline
  \end{tabular}
\end{table}


\section{Conclusion}

This paper presents a novel streaming algorithm for approximating $k$-graphlet distributions in large graphs that cannot fit entirely in memory. By introducing an improved method for computing an approximate degree-dominating (DD) order---a key preprocessing step---our algorithm reduces the number of passes from $O(\log n)$ (as in prior work) to $O(1/c)$, using $\tilde{O}(n^{1+c})$ memory. This result is nearly optimal, as a lower bound rules out $O(1)$-pass algorithms with sublinear memory.

Empirical evaluations on real-world and synthetic graphs demonstrate that our algorithm significantly outperforms the state-of-the-art, especially on moderately dense graphs, where it reduces the number of passes by orders of magnitude while maintaining comparable memory usage. The integration of a Horvitz–Thompson estimator further enhances efficiency by avoiding the high rejection rates of earlier sampling methods.

Overall, our work provides a scalable, predictable, and theoretically sound solution for graphlet analysis in the streaming model, bridging the gap between theoretical guarantees and practical performance.

\ignore{
We develop an efficient semi-streaming algorithm for approximating the $k$-graphlet distribution of a graph, which requires constant passes and $\tilde O(n^{1+c})$ memory for an arbitrary constant $c$ with high probability. We conducted experiments showing the better efficiency on the number of passes for computing the DD order against the previous uniform graphlet sampling algorithm in \cite{Bourreau2024}, and the lower memory usage against the existing method that approximates the $k$-graphlet distribution via graphlet sampling.

We believe that this is the first work to consider approximating the $k$-graphlet distribution in massive graphs, which is broadly interesting in data mining and relative areas, and deserves further improvement in the streaming model or more large-scale computational models such as Massively Parallel Computation and MapReduce.
}

\ignore{
We develop efficient semi-streaming algorithms for uniform graphlet sampling, requiring sublinear amount of memory and polylogarithmic number of passes in the size of the input data. We also provide a space lower bound showing that the tradeoff between memory and number of passes of our algorithms is near-optimal. Our theoretical results are complemented with an experimental evaluation on large real-world graphs showing the effectiveness of our algorithms and that they can provide a valuable tool in graph mining and data analysis. Our work paves the way for developing efficient algorithms in other computational model for large-scale processing, such as the MapReduce model.
}

\bibliographystyle{plain}
\bibliography{ref}

\clearpage

\appendix
\section{Additional Figures}
\label{sec:morefigures}

\Cref{fig:ddeval-res,fig:Linfty-res} show the figures omitted from the experiments in \Cref{sec:experiments}.

\begin{figure}[H]
  \centering
  \begin{subfigure}{0.46\textwidth}
    \includegraphics[width=\textwidth]{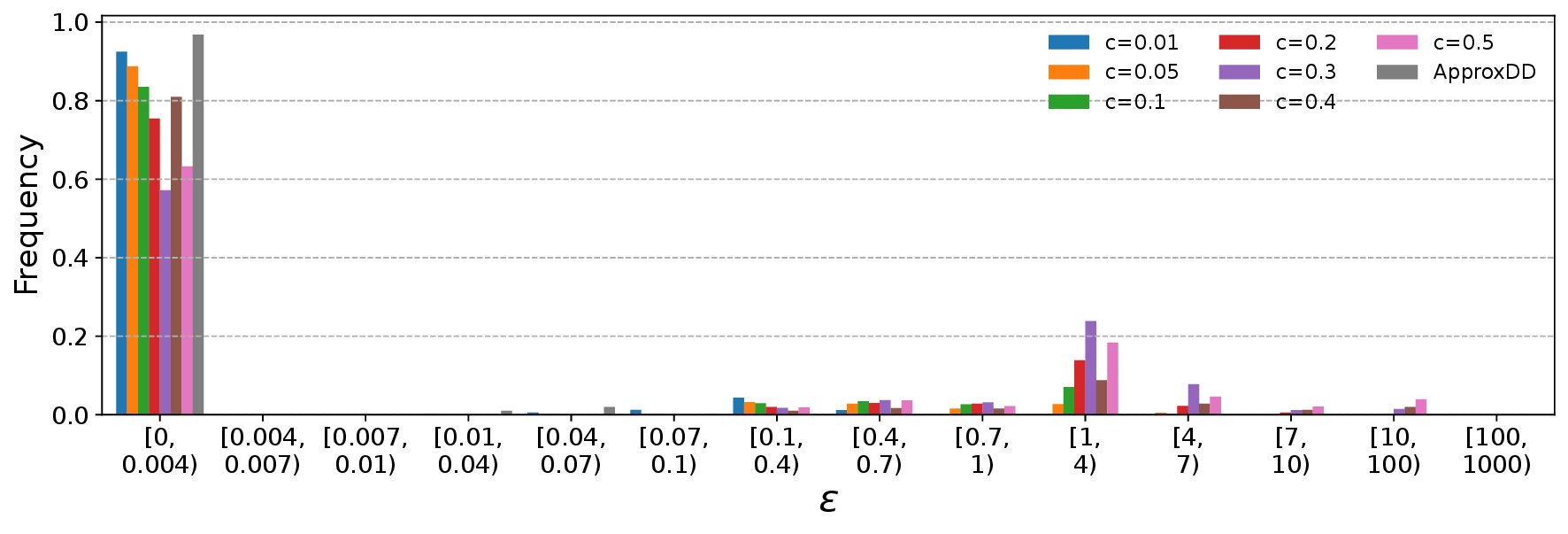}
    \caption{Sim-0}
  \end{subfigure}
  \begin{subfigure}{0.46\textwidth}
    \includegraphics[width=\textwidth]{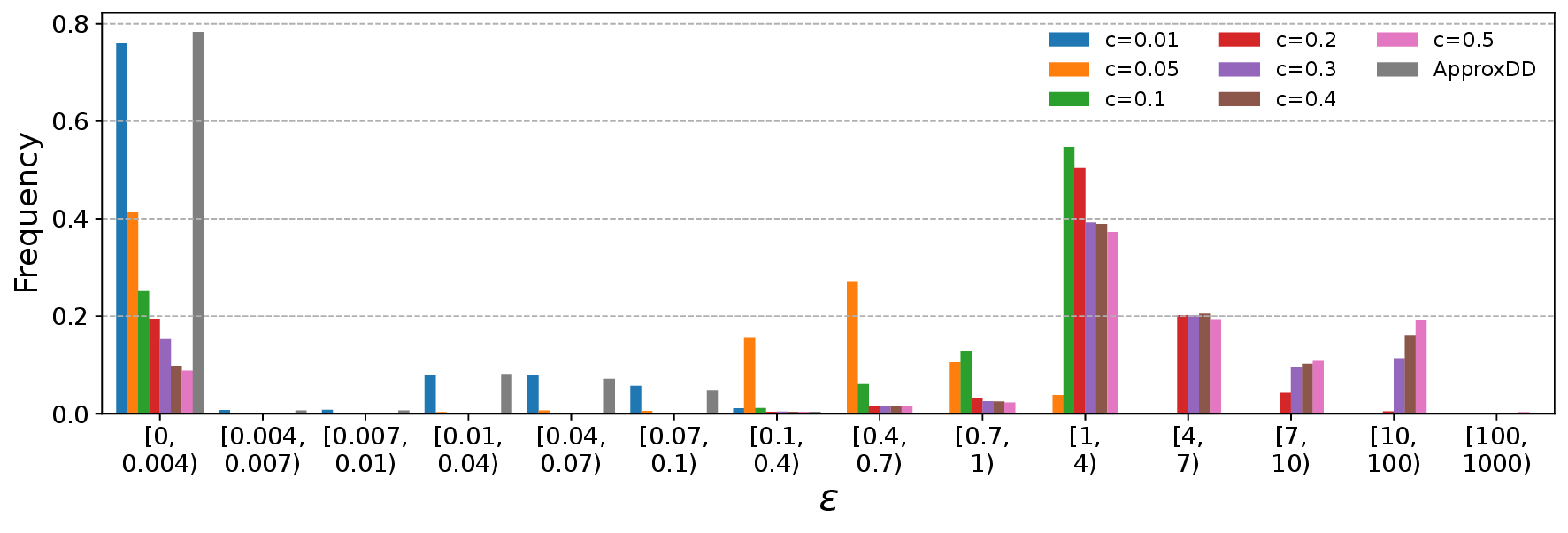}
    \caption{Sim-2}
  \end{subfigure}
  \begin{subfigure}{0.46\textwidth}
    \includegraphics[width=\textwidth]{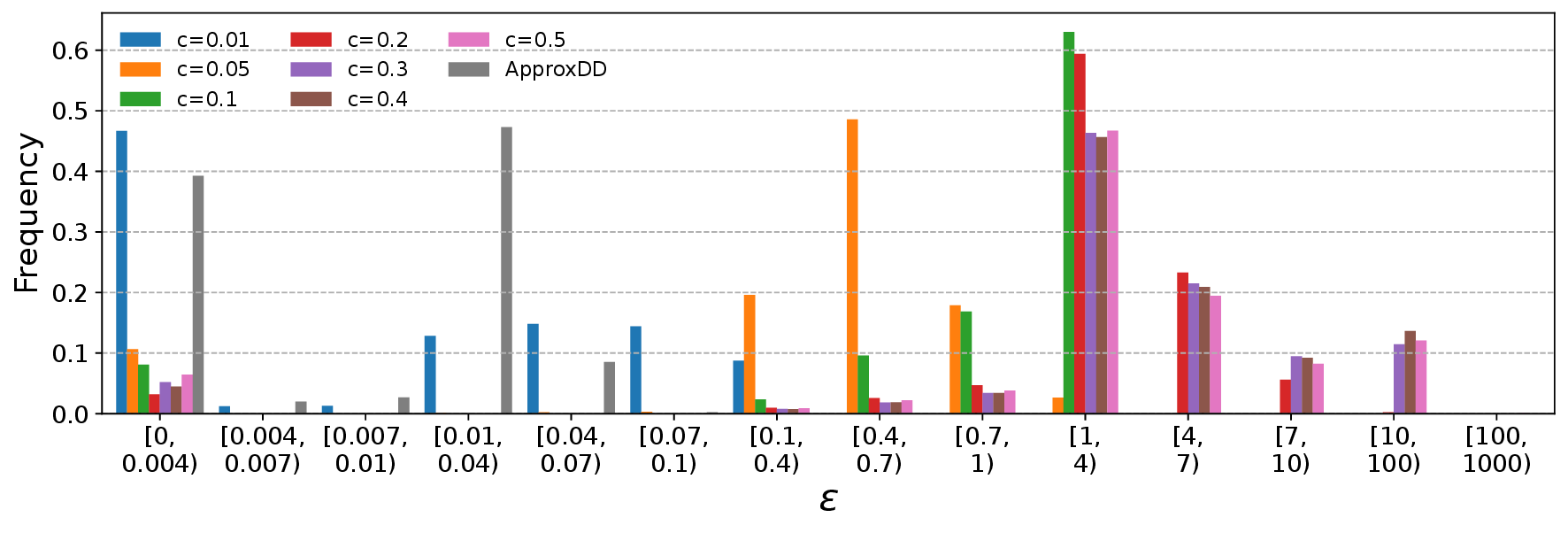}
    \caption{Sim-4}
  \end{subfigure}
  \begin{subfigure}{0.46\textwidth}
    \includegraphics[width=\textwidth]{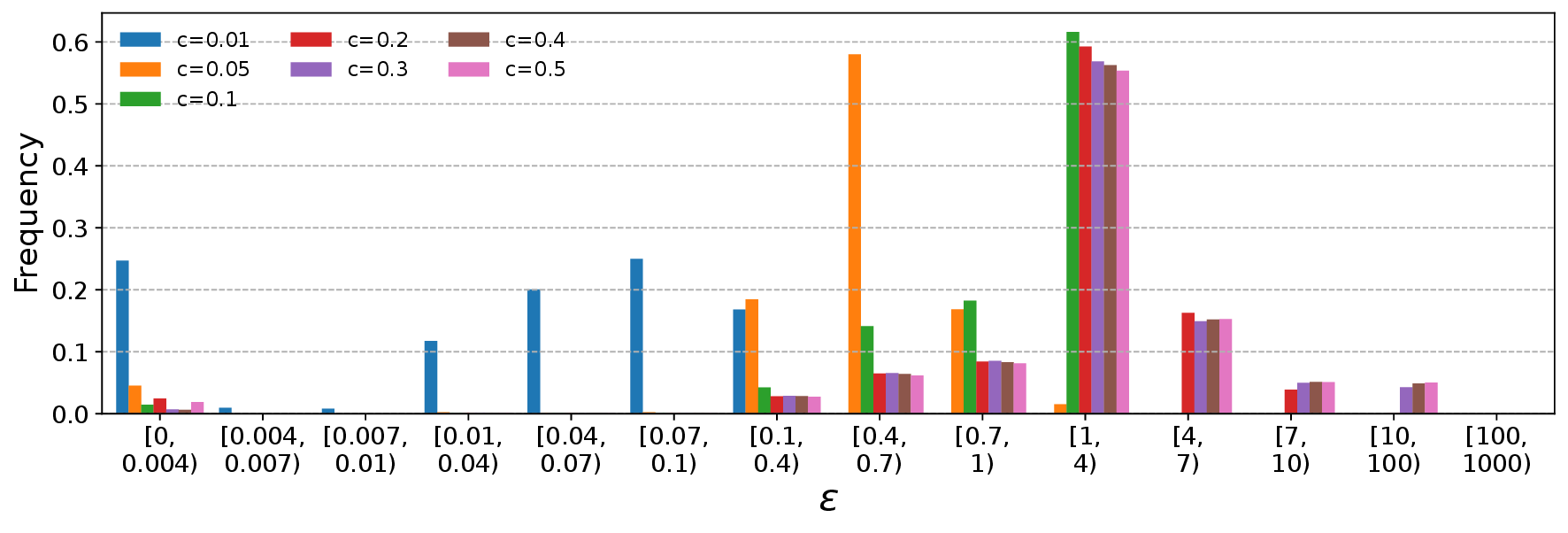}
    \caption{Sim-6}
  \end{subfigure}
  \begin{subfigure}{0.46\textwidth}
    \includegraphics[width=\textwidth]{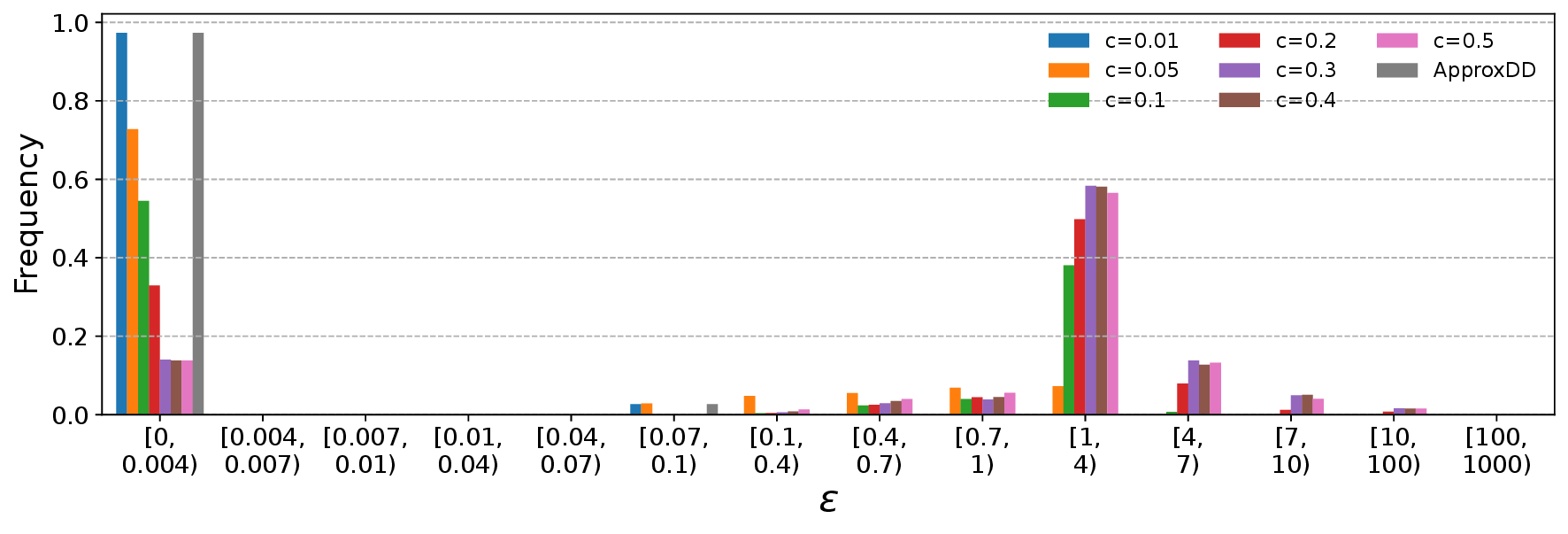}
    \caption{ER-0}
  \end{subfigure}
  \begin{subfigure}{0.46\textwidth}
    \includegraphics[width=\textwidth]{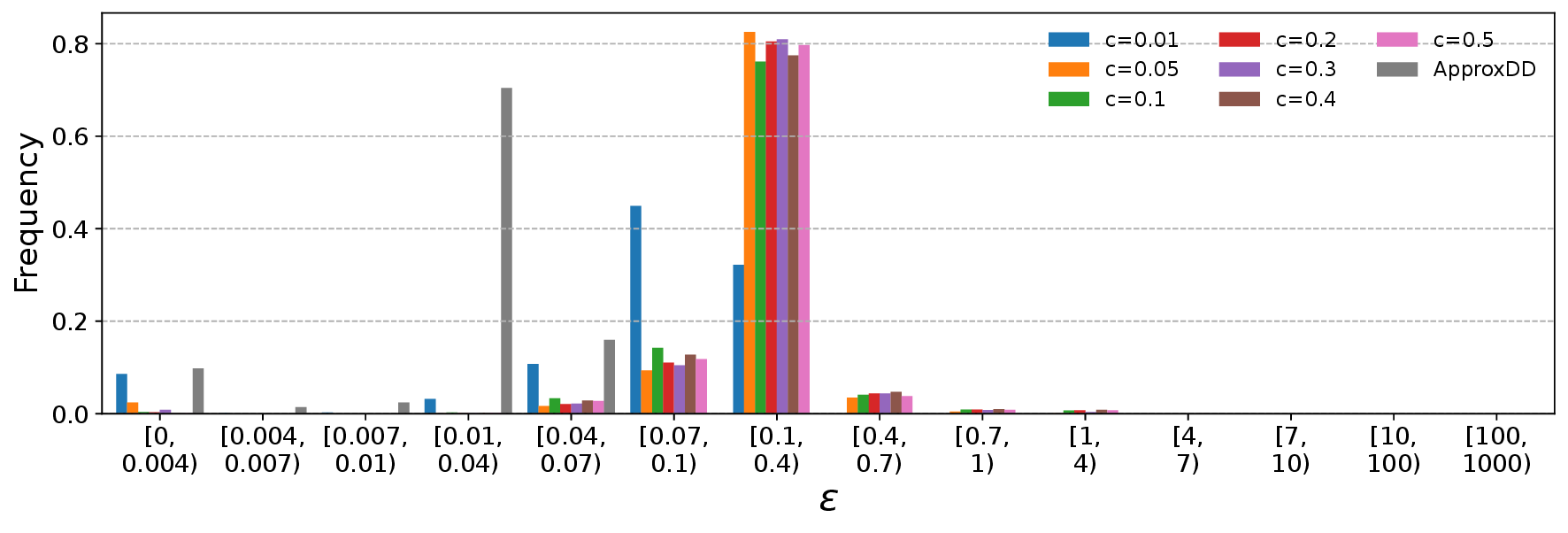}
    \caption{ER-3}
  \end{subfigure}
  \begin{subfigure}{0.46\textwidth}
    \includegraphics[width=\textwidth]{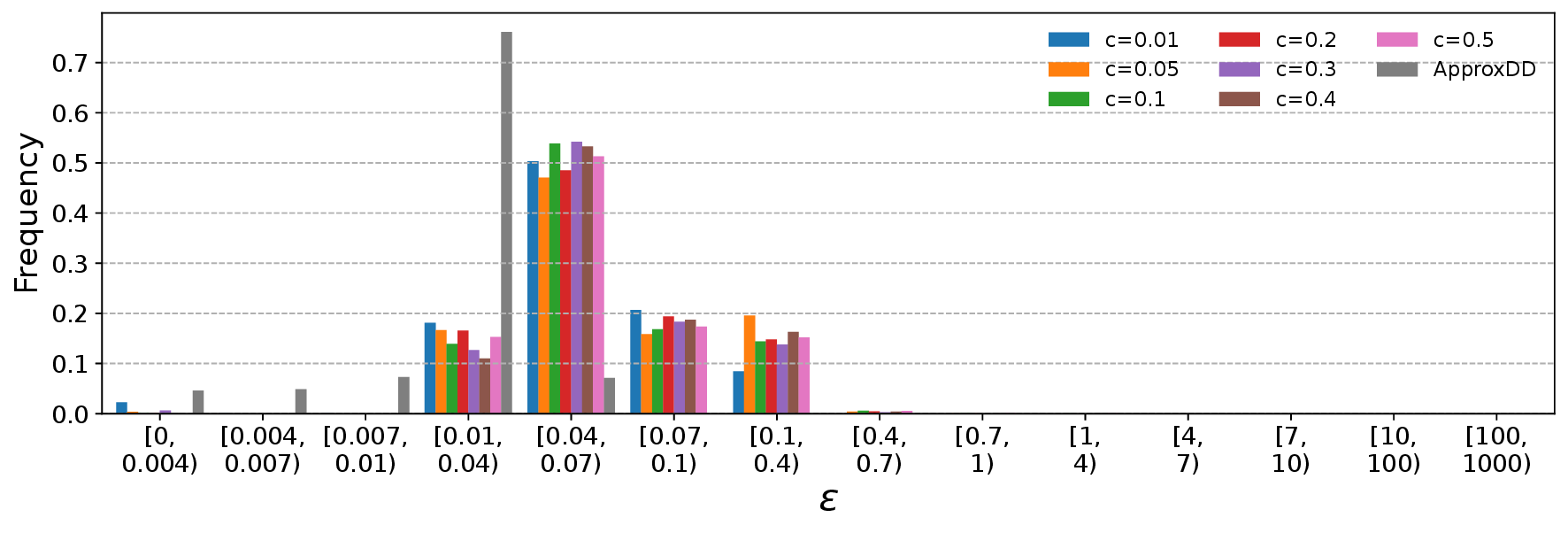}
    \caption{ER-5}
  \end{subfigure}
  \caption{Empirical distribution of $\epsilon_v$ under different $c$, using \ApproxDDb.}\label{fig:ddeval-res}
\end{figure}

\begin{figure*}
  \centering
  \begin{subfigure}{0.32\textwidth}
    \includegraphics[width=\textwidth,trim=21 25 0 0,clip]{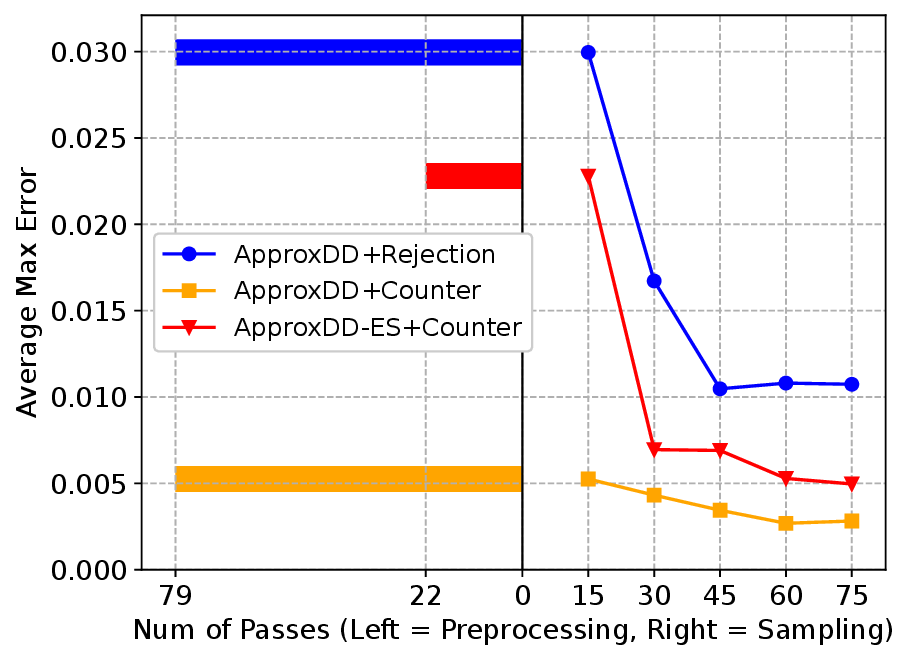}
    \caption{Sim-1, $k=4$}
  \end{subfigure}
  \begin{subfigure}{0.32\textwidth}
    \includegraphics[width=\textwidth,trim=21 25 0 0,clip]{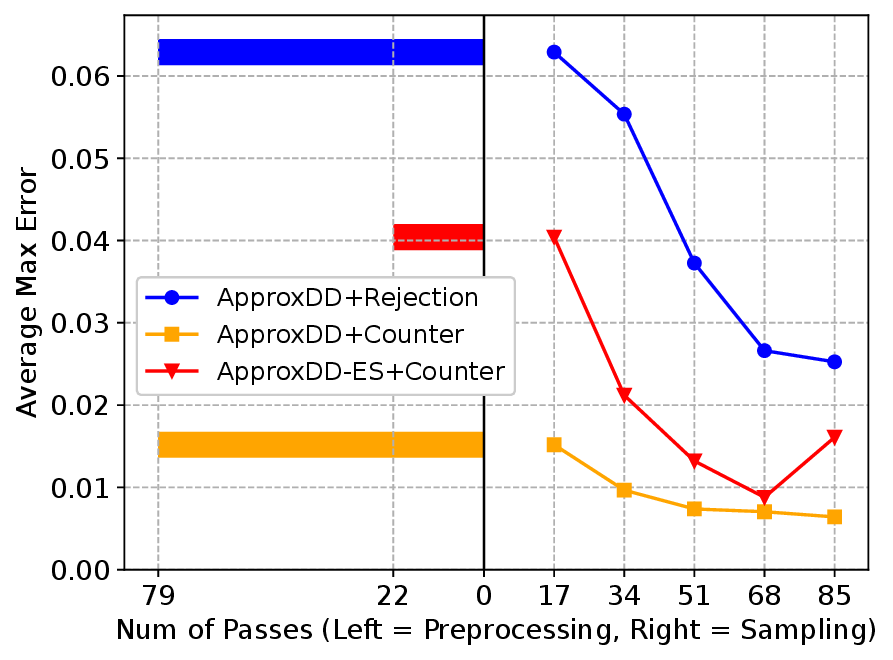}
    \caption{Sim-1, $k=5$}
  \end{subfigure}
  \begin{subfigure}{0.32\textwidth}
    \includegraphics[width=\textwidth,trim=21 25 0 0,clip]{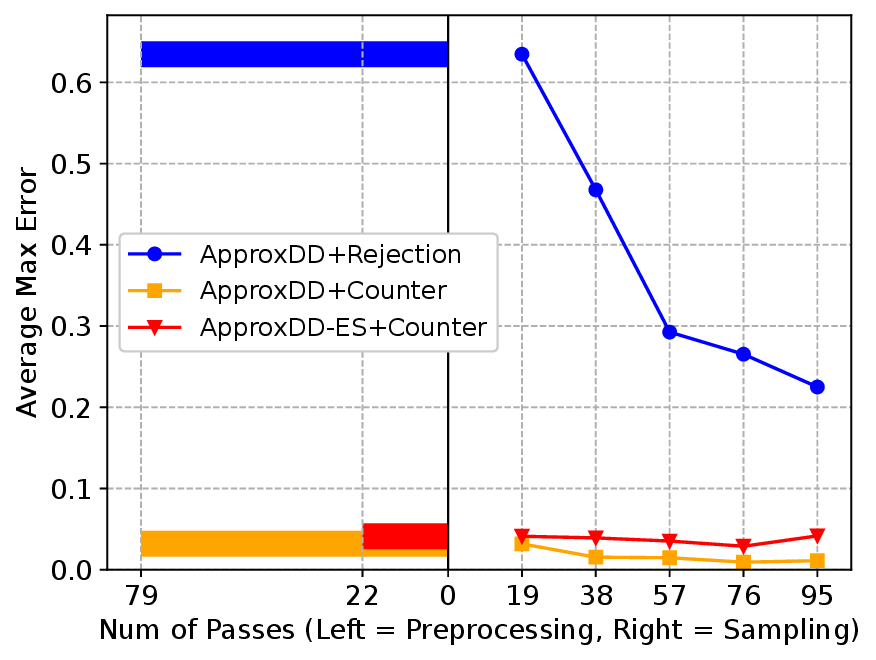}
    \caption{Sim-1, $k=6$}
  \end{subfigure}
  \begin{subfigure}{0.32\textwidth}
    \includegraphics[width=\textwidth,trim=21 25 0 0,clip]{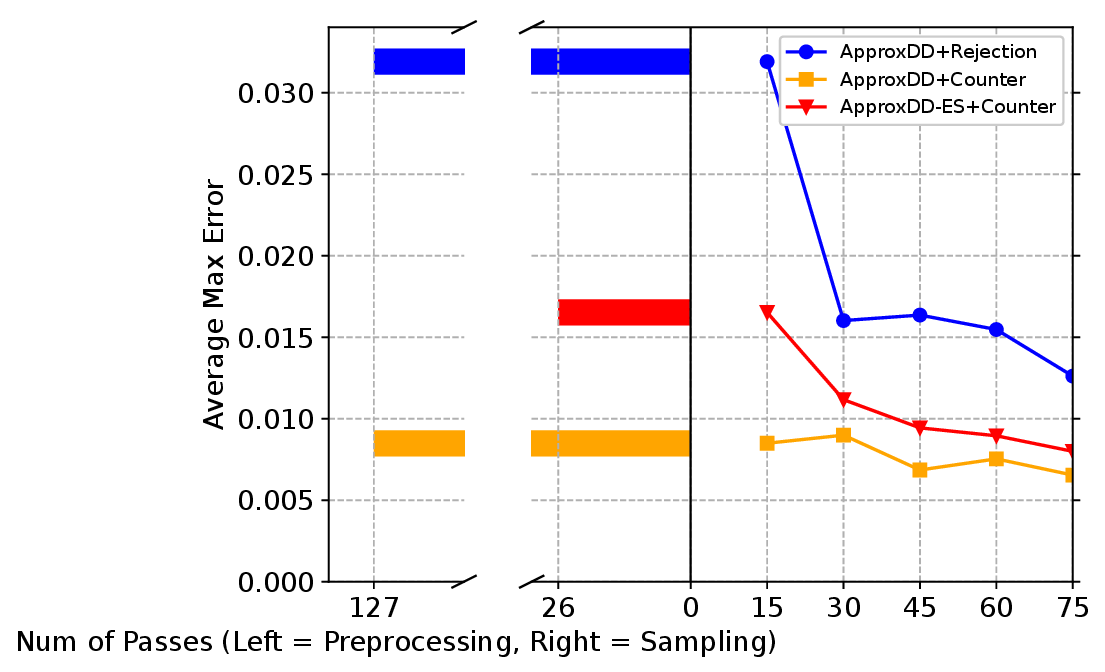}
    \caption{Sim-2, $k=4$}
  \end{subfigure}
  \begin{subfigure}{0.32\textwidth}
    \includegraphics[width=\textwidth,trim=21 25 0 0,clip]{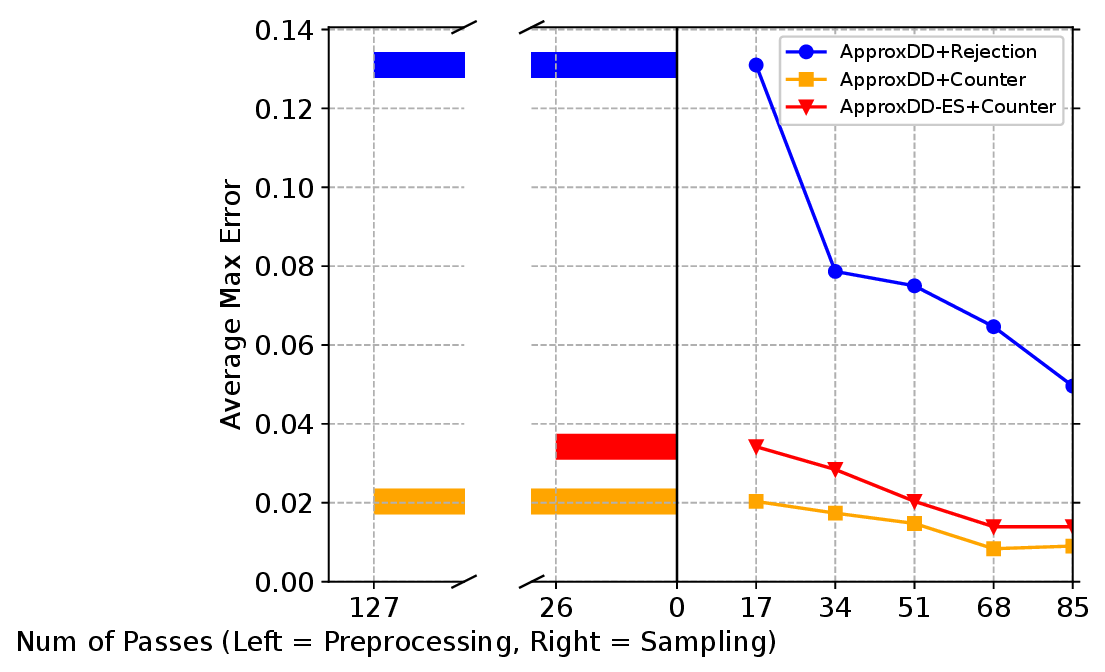}
    \caption{Sim-2, $k=5$}
  \end{subfigure}
  \begin{subfigure}{0.32\textwidth}
    \includegraphics[width=\textwidth,trim=21 25 0 0,clip]{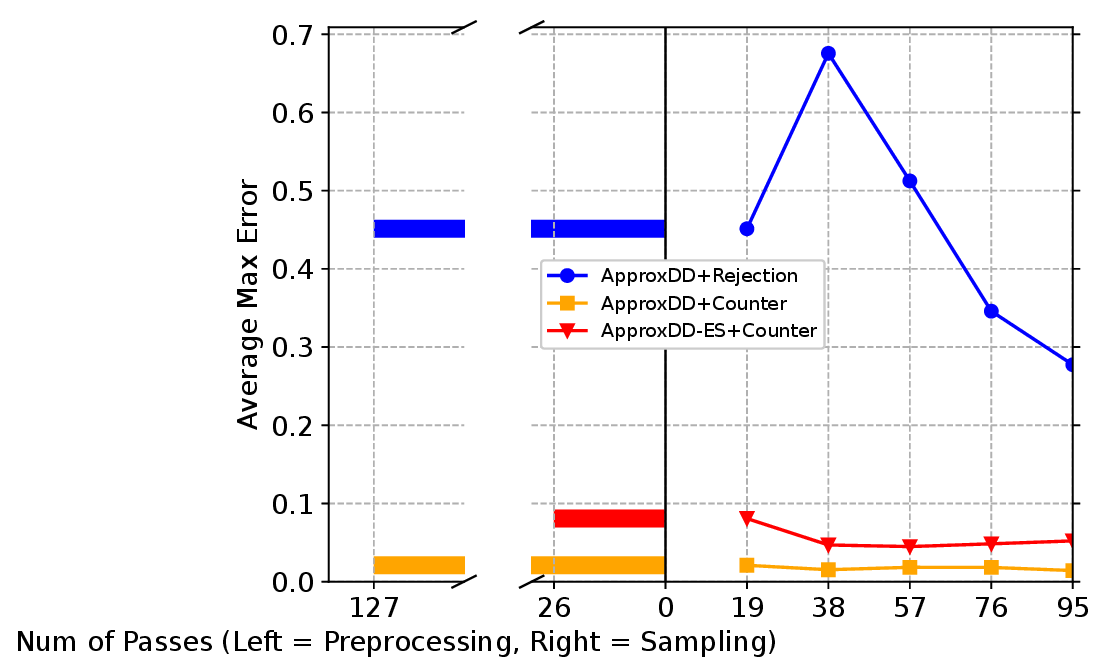}
    \caption{Sim-2, $k=6$}
  \end{subfigure}
  \begin{subfigure}{0.32\textwidth}
    \includegraphics[width=\textwidth,trim=21 25 0 0,clip]{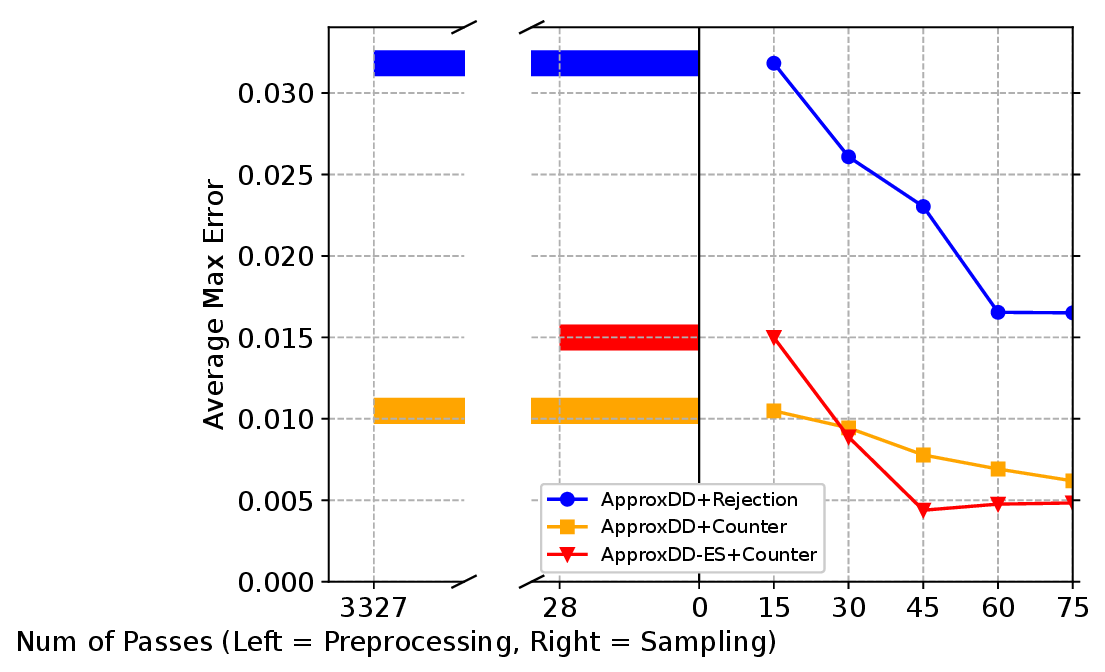}
    \caption{Sim-4, $k=4$}
  \end{subfigure}
  \begin{subfigure}{0.32\textwidth}
    \includegraphics[width=\textwidth,trim=21 25 0 0,clip]{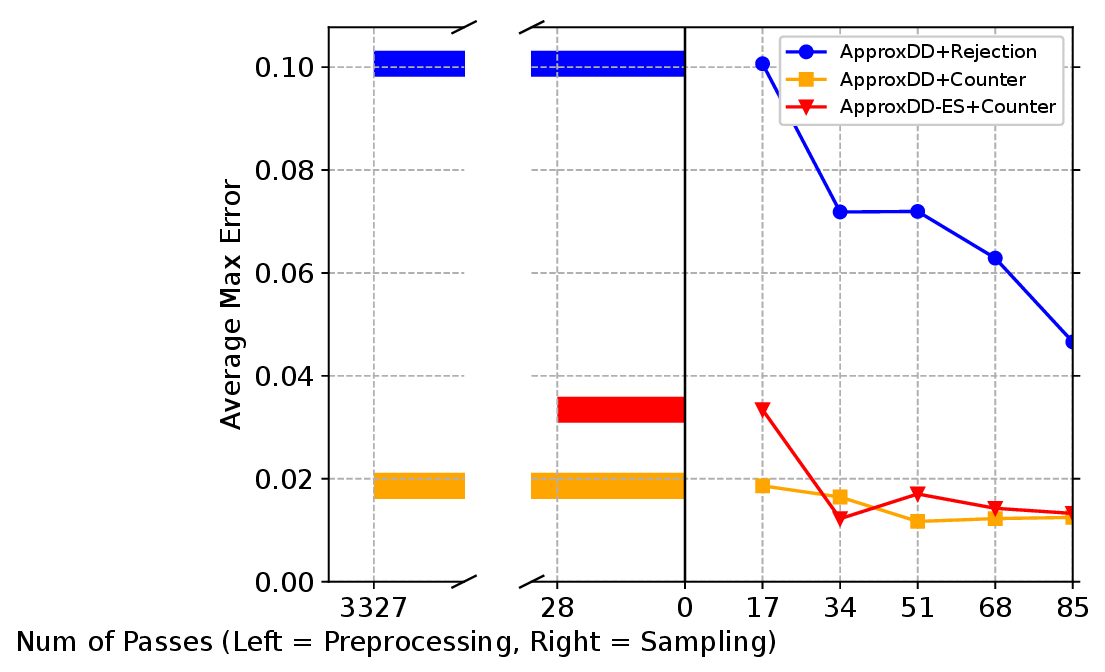}
    \caption{Sim-4, $k=5$}
  \end{subfigure}
  \begin{subfigure}{0.32\textwidth}
    \includegraphics[width=\textwidth,trim=21 25 0 0,clip]{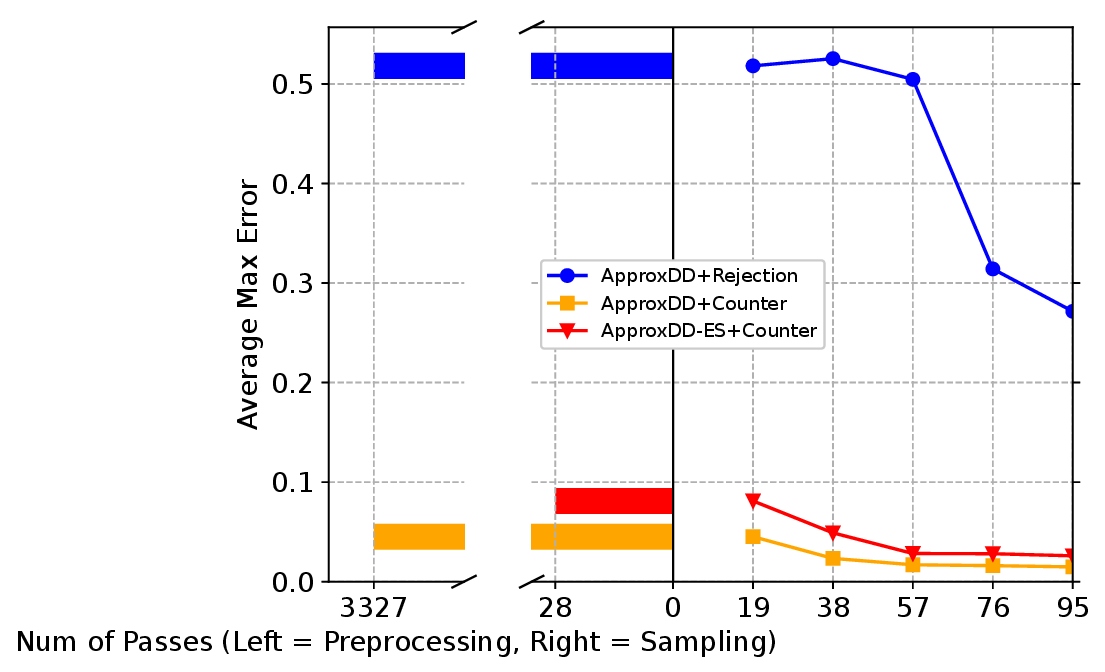}
    \caption{Sim-4, $k=6$}
  \end{subfigure}
  \begin{subfigure}{0.32\textwidth}
    \includegraphics[width=\textwidth,trim=21 25 0 0,clip]{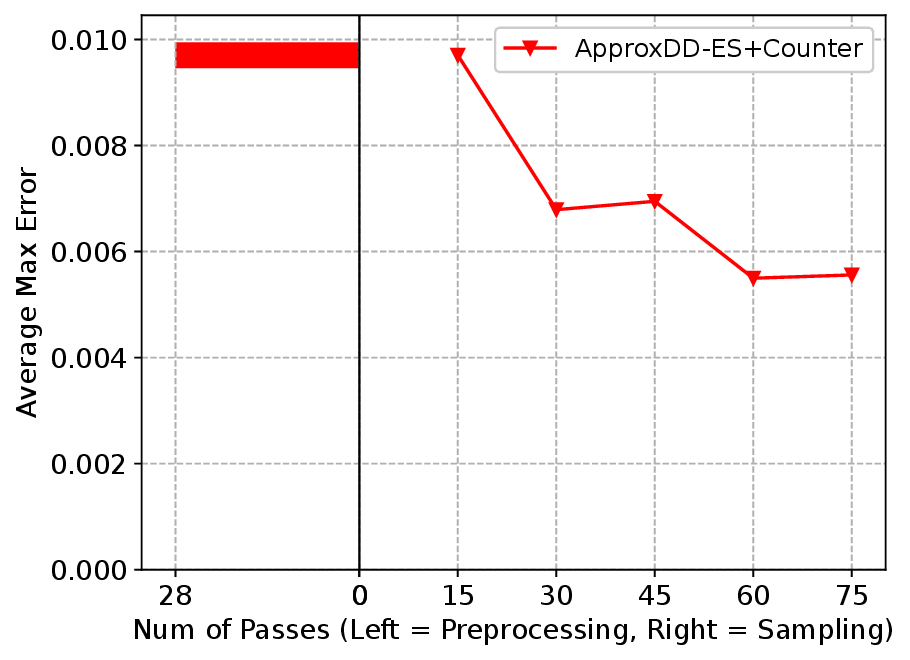}
    \caption{Sim-6, $k=4$}
  \end{subfigure}
  \begin{subfigure}{0.32\textwidth}
    \includegraphics[width=\textwidth,trim=21 25 0 0,clip]{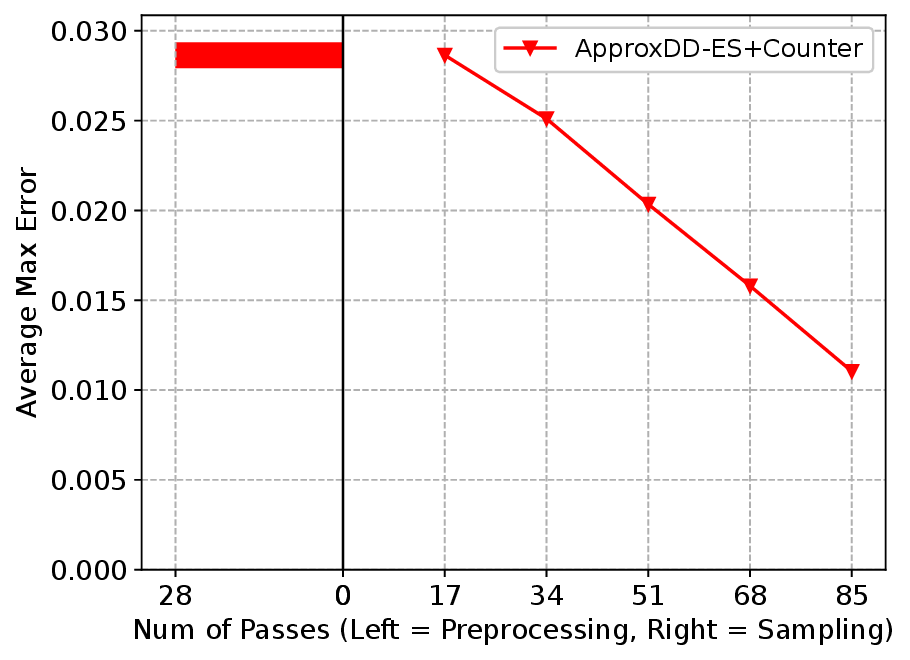}
    \caption{Sim-6, $k=5$}
  \end{subfigure}
  \begin{subfigure}{0.32\textwidth}
    \includegraphics[width=\textwidth,trim=21 25 0 0,clip]{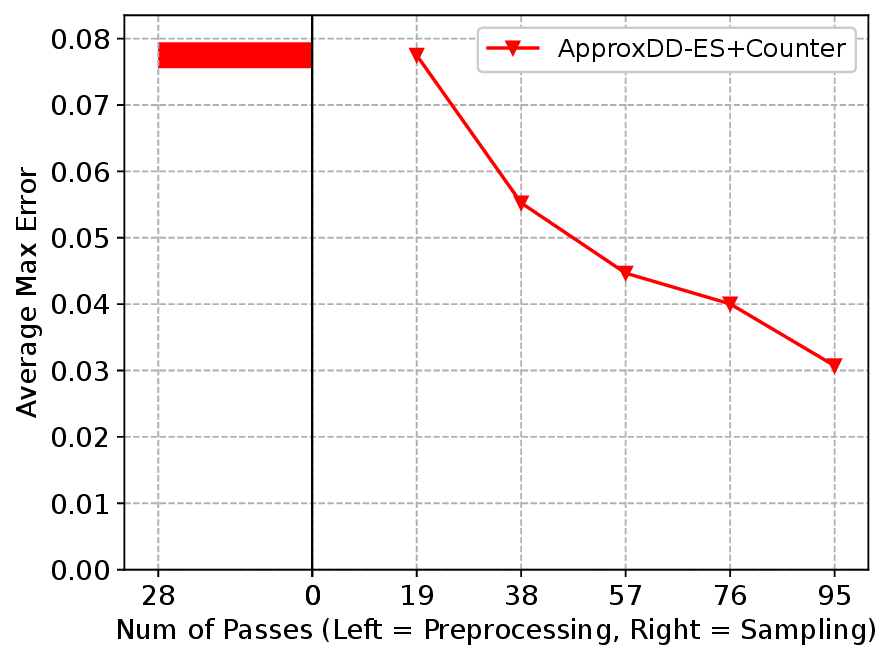}
    \caption{Sim-6, $k=6$}
  \end{subfigure}
  \begin{subfigure}{0.32\textwidth}
    \includegraphics[width=\textwidth,trim=21 25 0 0,clip]{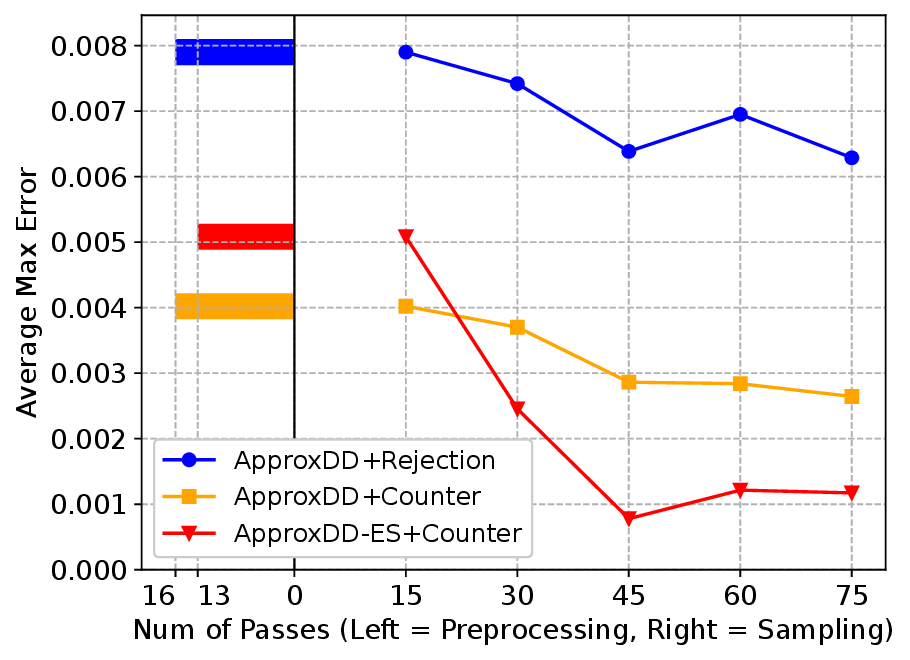}
    \caption{ER-0, $k=4$}
  \end{subfigure}
  \begin{subfigure}{0.32\textwidth}
    \includegraphics[width=\textwidth,trim=21 25 0 0,clip]{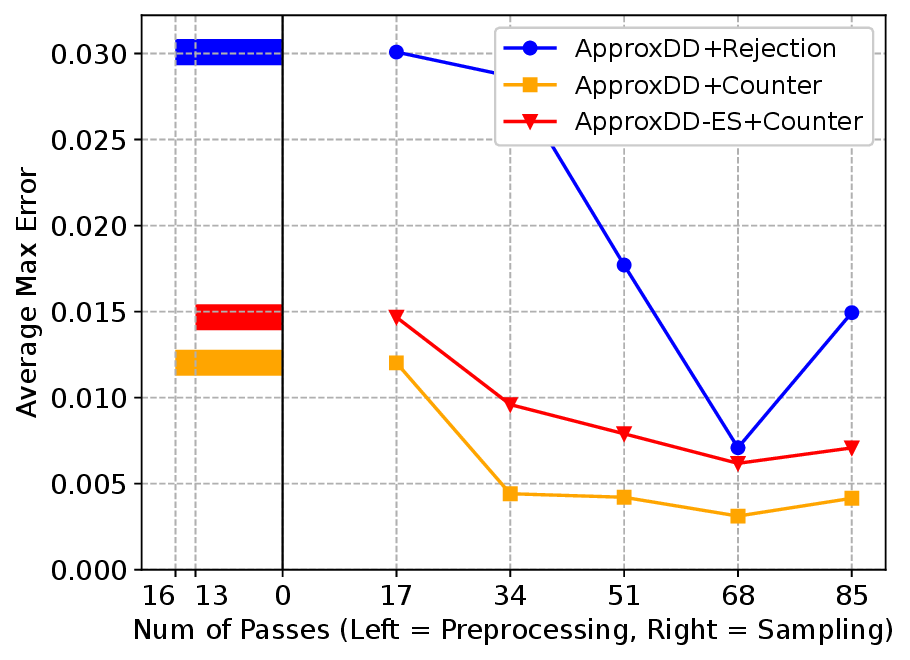}
    \caption{ER-0, $k=5$}
  \end{subfigure}
  \begin{subfigure}{0.32\textwidth}
    \includegraphics[width=\textwidth,trim=21 25 0 0,clip]{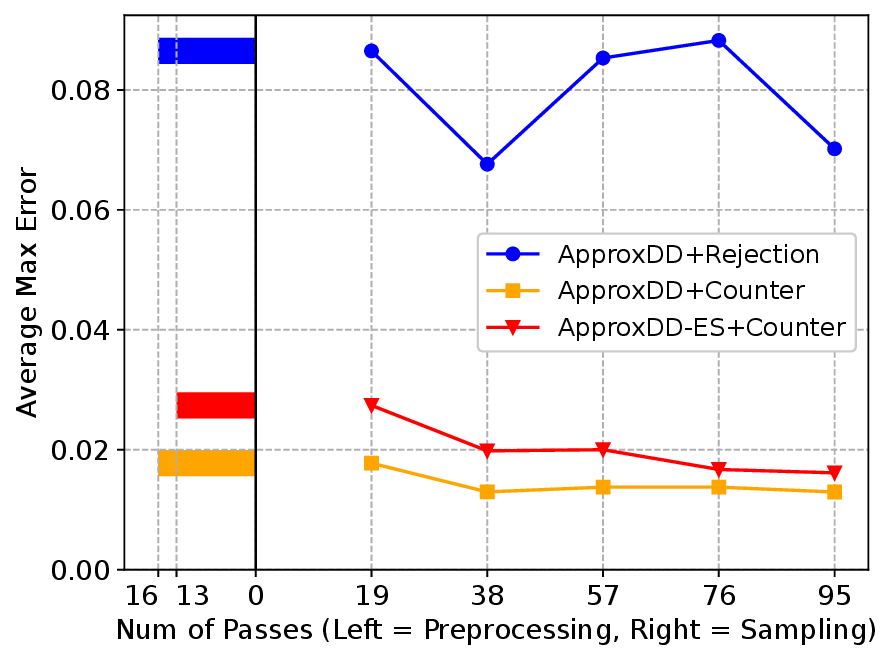}
    \caption{ER-0, $k=6$}
  \end{subfigure}
  \caption{L$_{\infty}$ distance between the ground-truth and the estimated $k$-graphlet distribution as a function of the number of passes for sampling. The X-axis shows the number of passes (preprocessing on the left, sampling on the right). Missing points for \ApproxDD mean it did not terminate within 36 hours.}\label{fig:Linfty-res}
\end{figure*}

\begin{figure*}\ContinuedFloat
  \begin{subfigure}{0.32\textwidth}
    \includegraphics[width=\textwidth,trim=21 25 0 0,clip]{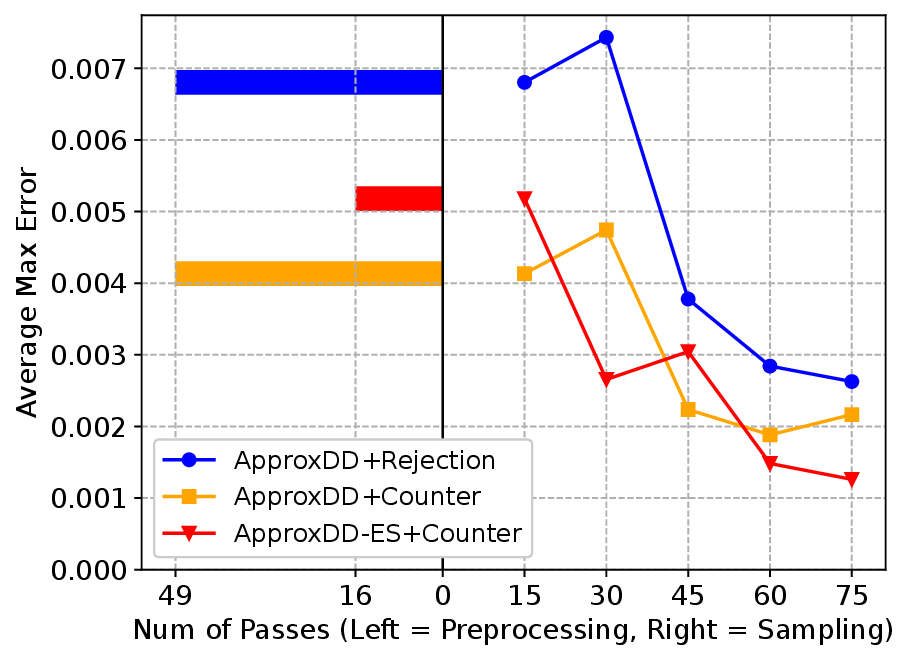}
    \caption{ER-1, $k=4$}
  \end{subfigure}
  \begin{subfigure}{0.32\textwidth}
    \includegraphics[width=\textwidth,trim=21 25 0 0,clip]{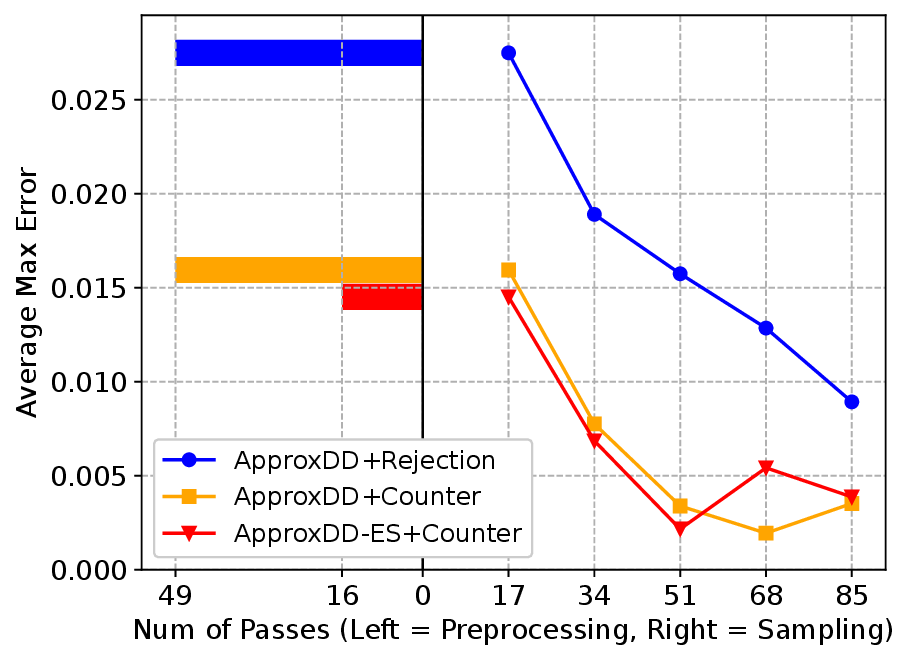}
    \caption{ER-1, $k=5$}
  \end{subfigure}
  \begin{subfigure}{0.32\textwidth}
    \includegraphics[width=\textwidth,trim=21 25 0 0,clip]{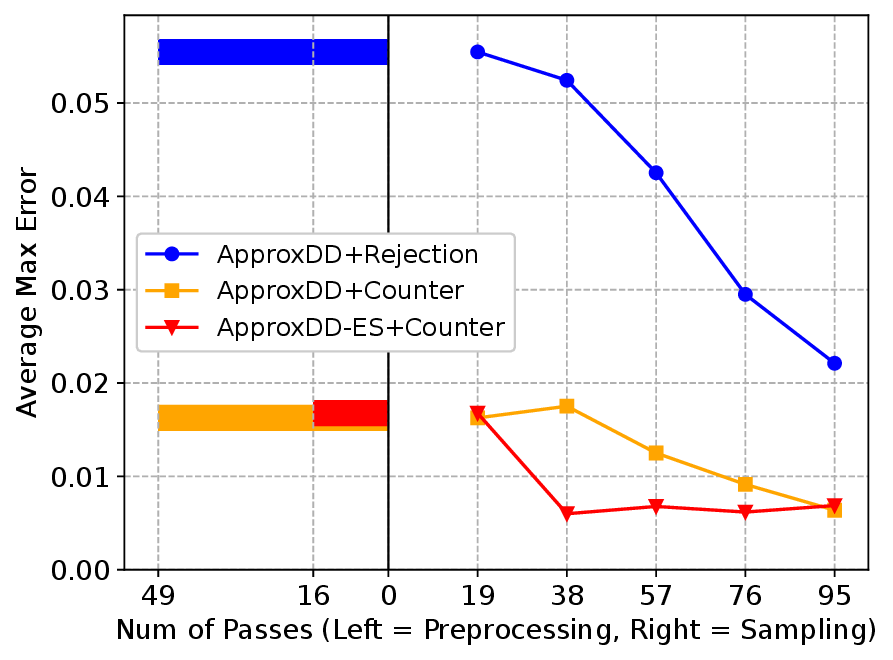}
    \caption{ER-1, $k=6$}
  \end{subfigure}
  \begin{subfigure}{0.32\textwidth}
    \includegraphics[width=\textwidth,trim=21 25 0 0,clip]{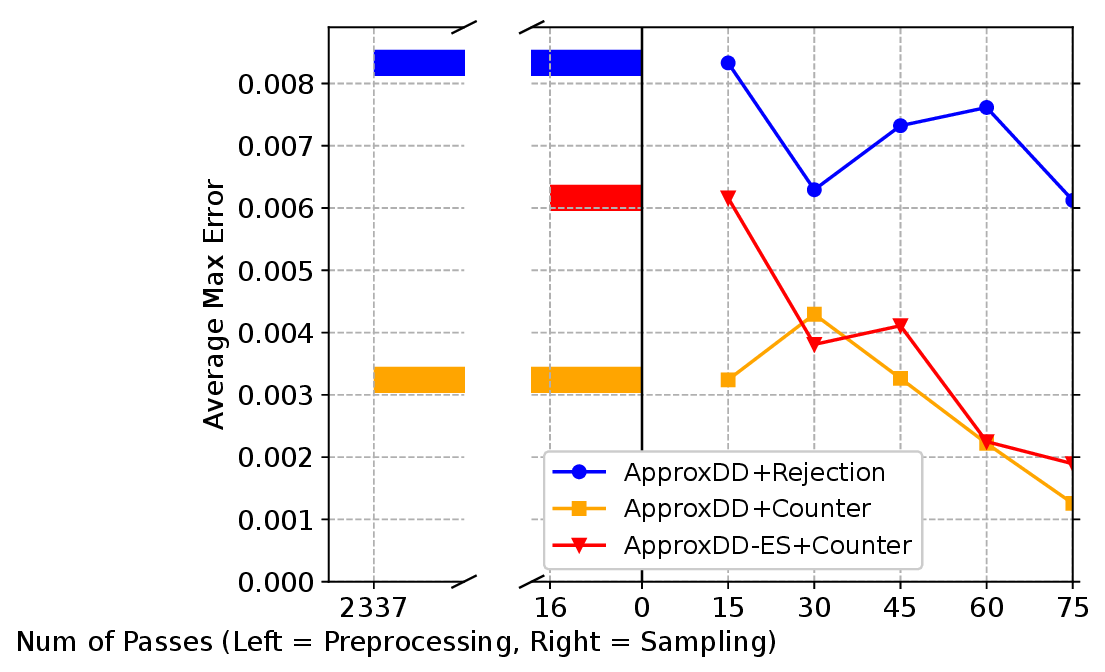}
    \caption{ER-3, $k=4$}
  \end{subfigure}
  \begin{subfigure}{0.32\textwidth}
    \includegraphics[width=\textwidth,trim=21 25 0 0,clip]{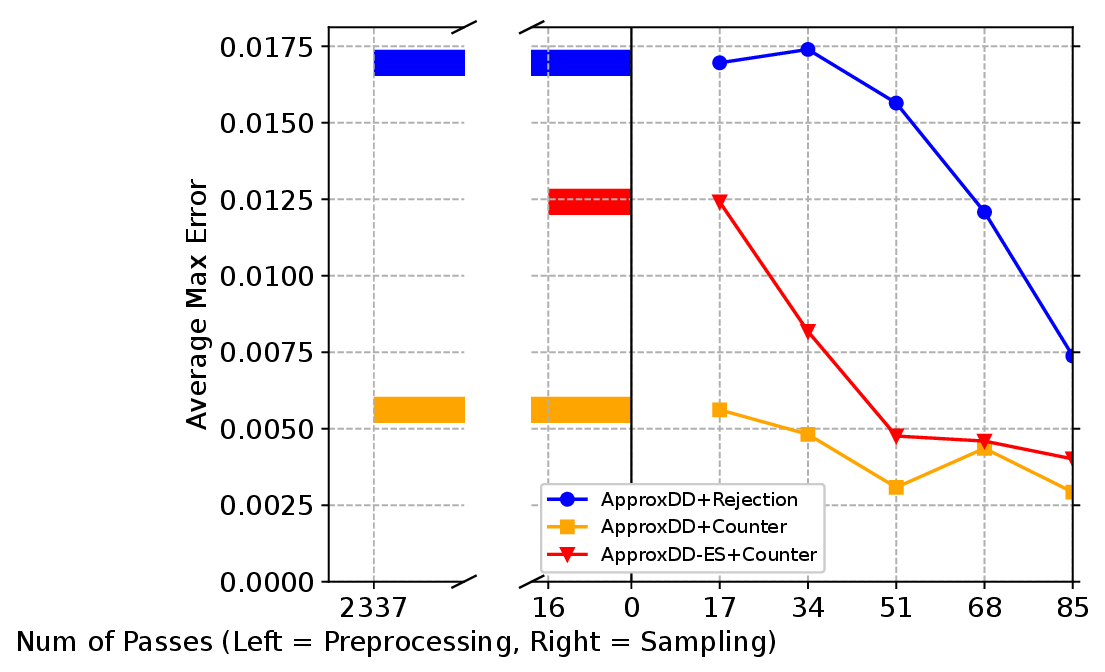}
    \caption{ER-3, $k=5$}
  \end{subfigure}
  \begin{subfigure}{0.32\textwidth}
    \includegraphics[width=\textwidth,trim=21 25 0 0,clip]{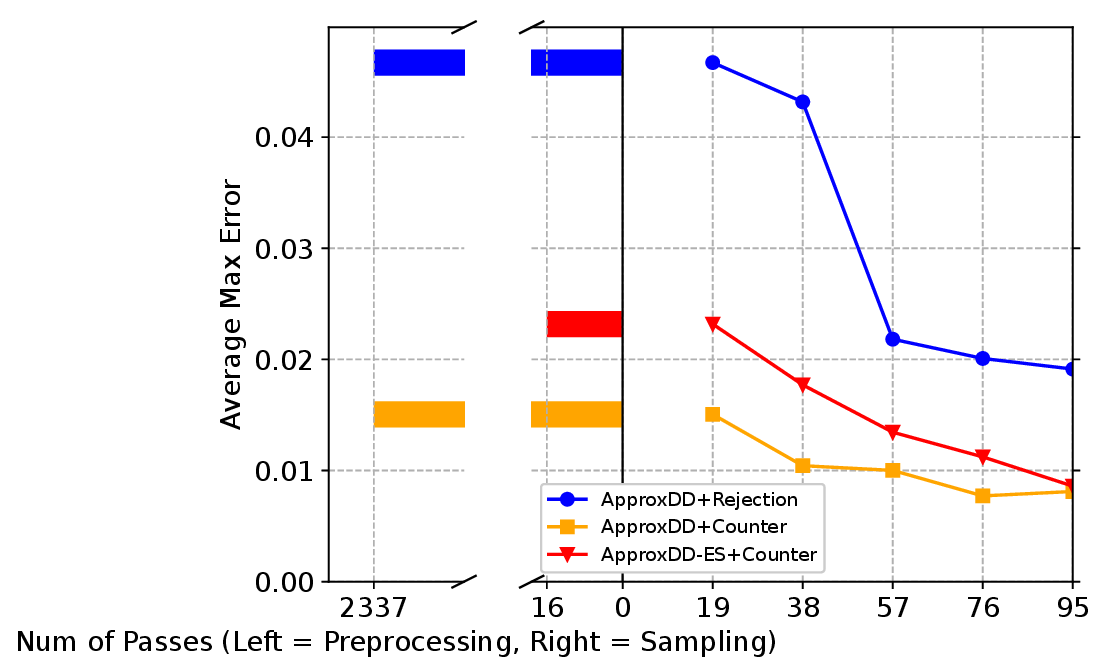}
    \caption{ER-3, $k=6$}
  \end{subfigure}
  \begin{subfigure}{0.32\textwidth}
    \includegraphics[width=\textwidth,trim=21 25 0 0,clip]{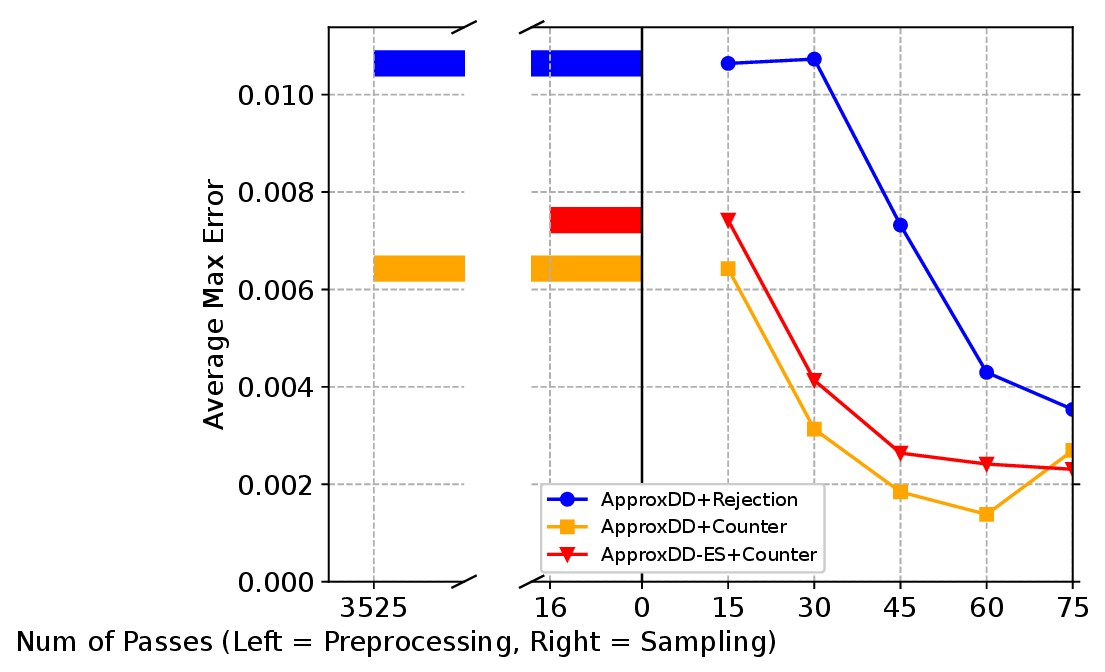}
    \caption{ER-5, $k=4$}
  \end{subfigure}
  \begin{subfigure}{0.32\textwidth}
    \includegraphics[width=\textwidth,trim=21 25 0 0,clip]{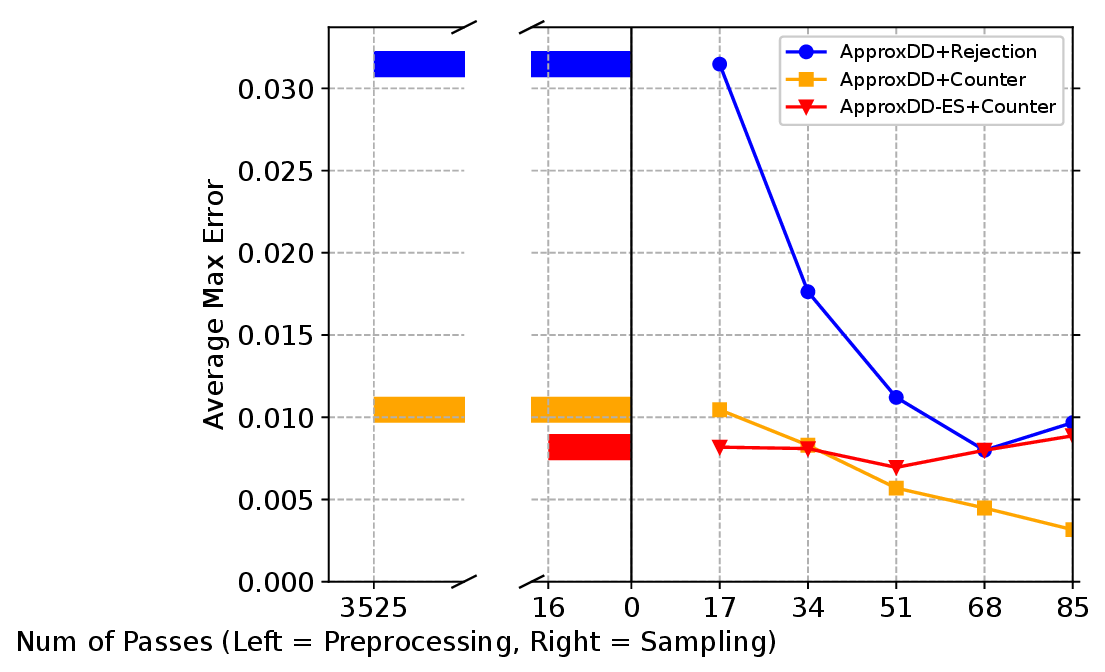}
    \caption{ER-5, $k=5$}
  \end{subfigure}
  \begin{subfigure}{0.32\textwidth}
    \includegraphics[width=\textwidth,trim=21 25 0 0,clip]{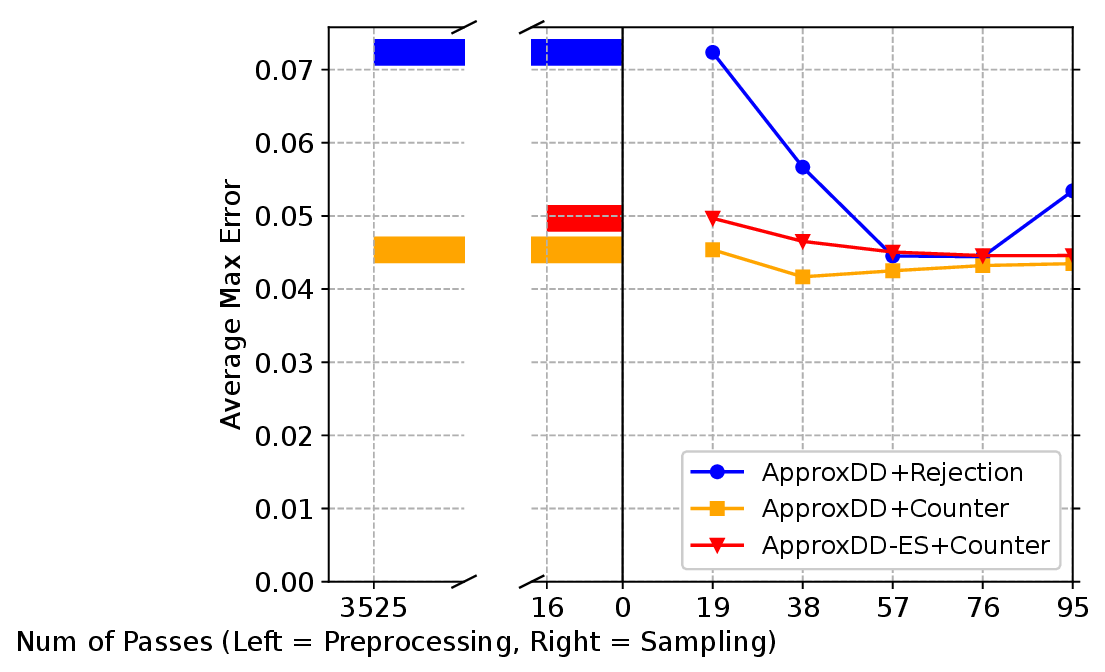}
    \caption{ER-5, $k=6$}
  \end{subfigure}
  \caption{(Continued) L$_{\infty}$ distance between the ground-truth and the estimated $k$-graphlet distribution as a function of the number of passes for sampling. The X-axis shows the number of passes (preprocessing on the left, sampling on the right). Missing points for \ApproxDD mean it did not terminate within 36 hours.}
\end{figure*} 

\end{document}